\documentclass[envcountsame,envcountsect]{llncs}

\usepackage{ifpdf}
\ifpdf
\usepackage{pdfsync}
\fi

\usepackage{amssymb}
\usepackage{amsfonts}
\usepackage{ifthen}
\usepackage{amsmath}
\usepackage{mathrsfs}
\usepackage{enumerate}
\usepackage{verbatim}
\usepackage{hyperref}

\hypersetup{
  colorlinks=false,
  pdfauthor={Alexander Kartzow, Pawel Parys},
  pdftitle={Pumping on CPG},
  pdfkeywords={collapsible pushdown graphs, pumping}
}

\newcommand{\typOrd}[0]{\sqsubseteq}
\newcommand{\type}[0]{\mathsf{type}}
\newcommand{\ctype}{\mathsf{ctype}}
\newcommand{\stype}{\mathsf{stype}}
\newcommand{\noteps}{{\!\not\,\varepsilon}}
\newcommand{\notepsBig}{{\,\not\!\varepsilon}}
\newcommand{\teA}{\sigma}
\newcommand{\teB}{\tau}
\newcommand{\teC}{\rho}
\newcommand{\deA}{\widehat\teA}
\newcommand{\deB}{\widehat\teB}
\newcommand{\deC}{\widehat\teC}
\newcommand{\comp}{\mathsf{comp}}

\newcommand{\sL}{n}

\newcommand{\N}{\ensuremath{\mathbb{N}}}

\newcommand{\Dd}{\mathcal{D}}

\newcommand{\Gg}{\mathcal{G}}
\newcommand{\Tt}{\mathcal{T}}
\newcommand{\Rr}{\mathcal{R}}
\newcommand{\Cc}{\mathcal{C}}
\newcommand{\Qq}{\mathcal{Q}}
\newcommand{\Pp}{\mathcal{P}}
\newcommand{\Nn}{\mathcal{N}}
\newcommand{\Ss}{\mathcal{S}}
\newcommand{\Xx}{\mathcal{X}}
\newcommand{\Vv}{\mathcal{V}}
\newcommand{\Yy}{\mathcal{Y}}

\newcommand{\Powerset}[0]{\mathcal{P}}

\newcommand{\push}[1]{\mathsf{push}^{#1}}
\newcommand{\pop}[1]{\mathsf{pop}^{#1}}
\newcommand{\col}[1]{\mathsf{col}^{#1}}
\newcommand{\TOP}[1]{\mathsf{top}^{#1}}

\newcommand{\hist}[2]{\mathsf{hist}({{#2}},{#1})}
\newcommand{\changelev}[0]{\mathsf{chl}}
\newcommand{\lev}[0]{\mathsf{lev}}
\newcommand{\NestingRank}[0]{\mathsf{nr}}
\newcommand{\posNew}[2]{\overset{#1}{\rightarrow}{#2}}   
\newcommand{\subrun}[3]{#1{\restriction}_{#2,#3}}

\newcommand{\Rule}[2]{{#1}\supseteq{#2}}

\newcommand{\lexOrd}{\preceq}
\newcommand{\lexOrdstrict}{\prec}
\newcommand{\lexOrdR}{\succeq}
\newcommand{\pack}{\mathsf{pack}}

\spnewtheorem{mydefinition}[theorem]{Definition}{\bfseries}{\rmfamily}
\spnewtheorem*{theorem*}{Theorem}{\bf}{\itshape}

\def\koniec{{}\hfill$\square$}
{\noindent{\itshape Proof%
   \ifthenelse{\equal{#1}{!*!,!}}{}{%
      \ (#1)}. }}%
{\koniec\endtrivlist}

\begin{document}
\title{Strictness of the Collapsible Pushdown Hierarchy\thanks{The first author is supported by the
  DFG  project ``GELO''.
  The second author is partially supported by the Polish Ministry of
  Science grant nr N N206 567840.
  The collaboration of the authors is supported by
  the ESF project ``Games for Design and Verification''.}}

\author{Alexander Kartzow\inst{1} \and Pawe\l\ Parys\inst{2}}
\institute{Universit\"at Leipzig\\Johannisgasse 26, 04103 Leipzig, Germany\\\email{kartzow@informatik.uni-leipzig.de}
\and University of Warsaw\\ul. Banacha 2, 02-097 Warszawa, Poland\\\email{parys@mimuw.edu.pl}}

\maketitle

\pagestyle{plain}

\begin{abstract}
  We present a pumping lemma for each level of the collapsible
  pushdown graph hierarchy in analogy to the second author's pumping
  lemma for higher-order pushdown graphs (without collapse). 
  Using this lemma, we give the first known examples that separate the
  levels of the collapsible pushdown graph hierarchy and of the
  collapsible pushdown tree hierarchy, i.e., the hierarchy of trees
  generated by higher-order recursion schemes. 
  This confirms the open conjecture that higher orders allow one to
  generate more graphs and more trees.
\end{abstract}

\section{Introduction}
\label{sec:intro}
Already in the 70's, Maslov (\cite{Mas74,Mas76}) generalised the concept
of pushdown systems to higher-order pushdown systems and studied
such devices as acceptors of string languages. In the last decade,
renewed interest in these systems has arisen. They are now studied
as generators of graphs and trees.
Knapik et al.\ \cite{easy-trees} showed that the class of trees generated by deterministic level $n$ pushdown systems coincides 
with the class of trees generated by \emph{safe} level $n$ recursion schemes,\footnote{
	Safety is a  syntactic restriction on the recursion scheme.}
and Caucal \cite{Caucal02} gave another characterisation: trees on level $n+1$ are obtained from trees on level $n$ 
by an MSO-interpretation  followed by an unfolding. 
Carayol and W\"ohrle \cite{cawo03} studied the
$\varepsilon$-contractions of configuration graphs of level $n$
pushdown systems  
and proved that these are exactly the graphs in the $n$-th level of
the Caucal hierarchy.

Driven by the question whether safety implies a semantical restriction
to recursion schemes, Hague et al.~\cite{collapsible} extended the
model of higher-order pushdown systems by
introducing a new stack operation called collapse. 
They showed that the trees generated by the resulting collapsible pushdown
systems coincide exactly with the class of trees generated by all
higher-order recursion schemes and this correspondence is
level-by-level. 
Recently, Parys (\cite{parys2011,parys-panic-new}) proved the safety
conjecture, i.e., he showed 
that higher-order recursion schemes generate
more trees than safe higher-order recursion schemes, which implies
that the class of collapsible pushdown trees is a proper extension of the
class of higher-order pushdown trees. Similarly, due to their different
behaviour with respect to MSO model checking, we know that the class of
collapsible pushdown graphs forms a proper extension of the class of
higher-order pushdown graphs. 

Several questions concerning the relationship of these
classes have been left open so far. Up to now it was not known whether
collapsible pushdown graphs form a strict hierarchy in the sense
that for each $n\in\N$ the class of level $n$ collapsible pushdown
graphs is strictly contained in the class of level $(n+1)$ collapsible
pushdown graphs.
The same question was open for the hierarchy of trees generated by collapsible pushdown systems (i.e.~by recursion schemes).
Extending the pumping arguments of Parys for higher-order pushdown
systems \cite{parys-pumping} to the collapsible pushdown setting, we
answer both questions in the affirmative. 

Our main technical contribution is the following \emph{pumping lemma} for
collapsible pushdown systems. It subsumes the known pumping lemmas  
for level $2$ collapsible pushdown systems \cite{kartzow-pumping} and
for higher-order pushdown systems \cite{parys-pumping}.  
Set $\exp_0(i)=i$ and $\exp_{k+1}(i) = 2^{\exp_k(i)}$.

\begin{theorem}\label{thm:pumping}
  Let $\Ss$ be a collapsible pushdown system of level $n$.
  Let $\Gg$ be the $\varepsilon$-contraction of the configuration
  graph of $\Ss$. Assume that it is finitely branching
  and that there is a path in $\Gg$ of length $m$ from the initial
  configuration to some configuration $c$.   
  For $C_\Ss$ a constant only depending on $\Ss$,
  if  there is a path $p$ in $\Gg$ of length at least 
  $\exp_{n-1}((m+1)\cdot C_\Ss)$ which starts in $c$, 
  then there are infinitely many paths in $\Gg$ which start in $c$ and
  end in configurations having the same control state as the last
  configuration of $p$.  
\end{theorem}

\begin{corollary}\label{cor:infinite}
  Let $\Gg$ be the successor tree induced by 
  $\{1^i0^{\exp_n(i)}\mid i\in\N \}$.

  $\Gg$ is the $\varepsilon$-contraction of the configuration graph of
  a pushdown system of level $n+1$  but not 
  the $\varepsilon$-contraction of the configuration
  graph of any collapsible pushdown system of level $n$. Moreover,
  $\Gg$ is generated by a safe level $(n+1)$ recursion scheme but not
  by any level $n$ recursion scheme. 
\end{corollary}
$\Gg$ is not in the $n$-th level because application of the pumping
lemma to the node
$1^{2 \cdot C_\mathcal{S}}0$ yields a contradiction. 
The proof that $\Gg$ is in level $n+1$ follows from \cite{blumensath-pumping}.
Beside this main result, our techniques allow us to decide the
following problems.\footnote{We thank several anonymous referees of
  our LICS submissions for pointing our interest towards these
  problems.}
\begin{lemma}\label{Lemma:Decidabilities}
Given a
collapsible pushdown system, it is decidable 
\begin{enumerate}
\item whether the
  $\varepsilon$-contraction of its configuration graph is finitely
  branching,
\item whether the $\varepsilon$-contraction of its configuration graph
  is finite, and
\item whether the unfolding of the $\varepsilon$-contraction of its
  configuration graph is finite.
\end{enumerate}
\end{lemma}

\subsection{Related Work}
Hayashi \cite{hayashi-pumping} and Gilman \cite{gilman-pumping}
proved a pumping and a 
shrinking lemma for indexed languages. It is shown in
\cite{AehligMO05} that
indexed 
languages are exactly the string languages accepted by level $2$
collapsible pushdown systems. For higher levels, no shrinking
techniques are known so far. 
Since our pumping lemma can be used only for finitely branching systems,
it cannot be used to show that certain string languages do not occur on certain
levels of the (collapsible) higher-order pushdown hierarchy.
Note that we do not know whether the string languages accepted by
nondeterministic level $n$ pushdown systems and by nondeterministic
level $n$ collapsible pushdown systems coincide for $n>2$. 
Thus, it
is an interesting open question whether there is a stronger pumping lemma for runs of higher-order systems that could be used to separate
these classes of string languages.


\section{Collapsible Pushdown Graphs}
\label{sec:preliminaries}
\label{sec:CPS}

Collapsible pushdown systems of level $\sL$ (from now on $\sL\in\N$ is
fixed) are an extension of
pushdown systems where we replace the stack by an  $\sL$-fold nested
stack structure. This higher-order stack is manipulated using a
a push, a pop and a collapse operation for each
stack level $1\leq i\leq \sL$. When a new symbol is pushed onto the
stack, 
we attach a copy of a certain level $k$ substack of the current stack
to this symbol (for some $1\leq k \leq n$) and at some later point the
collapse   
operation may replace the topmost level $k$ stack with the level $k$
stack stored in the topmost symbol of the stack (we also talk about
the linked $k$-stack of the topmost symbol). 
In some weak sense the collapse operation
allows one to jump back to the (level $k$) stack where the current topmost
symbol was created for the first time. 

\begin{definition}
  Given a number $n$ (the level of the system) and stack alphabet
  $\Gamma$, we define the set of stacks as the smallest set satisfying
  the following. 
  \begin{itemize}
  \item If $s_1, s_2,\dots, s_m$ are $(k-1)$-stacks, where $1\leq
    k\leq n$, then the sequence $[s_1, s_2,\dots, s_m]$ is a
    $k$-stack (with $s_m$ its topmost $k-1$-stack).  
    This includes the empty sequence ($m=0$).
  \item If $s^k$ is a $k$-stack, where $1\leq k\leq n$, and
    $\gamma\in\Gamma$, then $(\gamma,k,s^k)$ is a $0$-stack. 
  \end{itemize}
  
  For a $0$-stack $s^0=(\gamma, k, t^k)$ we call $\gamma$ the
  \emph{symbol} of $s^0$ and for some $k$-stack $t^k$ the 
  \emph{topmost symbol} is the symbol of its topmost $0$-stack. 

  For a $k$-stack $s^k$ and a $(k-1)$-stack $s^{k-1}$ we write
  $s^k:s^{k-1}$ to denote the $k$-stack obtained by appending
  $s^{k-1}$ on top of $s^k$. 
  We write $s^2:s^1:s^0$  for $s^2:(s^1:s^0)$.
\end{definition}

Let us remark that in the original definition stacks are defined
differently: they are not nested, a $0$-stack does not store the
linked $k$-stack but the number of pop-operations a collapse is
equivalent to. With respect to stacks constructible from the initial
stack by the stack operations, this is only a syntactical difference
as discussed in 
Appendix \ref{Appendix:Definitions}. Independently, Broadbent et al.\
recently also introduced our definition of stack under the name
\emph{annotated stacks} in \cite{BroCarHagSer12}.

\begin{definition}
  We define the set of stack operations $OP$ as follows.
  We decompose a stack $s$ of level $\sL$ into its topmost stacks
  $s^\sL:s^{\sL-1}:\dots:s^0$. 
  We have $\pop{i}(s):=s^\sL:\dots:s^{i+1}:s^i$ for all $1\leq i\leq
  \sL$. The result is undefined if $s^i$ is empty.  
  For $2\leq i\leq \sL$ we have
  $\push{i}(s):=s^\sL:\dots:s^{i+1}:(s^{i}:\dots:s^0):s^{i-1}:\dots:s^0$.
  The level $1$ push is $\push{1}_{\gamma,k}$ for $\gamma\in\Gamma$, $1\leq k \leq \sL$
  which is defined by
  $\push{1}_{\gamma,k}(s):=s^\sL:\dots:s^2:(s^1:s^0):(\gamma,k,s^{k})$.\footnote{In
    the following, we
    write $\push{1}$ whenever we mean some $\push{1}_{\gamma,k}$
    operation where the values of $\gamma$ and $k$ do not matter for the
    argument.}
  The collapse operation $\col i$ (where $1\leq i\leq \sL$) is defined
  if the topmost $0$-stack is 
  $(\gamma,i,t^{i})$, and $t^i$ is not empty.
  Then it is $\col{i}(s):= s^\sL:\dots :s^{i+1}:t^{i}$. 
  Otherwise the collapse operation is undefined.  
\end{definition}

\begin{definition}
  The \emph{initial $0$-stack} $\bot_0$ is $(\bot, n, [])$
  for a special symbol $\bot\in\Gamma$, i.e., a $0$-stack only
  containing the symbol $\bot$ with link to the empty stack. 
  The \emph{initial $(k+1)$-stack} is $[\bot_k]$.
  Some $\sL$-stack $s$ is a \emph{pushdown store} (or \emph{pds}), if
  there is a finite sequence of stack operations that create $s$
  from $\bot_\sL$. 
\end{definition}

\begin{remark}
  \label{rem:ColisPop}
  If $s$ is a pds and if $\col{j}(s)$ is
  defined, then there is a $k\geq 1$ such that $\col{j}(s)$ is the stack
  obtained from $s$ by applying $\pop{j}$ $k$ times. 
\end{remark}

\begin{definition}
  A \emph{collapsible pushdown system of level $\sL$} (an $\sL$-CPS)
  is a tuple $\Ss=(\Gamma,A,Q,q_I,\bot,\Delta)$ where $\Gamma$ is a
  finite stack
  alphabet containing the special symbol $\bot$, $A$ is a finite input
  alphabet, $Q$ is a finite 
  set of states, $q_I\in Q$ is an initial state, and $\Delta\subseteq
  Q\times\Gamma\times(A\cup\{\varepsilon\})\times Q\times OP$ is a
  transition relation.\\
  A \emph{configuration} is a pair $(q,s)$  with $q\in Q$ and $s$
  a pds. The \emph{initial configuration of $\Ss$} is
  $(q_I,\bot_\sL)$. 
\end{definition}

\begin{definition}
  We define a \emph{run} of a CPS $\Ss$. 
  For $0\leq i \leq m$, let $c_i=(q_i,s_i)$ be a configuration of $\Ss$ and let 
  $\gamma_i$ denote the topmost stack symbol of $s_i$. 
  A run $R$ of length $m$ from $c_0$ to $c_m$ is a sequence
  $c_0\vdash^{a_1}c_1\vdash^{a_2}\dots\vdash^{a_m}c_m$ 
  such that, for $1\leq i\leq m$, there is a transition
  $(q_{i-1},\gamma_{i-1},a_i,q_i,op)$ where $s_i = op(s_{i-1})$.
  We set $R(i):=c_i$ and call $\lvert R\rvert :=m$ the \emph{length of $R$}. 
  The \emph{subrun} $\subrun{R}{i}{j}$ is
  $c_i\vdash^{a_{i+1}}c_{i+1}\vdash^{a_{i+2}}\dots\vdash^{a_j}c_j$. 
  For runs $R,S$ with $R(\lvert R \rvert)=S(0)$, we write  $R\circ S$
  for the \emph{composition} of $R$ and $S$ which is 
  defined as expected. 
\end{definition}

\begin{definition}
  Let $\Ss$ be a collapsible pushdown system.
  The \emph{(collapsible pushdown) graph}\footnote{In fact it is an
    edge-labelled graph; sets $E_a$ need not to be disjoint.} 
  of $\Ss=(\Gamma,A,Q,q_I,\bot,\Delta)$ is 
  $\Gg:=(G, (E_a)_{a\in A\cup\{\varepsilon\}})$ where $G$ consists of
  all configurations 
  reachable from $(q_0, \bot_{\sL})$ and  
  there is an $a$-labelled edge from a configuration $c$ to a configuration $d$
  if there is a
  run $c\vdash^ad$.
  The \emph{$\varepsilon$-contraction} of $\Gg$ is the graph
  $(G', (E_a')_{a\in A})$ where 
  $G':=\{c\in G: \exists d\in G\ d\vdash^a c$ for some $a\in A\}$ and
  two configurations $c,d$ are 
  connected by $E_a'$ if there is a run 
  $c\vdash^\varepsilon c_1\vdash^\varepsilon \dots \vdash^\varepsilon
  c_n \vdash^a d$ for some $n\in\N$.
\end{definition}


\section{Proof Structure}
\label{sec:proofStructure}
The proof of the pumping lemma consists of three parts.
In the first part we introduce a special kind of context free
grammars (called well-formed grammars) for runs of a collapsible
pushdown system $\Ss$.  
In such a grammar, each nonterminal represents a set of runs and  
each terminal is one of the transitions of $\Ss$. 
Let $X$ and $X_1,\dots, X_m$ be sets of runs and $\delta$ some transition. 
A rule $\Rule{X}{\delta X_1X_2\dots X_m}$ describes a run $R$ if 
$R=S\circ T_1 \circ T_2\circ\dots\circ T_n$ such that $S$ performs 
only the transition $\delta$ and $T_i\in X_i$. A grammar describes a
family $\Xx$ of sets of runs if the rules for each $X\in\Xx$ describe
exactly the runs 
in $X$. Well-formed grammars are syntactically restricted in order to
obtain the following result.
If $\Xx$ is a finite family described by a well-formed grammar,
 we can define
\begin{enumerate}
\item a function $\ctype_\Xx$ from configurations of $\Ss$ to a
  finite partial order $(\Tt_\Ss,\typOrd)$ (of 
  \emph{types of configurations}), and 
\item for each $X\in \Xx$ a level $\lev(X)\in\{0,1,\dots, \sL\}$
\end{enumerate}
such that the following
transfer property of runs holds.

\begin{theorem}\label{thm:types}
  Let $\Xx$ be a family of sets of runs
  described by a  well-formed grammar,
  $R\in X\in\Xx$, and $c$ be a configuration with
  $\ctype_\Xx(R(0))\typOrd \ctype_\Xx(c)$.
  \begin{enumerate}
  \item  \label{thm:typesA} There is a run $S\in X$ starting in $c$
    which has the same final state as $R$ and
  \item \label{thm:typesB}
    if $\lev(X)=0$, then 
    $\ctype_\Xx(R(\lvert R \rvert)) \typOrd\ctype_\Xx(S(\lvert S \rvert))$.
  \end{enumerate}
\end{theorem}

The idea behind the definition of $\ctype_\Xx$ is that we assign a
type not only to the whole configuration, but also to every $k$-stack
(for every $k$). 
This type summarises possible behaviours
of the $k$-stack in dependence on the type of the $n$-stack
below this $k$-stack. 
This makes types compositive: the type of a stack $s^{k+1}:s^k$ is
determined by the type of $s^{k+1}$ and of $s^k$.
The above theorem generalises results of \cite{parys-pumping} in two ways: 
first, it works for collapsible systems; second, it works for
arbitrary well-formed grammars instead of a fixed family of sets of
runs. 
The corresponding part of the proof from \cite{parys-pumping} is
not transferable to collapsible systems at all.
For collapsible systems we even need a new
definition of types (see Appendix \ref{sec:types}). 
We stress that the new definition of types relies on  the different form of
representing links  in stacks:
our $k$-stack already contains all linked stacks, so  we can
summarise it using a type from a finite set.  
On the other hand the original $k$-stack has arbitrarily many numbers
pointing to stacks ``outside'', and 
we could not define a type from a finite set because the behaviour of a
$k$-stack would depend on this unbounded context ``outside''. 

In the second part of the proof (cf.\ Section \ref{sec:Runs}), we
introduce a well-formed grammar 
for a certain family $\Xx$. As a main feature, $\Xx$  contains the
set of so-called \emph{pumping runs} $\Pp$.
In the grammar describing $\Xx$, the level of $\Pp$ is $0$ whence the
strong version of Theorem \ref{thm:types} applies. If a pumping run
$R$ starts and ends in configurations of the same type, this theorem
then allows to pump this run, i.e., basically we can append a copy of
this run to its end and iterating this process we obtain infinitely
many pumping runs.

The last part of the proof uses Theorem \ref{thm:types} for the above family $\Xx$ to deduce the pumping lemma.
This part follows closely the analogous proof for the non-collapsible
pushdown systems in \cite{parys-pumping} (see Appendices
\ref{app:sketchPumping}--\ref{app:PumpingLemma}): we prove
that a long run contains a pumping run such that the application of
Theorem \ref{thm:types} yields a configuration $c$ on this path such
that either the graph is infinitely branching at $c$ or the pumped
runs yield longer and longer paths in the $\varepsilon$-contraction of
the pushdown graph.


\section{Run Grammars}
\label{sec:grammars}
Let $\Xx$ be a finite family whose elements are sets of runs of $\Ss$.
We want to describe this family using a kind of context free grammar.
In this grammar the members of $\Xx$  appear as nonterminals and
the transitions of $\Ss$ play the role of terminals.

We assume that there is a partition $\Xx=\bigcup_{i=0}^\sL \Xx_i$ into
pairwise distinct families of sets of runs. 
For each set $X\in\Xx$, we define its level to be
$\lev(X):=i$ if $X\in \Xx_i$. 
We only consider \emph{well-formed grammars} that satisfy the restriction
that all rules of the grammar have to be \emph{well-formed}. 

\begin{definition}\label{def:wf-rule}
  A \emph{well-formed rule over $\Xx$} (\emph{wf-rule} for short)  is
  of the form  
  \begin{enumerate}
  \item	$\Rule{X}{}$  where $X\in\Xx$, or
  \item	$\Rule{X}{\delta}$ where $\delta\in\Delta$, $X\in\Xx$
    and if the operation in $\delta$ is $\mathsf{pop}^k$ or
    $\col k$ then $k\leq \lev(X)$, or 
  \item	$\Rule{X}{\delta Y}$ where $\delta\in\Delta$,
    $X,Y\in\Xx$, $\lev(Y)\leq \lev(X)$
    and if the operation in $\delta$ is $\mathsf{pop}^k$ or
    $\col k$ then $k\leq \lev(Y)$, or 
  \item	$\Rule{X}{\delta YZ}$ where $\delta\in\Delta$,
    $X,Y,Z\in\Xx$, $\lev(Z)\leq \lev(X)$,
    if the operation in $\delta$ is $\mathsf{pop}^k$ or
    $\col k$ then $k\leq \lev(Y)$, and 
    whenever $R$ is a composition of a one-step run performing
    transition $\delta$ with a run from $Y$, then the topmost
    $\lev(Y)$-stacks of $R(0)$ and $R(\lvert R \rvert)$ coincide.
  \end{enumerate}
\end{definition}
\begin{definition}\label{def:X-described}
  We say that a run $R$ is \emph{described} by a wf-rule
  $\Rule{X}{\delta X_1 \dots X_m}$, $m\in\{0,1,2\}$
  if there is a decomposition
  $R=R_0 \circ R_1 \circ \dots \circ R_m$ such that 
  $R_0$ has length $1$ and performs $\delta$ and $R_i\in X_i$ for each
  $1\leq i \leq m$; 
  a run $R$ is described by $\Rule{X}{}$ if $\lvert R\rvert=0$.
  We say that a family $\Xx$ is \emph{described} by a well-formed
  grammar $\Rr_\Xx$ if for each $X\in\Xx$, a run $R$ is in $X$ if and
  only if it is described by some rule $\Rule{X}{\delta X_1 \dots
    X_m}\in \Rr_\Xx$.
\end{definition}

\begin{example} \label{exa:Rules}
  Let $\Qq$ be the set of all runs.
  Setting 
  $\lev(\Qq)=\sL$, the one-element family $\{\Qq\}$  is described by
  the wf-rules
  $\Rule{\Qq}{\delta \Qq}$ for each transition $\delta$, and
  $\Rule{\Qq}{}$.

  Indeed, for every run $R$ either $\lvert R\rvert=0$  or 
  $R$ consists of a first transition
  followed by some run.
  Note that we cannot choose $\lev(\Qq)$ different from $\sL$ whenever
  $\Ss$ contains a transition $\delta_0$ performing $\col{\sL}$ or
  $\pop{\sL}$. If we set $\lev(\Qq)<\sL$, then 
  $\Rule{\Qq}{\delta_0 \Qq}$ would not be a wf-rule. 
\end{example}

Next we prove that the class of families described by well-formed
grammars is closed under addition of unions and compositions. 
This is crucial for the decidability results mentioned in 
Lemma \ref{Lemma:Decidabilities}.
If $X$ and $Y$ are sets of runs,
we set $X\circ Y:=\{R\circ S: R\in X, S\in Y\}$.

\begin{lemma}\label{lem:compos}
  Let $\Xx$ be a family described by a well-formed grammar. 
  For $X,Y\in \Xx$ the family
  $\Xx\cup\{ X\cup Y\}$ is described by a well-formed grammar. 
  Moreover, there is a family $\Yy\supseteq \Xx\cup\{X\circ Y\}$ that is
  described by a well-formed grammar. In these grammars, we have 
  $\lev(X\cup Y) = \lev(X\circ Y) = \max(\lev(X),\lev(Y))$.
\end{lemma}
\begin{proof}
  For each rule $\Rule{Z}{\delta Z_1 \dots Z_m}$ with $Z\in\{X,Y\}$
  adding the rule 
  $\Rule{(X\cup Y)}{\delta Z_1 \dots Z_m}$ settles the case of unions.
  
  For the composition, we add a set $Z\circ Y$ for each  $Z\in\Xx$, and a new set $Y^{i}$ for $0\leq i \leq n$. 
  $Y^{i}$ contains exactly the same runs as $Y$, 
  but we set $\lev(Y^i):= \max(\lev(Y), i)$. Wf-rules
  describing $Y^i$ are clearly obtained from the rules for $Y$ by
  replacing 
  the left-hand side by $Y^i$. Note that increasing the level of the
  left-hand size turns well-formed rules into well-formed rules.
  Rules for each of the $Z\circ Y$ are easily obtained from rules
  for  $Z$ as follows.
  \begin{itemize}
  \item If there is a rule $\Rule{Z}{}$, for each rule having $Y$ on
    the left side we add the same rule with $Z\circ Y$ on the left
    side, 
  \item for each rule $\Rule{Z}{\delta}$ we add a rule 
    $\Rule{(Z\circ Y)}{\delta Y^{\lev(Z)}}$, 
  \item for each rule $\Rule{Z}{\delta X_1}$ we add a rule
    $\Rule{(Z\circ Y)}{\delta (X_1\circ Y)}$, 
  \item for each rule $\Rule{Z}{\delta X_1X_2}$ we add a rule
    $\Rule{(Z\circ Y)}{\delta X_1(X_2\circ Y)}$. 
  \end{itemize}
  It is straightforward to check that this is a well-formed grammar
  describing the family $\Yy:=\Xx\cup\{ Z\circ Y:Z\in\Xx\}\cup\{
  Y^{i}: 0\leq i \leq n\}$.\qed
\end{proof}


\section{A Family of Runs}
\label{sec:Runs}
We now define a family $\Xx$ described by
a well-formed grammar.
We first name the sets of runs that we define in the following. 
Some of our 
classes have subscripts from $\varepsilon$, $\notepsBig$,
$=$, and $<$. 
Subscript $\varepsilon$ marks a set if all runs in the set only
perform $\varepsilon$-transitions, while $\notepsBig$ marks a set if
each run in the set performs at least one non-$\varepsilon$-transitions.
Subscript $<$ marks sets (of pumping runs) where each run starts in a
smaller stack than it ends, 
while $=$ marks sets where no run starts in a smaller stack than it
ends (it follows that each such pumping run ends in the same stack as
it starts). 
$\Xx$ consists of the following sets (which we describe in detail on
the following pages).
\begin{itemize}
\item $\Qq$ of all runs,
\item $\Nn_k$ of $\TOP{k}$-non-erasing runs,
\item $\Pp$ of pumping runs which is the disjoint union of 
  the sets $\Pp_{x,y}$ for $x\in\{<,=\}$,  $y\in\{\varepsilon,
  \notepsBig\}$. Additionally, we set
  $\Pp_\noteps=\Pp_{<,\noteps}\cup\Pp_{=,\noteps}$ and
  $\Pp_\varepsilon=\Pp_{<,\varepsilon}\cup\Pp_{=,\varepsilon}$. 
\item $\Rr_{k,j}$ of $k$-returns of change level $j\geq k$ which is the
  disjoint union of the sets $\Rr_{k,j,y}$ for
  $y\in\{\varepsilon, \notepsBig\}$, and
\item $\Cc_{k,j}$ of $k$-colreturns of change level $j\geq k$
  which is the 
  disjoint union of the sets $\Cc_{k,j,y}$ for
  $y\in\{\varepsilon, \notepsBig\}$.
\end{itemize}

In order to easily distinguish between $\varepsilon$-runs and
$\notepsBig$-runs in the rules, we partition the transition relation
$\Delta=\Delta_\varepsilon \cup \Delta_\noteps$ such that
$\Delta_\varepsilon$ contains exactly the $\varepsilon$-labelled
transitions.
Before we can give rules for the family we need to define the levels
of its sets. 
We set $\lev(\Qq) = n, \lev(\Rr_{k,j,y}) = k, \lev(\Cc_{k,j,y}) = k,
\lev(\Nn_k) = n$ and $\lev(\Pp_{x,y}) = 0$. 

Now we give rules for these sets and we describe the
main properties of runs in each of the sets. 
In  Appendix \ref{app:Characterisation} we prove that these
descriptions are correct.

Recall that we have described $\Qq$ by well-formed rules in 
Example \ref{exa:Rules}. The sets of returns and colreturns are
auxiliary sets. Returns occur in the wf-rules for $\Nn_k$ and
$\Pp_{x,y}$ while colreturns are necessary to give wf-rules for
returns. 

$\Nn_k$ contains all runs $R$ where the topmost $k$-stack of
$R(0)$ is never removed during the run. 
First, we give an idea how the set $\Nn_0$ plays an important role in
our pumping lemma. Recall that we want to apply Theorem
\ref{thm:types} to pumping runs in order to obtain arbitrarily many runs
starting in a given configuration. Our final goal is to construct
infinitely many different paths in the $\varepsilon$-contraction of
the graph of a given collapsible pushdown system that all end in a
specific state $q$. But in general, the pumping runs we construct end in a
different state. Thus, the type of the stack reached by each of the
pumping runs should determine that we can reach a configuration with
state $q$ from this position. This could be done using the set $\Qq$
but it is not enough: if the pumping runs induce
$\varepsilon$-labelled paths then we could append runs from $\Qq$ that
all lead to the same configuration. In this case, we construct longer
and longer runs but all these runs encode the same edge in the
$\varepsilon$-contraction. This is prohibited by the use of runs from
$\Nn_0$: we can prove that the longer pumping runs we construct end in
larger stacks. Appending a run from $\Nn_0$ to such a run ensures that
the resulting run also ends in a large stack. From this observation we
will obtain infinitely many runs that end in different configurations
with state $q$. Thus, they induce infinitely many paths in the
$\varepsilon$-contraction. 
The rules for $\Nn_k$ are
\begin{itemize}
\item $\Rule{\Nn_k}{}$,
\item $\Rule{\Nn_k}{\delta \Nn_k}$ for each 
  $\delta\in\Delta$
  performing an operation of level at most $k$,
\item $\Rule{\Nn_k}{\delta^j \Nn_{j-1}}$ for each
  $\delta^j\in\Delta$ performing a
  $\push{j}$ and $j\geq k+1$,
\item $\Rule{\Nn_k}{\delta^j \Rr_{j,j} \Nn_k}$ 
  for each
  $\delta^j\in\Delta$  performing a
  $\push{j}$.
\end{itemize}
Our analysis of returns 
reveals that $\delta^j$ followed by 
a run from $\Rr_{j,j}$ starts and ends in the same stack. Thus,
the last rule satisfies the requirement that the topmost
$j$-stacks of these stacks coincide. Moreover, such a run
never changes the topmost $j$-stack of the initial configuration. 
Using this fact it is straightforward to see that every run described
by these rules does not remove the topmost $k$-stack. The other
direction, i.e., the proof that every run preserving the existence
of the topmost $k$-stack is described by one of these rules, can be
found in the appendix.

Some run $R$ is a pumping run, i.e., $R\in\Pp$,  if its final stack is
created completely 
on top of its initial stack in the following sense: 
the topmost $1$-stack of $R(\lvert R\rvert)$ is obtained as a
(possibly modified) 
copy of the topmost $1$-stack of $R(0)$, 
and in this copy the topmost $0$-stack of $R(0)$ was never removed.
Another view on this definition is as follows: for each $k$, the run
$R$ may look into a copy of the topmost $k$-stack of $R(0)$ only if
this copy is not directly involved in the creation of the topmost
$k$-stack of $R(\lvert R \rvert)$. 
In Appendix \ref{app:History} we define a history function that makes
the notion of being involved in the creation of some stack
precise: for each $i<\lvert R \rvert$ and for each $k$-stack $s^k$ of
$R(\lvert R \rvert)$ we can identify a $k$-stack $t^k$ in $R(i)$ which
is the maximal $k$-stack involved in the creation of this stack. 

In the rest of this section  $y, y_0, y_1, y_2$ are
variables in
$\{\varepsilon, \notepsBig\}$ where we assume that either all are
$\varepsilon$ or $y=\notepsBig$ and one of the $y_i$ occurring in the
rule is $\notepsBig$. 
$j$ is a variable ranging over $\{1, 2, \dots, n\}$.
The rules for $\Pp$ are
\begin{itemize}
\item $\Rule{\Pp_{=,\varepsilon}}{}$,
\item $\Rule{\Pp_{<,y}}{\delta_{y_0} \Pp_{x,y_1}}$ for each
  $\delta_{y_0}\in\Delta_{y_0}$ performing $\push{j}$ and
  $x\in\{=,<\}$,
\item $\Rule{\Pp_{x,y}}{\delta^j_{y_0} \Rr_{j,j,y_1} \Pp_{x,y}}$ for
  each $\delta^j_{y_0}\in\Delta_{y_0}$ performing $\push{j}$ and
  $x\in\{=,<\}$, 
\item 
  $\Rule{\Pp_{<,y}}{\delta^j_{y_0} \Rr_{j,j',y_1} \Pp_{x,y_2}}$ for
  each $\delta^j_{y_0}\in\Delta_{y_0}$ performing $\push{j}$, $j'>j$
  and $x\in\{=,<\}$. 
\end{itemize}
Proving the correctness of this set of rules with respect to our
intended meaning of the sets $\Pp_{x,y}$ requires a very detailed
study of returns which is done in Appendix
\ref{app:Characterisation}. 
In order to see that the rules of the last two kinds are
well-formed  we need the property that for every run $R$ which first
performs a $\push j$ operation followed by a
$j$-return, 
the topmost $j$-stack of $R(0)$ and $R(|R|)$ is the same.

\begin{example}
  A run of length $1$ performing a $\push 1$ operation is a pumping run.
  Also a run of length $2$ performing a $\push 1$ operation followed
  by a $\pop 1$ operation is a pumping run. 
  However a run of length $2$ performing first a $\pop 1$ operation
  and then a $\push 1$ operation is not a pumping run. 
  This shows that in the definition of a pumping run we do not only care
  about the initial and final configuration, but about the way the
  final configuration is created by the run: 
  a pumping run $R$ may never remove the topmost $0$-stack of $R(0)$. 

  Next consider a run $R$ of length $3$ performing the sequence of
  operations 
  $\push2,\ \pop1,\ \pop 2.$
  It is also a pumping run.
  Notice that this run ``looks'' into a copy of the topmost $1$-stack
  of  $R(0)$, i.e., it removes its topmost $0$-stack whence it depends on 
  symbols of $R(0)$ other than the topmost one. 
  One can see that in any 2-CPS, whenever a pumping run $R$
  looks into a copy of the topmost $1$-stack of $R(0)$, then this copy
  is completely removed from the stack at some later point in the
  run. 
  However, this is not true for higher levels. A counter example is
  a run performing 
  $\push2,\ \pop1,\ \push3,\ \pop 2.$
\end{example}

Next we define returns.
A run $R$ is a $k$-return (where $1\leq k\leq n$) if
\begin{itemize}
\item the topmost $(k-1)$-stack of $R(|R|)$ is obtained as a copy of
  the second topmost $(k-1)$-stack of $R(0)$ (in particular we require
  that there are at least two $(k-1)$-stacks in the topmost
  $k$-stack of $R(0)$), and  
\item while tracing this copy of the second topmost $(k-1)$-stack of
  $R(0)$ which finally becomes the topmost $(k-1)$-stack of $R(|R|)$, 
  it is not the topmost $(k-1)$-stack of $R(i)$ for any $i<|R|$.
\end{itemize}
Additionally, for a $k$-return $R$ its change level is the maximal $j$
such that the topmost $j$-stack of the initial and of the final stack
of $R$ differ in size (i.e.~in the number of $(j-1)$-stacks they
contain).\footnote{ 
  One can see that it is the same as saying that the topmost $j$-stack
  of the initial and of
  the final stack of $R$ differ. However a definition using size is
  more convenient.} 
One can see that the topmost $k$-stack of $R(0)$ is always greater by
one than the topmost $k$-stack of $R(|R|)$, so we have $j\geq k$. 
Recall that $\Rr_{k,j}$ is the set of $k$-returns of change level $j$.

Let us just give some intuition about returns before we state their
exact characterisation using wf-rules. 
The easiest sets of returns are those where $k=j$. 
A run $R\in\Rr_{k,k}$ starts in some stack $s$, ends in the
stack $\pop{k}(s)$, and never visits $\pop{k}(s)$ (or any smaller
stack) before the final configuration. 
Notice also that there is a minor restriction on the use of collapse
operations: $R$ is not allowed to use
links of level $k$ stored in $s$ in order to reach $\pop{k}(s)$.
Indeed, if such a link was used, then the topmost $(k-1)$-stack of
$R(\lvert R\rvert)$ would not be a copy of the second topmost
$(k-1)$-stack of $R(0)$, but a copy of the $(k-1)$-stack stored in the
link used. Note that this distinction is due to our special
representation of links,  yet it is useful for the understanding
of the definitions.

In the case that $j>k$ things are more complicated but similar. 
This time $R\in\Rr_{k,j}$ makes a number of copies of the
(possibly modified) topmost $(j-1)$-stack of the initial stack $s$
whence the topmost 
$j$-stack of the final stack $s'$ is of bigger size than the topmost
$j$-stack of $s$. 
But again the topmost $k$-stack of $s'$ is the same as the topmost
$k$-stack of $\pop{k}(s)$, 
and is in fact created as a modified copy of the topmost $k$-stack of
$s$. 
Furthermore, while tracing the history of this copy along the
configurations of the run, the 
size of this copy is always greater than its size in $R(|R|)$.
Notice however that we may also create some other copies of the
topmost $k$-stack of $s$, in which we can remove arbitrarily many
$(k-1)$-stacks. 
Finally, there is again a minor restriction on the use of collapse
links stored in the initial stack $s$. This restriction
implies that the stack obtained via application of the stack
operations of the return to $s$ is independent of the linked stacks,
i.e., if we replace one of the links of the stack $s$ 
such that the 
stack operations of $R$ are still applicable to the resulting stack
$s'$, then this sequence of stack operations applied to $s'$ results in the same stack (as when applied to $s$). 
In Example \ref{exa:colreturnsinReturns} we discuss the conditions
under which we may use a link stored in the initial stack of some
return.

\begin{example}
Consider a run $R$ of length $6$ (of a collapsible pushdown system of
level $2$) which 
performs the sequence of operations
\begin{equation*}
  \push2, \pop1, \pop2, \pop1, \push1, \pop1.
\end{equation*}
Below we use the notation that symbols taken in square brackets are in
one $1$-stack (we omit the collapse links). 
Assume we start from a stack $[aa][aa]$.
The stacks of the following configurations of $R$ are:
\begin{align*}
  &R(0)=[aa][aa][aa],\
  R(1)=[aa][aa][a],\
  R(2)=[aa][aa],\\
  &R(3)=[aa][a],\
  R(4)=[aa][aa],\
  R(5)=[aa][a].  
\end{align*}
We have $\subrun{R}{0}{2}\in\Rr_{1,2}$,
$\subrun{R}{0}{4},\subrun{R}{1}{2},\subrun{R}{3}{4},\subrun{R}{5}{6}\in\Rr_{1,1}$ 
and 
$\subrun{R}{1}{3},\subrun{R}{2}{3}\in\Rr_{2,2}$.
These are the only subruns of $R$ being returns, in particular $R$ is
not a $1$-return because it visits its final stack before its final
configuration. 
\end{example}

\begin{example}
The run $R$ of length $5$ performing the sequence of
operations 
\begin{equation*}
  \push2, \pop1, \push3, \pop2, \pop1
\end{equation*}
is a $1$-return of change level $3$.
Notice that  the final stack contains a copy of the topmost $1$-stack
of $R(0)$ with its topmost $0$-stack removed. 
\end{example}

The rules for returns are as follows:
\begin{itemize}
\item $\Rule{\Rr_{k,k,y}}{\delta_{y_0}}$ for each
  $\delta_{y_0}\in\Delta_{y_0}$ performing $\pop{k}$,
\item $\Rule{\Rr_{k,j,y}}{\delta_{y_0} \Rr_{k,j,y_1} }$ for each
  $\delta_{y_0}\in\Delta_{y_0}$ performing an operation of level 
  $<k$,
\item $\Rule{\Rr_{k,j,y}}{\delta^{j_0}_{y_0} \Rr_{k,j_1,y_1}}$ for each
  $\delta^{j_0}_{y_0}\in\Delta_{y_0}$ performing a $\push{j_0}$ such
  that $j_0>k$ and $\max\{j_0, j_1\} = j$,
\item $\Rule{\Rr_{k,j,y}}{\delta^{j_0}_{y_0} 
    \Rr_{j_0,j_0,y_1} \Rr_{k,j,y_2}}$ 
  for each
  $\delta^{j_0}_{y_0}\in\Delta_{y_0}$ performing a push of level
  $j_0$,
\item $\Rule{\Rr_{k,j,y}}{\delta^{j_0}_{y_0} 
    \Rr_{j_0,j_1,y_1} \Rr_{k,j_2,y_2}}$ 
  for each
  $\delta^{j_0}_{y_0}\in\Delta_{y_0}$ performing a $\push{j_0}$ such
  that $j_1>j_0$ and $\max\{j_1,j_2\}=j$, and
\item $\Rule{\Rr_{k,j,y}}{\delta_{y_0} \Cc_{k,j,y_1}}$ for each
  $\delta_{y_0}\in\Delta_{y_0}$ performing a $\push{1}_{a,k}$.
\end{itemize}

A $k$-colreturn is a run $R$ that performs in
the last step a $\col{k}$ on a copy of the topmost symbol of its
initial stack.
The change level of $k$-colreturns is (again) defined 
as the maximal $j$ such that the topmost $j$-stack of the initial
and of the final stack of the colreturn $R$ differ in size. 

Note that $k$-colreturns appear in the rules for returns after a push
of level $1$.  
The simplest example of a return described by the 
last rule is a run $R$ starting in a stack $s$ and performing
$\push{1}_{a,k}$ and then $\col{k}$. Note that such a sequence has the
same effect as applying 
$\pop{k}$ to $s$. Note that $\subrun{R}{1}{2}$ in this example is a
run from the stack $s':=\push{1}_{a,k}(s)$ to $\pop{k}(s')$ (for $k\geq
2$). Nevertheless we exclude it from the definition of a $k$-return of
change level $k$ because this effect is not transferable to arbitrary
other stacks: of course, we can apply the transition of
$\subrun{R}{1}{2}$ to the stack $\push{k}(s')$ and obtain a run $R'$
from $\push{k}(s')$ to $\pop{k}(s')$. But apparently this is not a run
from some stack $s''$ to a stack $\pop{k}(s'')$, so it is not a $k$-return. 
For this reason our definition of returns 
disallows the application of certain stored collapse links. 
The colreturns take care of such situations where
we use the links stored in the stack. Notice that $k$-colreturns occur in the
rules defining the other sets of runs only at those points where we
performed a push of level $1$ whence we
can be sure that the effect of the collapse operation coincides with
the application of exactly one $\pop{k}$ operation to the initial
stack. 

\begin{example} \label{exa:colreturnsinReturns}
Consider a run $R$ of length $4$ performing  
$\push2, \col1, \pop2, \pop1$.
It is a $1$-return of change level $1$.
Notice that it performs a collapse operation using a (copy of a) link
already 
stored in $R(0)$. 
But $\subrun{R}{1}{3}$ is a $2$-return (of change level $2$) which
covers this collapse operation, i.e., whenever the whole
sequence is  applicable to some stack $s$ it ends
in the stack $\pop1(s)$. As a general rule, we allow the use of a
$\col{k}$ from a (copy) of a link stored in the initial stack of some
return $R$ if it occurs within some subrun $R'$ that is a $k'$-return or
$k'$-colreturn of higher level(i.e., $k'>k$). In such cases the
resulting stack does not depend on the stack stored in the link (as
long as the whole sequence of operations of the return is
applicable). 
Hence, the following sequence of operations also induces a $1$-return
of change level $1$:
$\push3, \col2, \pop3, \pop1$.
\end{example}

Finally, let us state the rules for $k$-colreturns. 
\begin{itemize}
\item $\Rule{\Cc_{k,k,y}}{\delta_{y_0}}$ for each
  $\delta_{y_0}\in\Delta_{y_0}$ performing a $\col{k}$,
\item $\Rule{\Cc_{k,j,y}}{\delta^{j_0}_{y_0} \Cc_{k,j_1,y_1}}$ for
  each
  $\delta^{j_0}_{y_0}\in\Delta_{y_0}$ performing a $\push{j_0}$ such that $j_0\geq 2$ and $\max\{j_0,j_1\}=j$,
\item $\Rule{\Cc_{k,j,y}}{\delta^{j_0}_{y_0} 
    \Rr_{j_0,j_0,y_1} \Cc_{k,j,y_2}}$ 
  for each
  $\delta^{j_0}_{y_0}\in\Delta_{y_0}$ performing a $\push{j_0}$, and
\item $\Rule{\Cc_{k,j,y}}{\delta^{j_0}_{y_0} 
    \Rr_{j_0,j_1,y_1} \Cc_{k,j_2,y_2}}$ 
  for each
  $\delta^{j_0}_{y_0}\in\Delta_{y_0}$ performing a $\push{j_0}$ such that $j_1>j_0$ and $\max\{j_1,j_2\}=j$. 
\end{itemize}
This completes the presentation of the well-formed grammar describing
$\Xx$.


\bibliographystyle{abbrv}
\bibliography{bib}
\newpage
\appendix

\section{Deviations from Standard Definitions}
\label{Appendix:Definitions}
  Our definition of a collapsible pushdown system of level $n$
  deviates from the original one from \cite{collapsible} in several
  respects. 
  \begin{enumerate}
  \item Instead of storing links which control the collapse operation,
    we store the content of the stack where the link points to. 
    Note that this is only a syntactical difference.
  \item Instead of having one collapse operation whose level is
    controlled by the topmost element of the stack, we have one
    collapse operation for each stack level. 
    Note that we can simulate a collapse transition in the original
    sense by using $n$ collapse transitions (one for each level).
  \item Finally, for reasons of uniformity, the
    definition of the $\push{1}_{a,1}$ operation differs from the
    original one. 
    If we apply this operation to a stack $s$ we do not create a link
    to the topmost $1$-stack $s^1$ of $s$ but to $\pop{1}(s^1)$. 
    Thus, the effect of using a collapse of level $1$ in the original
    definition is always  equal  to the effect of applying $\pop{1}$ while
    in our definition it is always equal to the effect of applying
    $\pop{1}$ \emph{twice}. Nevertheless, we can simulate every system
    of the original definition by simulating collapse of level one
    with a $\pop{1}$ operation directly.\footnote{
      Notice however that $\col 1$ can be performed only if the
      topmost $0$-stack stores a $1$-stack. 
      The same should be true for the new transition performing $\pop 1$.
      To ensure this, it is enough to extend the stack alphabet so
      that the stack symbol stores also the level of the stack stored
      in the $0$-stack.} 
  \end{enumerate}
  Due to these observations, it should be clear that every ``original''
  collapsible  pushdown system of level $n$ and size $s$ 
  is simulated by a collapsible pushdown system of level $n$ and size
  at most $n\cdot s$. Moreover, with respect to
  $\varepsilon$-contractions of the configuration graphs both
  definitions  are equal  (from our systems to the original ones, we
  need one $\varepsilon$-transition in order to simulate each
  $\col{1}$ operation by two $\pop{1}$ operations; all other
  transitions are translated one-to-one).

\section{Types of Stacks and Configurations---Definitions} 
\label{sec:types}
Before we start defining types, let us introduce one more restriction
for a set of wf-rules. 

\begin{definition}
  A set $\Rr_\Xx$ of well-formed rules over $\Xx$ is called
  \emph{well-formed} if for each rule 
  $(\Rule{X}{\delta YZ})\in\Rr_\Xx$ there is also the rule 
  $(\Rule{X_{\delta Y}}{\delta Y})\in \Rr_\Xx$ for a
  new nonterminal $X_{\delta Y}$ such that $\lev(X_{\delta Y})=\lev(Y)$, 
  and for each rule $(\Rule{X}{\delta X})\in \Rr_\Xx$ or 
  $(\Rule{X}{\delta XY})\in \Rr_\Xx$
  there is also the
  rule
  $(\Rule{X_{\delta}}{\delta})\in\Rr_\Xx$ for a new nonterminal $X_\delta$.   
\end{definition}
\begin{remark}\label{rem:Grammars-well-formed}
  Note that each set of well-formed rules over some family $\Xx$ can
  be turned into a well-formed set of well-formed rules by adding the
  necessary symbols $X_\delta$ and $X_{\delta Y}$ to $\Xx$, the corresponding
  rules to the set of rules, and by setting $\lev(X_\delta)=n$ and
  $\lev(X_{\delta Y})=\lev(Y)$. 
\end{remark}

For the rest of this section we fix some finite family $\Xx$ described
by a well-formed set $\Rr_\Xx$ of wf-rules. The aim of this section is to assign
to any $k$-stack $s^k$ a \emph{type} $\type_\Xx(s^k)$ that determines the
possible runs from any of the sets $X\in \Xx$ starting in a stack with
topmost $k$-stack $s^k$. The type of $s^k$ is a set of
\emph{run descriptors} which come from a set $\Tt^k$ that are
defined inductively from $k=n$ to $k=0$. A typical
element of $\Tt^n$ has the form 
\begin{align*}
  \teA=(p, \deA) 
  \text{ with }
  \deA =(X,\Omega^n,\Omega^{n-1},\dots,\Omega^{\lev(X)+1},q),
\end{align*}
and a typical element of $\Tt^k$ (for $0\leq k < n$) has the form 
\begin{align*}
  \teA=(\Sigma^n, \Sigma^{n-1}, \dots, \Sigma^{k+1}, p, \deA) 
  \text{ with }
  \deA =(X,\Omega^n,\Omega^{n-1},\dots,\Omega^{\lev(X)+1},q),
\end{align*}
where $X\in \Xx$ is one of the
sets of runs we are interested in, $\Sigma^i$ and $\Omega^i$ are
types of $i$-stacks, and $p,q\in Q$ are states 
of the CPS $\Ss$. 
Let us explain the intended meaning of such a
tuple.
We want to have $\teA\in \type_\Xx(s^k)$ if and only if for all stacks
$t^n, t^{n-1}, \dots, t^{k+1}$  where $\Sigma^i\subseteq \type_\Xx(t^i)$ 
there is a run $R\in X$ such that 
$R(0)=(p, t^n:t^{n-1}:\dots:t^{k+1}:s^k)$ and 
$R(\lvert R\rvert) = (q,u^n:u^{n-1}:\dots: u^0)$ such that
$\Omega^j\subseteq\type_\Xx(u^j)$ for all $\lev(X)+1\leq j \leq n$. 
In other words, if we put $\teA$ into $\type_\Xx(s^k)$ we claim the
following. If for each $k+1 \leq i \leq n$ we take an $i$-stack $t^i$ 
that satisfies the claims of  $\Sigma^i$, then there is a run
in X that starts in state $p$ 
and the stack obtained by putting $s^k$ on top of the $\sL$-stack 
$t^{\sL}:\dots:t^{k+1}$, ends in state $q$, and the final stack
decomposes into a sequence of stacks such that the $j$-th element satisfies
all claims of $\Omega^j$. 

Recall that a $0$-stack contains the whole stack to which it links. 
Thus the type of a $0$-stack depends not only on its symbol, but also
on the whole stack it contains in the link. In order to deal with this
fact we define the type of a stack by induction on its level and by
induction on the nesting depth of its links.

We first introduce the set $\Tt^k$ of possible run descriptors of
level $k$ (the possible types of $k$-stacks are the elements of
$\Powerset(\Tt^k)$).\footnote{%
We write $\Powerset(X)$ for the power set of $X$.}

\begin{definition}
  Let $0\leq k\leq n$. Assume we have already defined sets $\Tt^i$ for
  $k+1\leq i\leq n$. 
  We take 
  $$
    \Tt^k =
    \{\mathsf{ne}\}\cup\big(\Powerset(\Tt^n)\times\Powerset(\Tt^{n-1})\times\dots
    \times\Powerset(\Tt^{k+1})\times
    Q\times \Dd^k\big),
 $$ 
 $$
 \Dd^k = \bigcup\{\Dd_X\colon X\in\Xx, \lev(X)\geq k\},\quad\mbox{where}
 $$ 
 $$
 \Dd_{X} =
 \{X\}\times\Powerset(\Tt^n)\times\Powerset(\Tt^{n-1})\times\dots
 \times\Powerset(\Tt^{\lev(X)+1})\times Q.
 $$ 
\end{definition}

Note that beside the run descriptors of the typical form, we also have
$\mathsf{ne}\in \Tt^k$: $\mathsf{ne}$ will appear in the type of some
stack if and only if this stack is non-empty.
In order to easily talk about the intended meaning of types and run
descriptors  we introduce the following definition.

\begin{definition}\label{def:agrees}
  Let $\deA=(X,\Omega^n,\Omega^{n-1},\dots,\Omega^{\lev(X)+1},q)\in\Dd_X$ for some
  $X\in\Xx$. 
  We say that a run $R$ \emph{agrees with $\deA$} if $R\in X$,
  and $R$ ends in a configuration $(q,t^n:t^{n-1}:\dots:t^0)$ such that
  $\Omega^i\subseteq \type(t^i)$ for  $\lev(X)+1\leq i\leq n$.
\end{definition}

Now we can reformulate the intended meaning of a run descriptors and
types. 

\begin{lemma}\label{lem:run-equiv-type}
  Let $\deA\in\Dd_X$ for some
  $X\in\Xx$, and let $0\leq l\leq \lev(X)$. 
  Let $c=(p,s^n:s^{n-1}:\dots:s^l)$ be a configuration.
  Then there is a run from $c$ which agrees with $\deA$ if and only if
  there is  a tuple $(\Psi^n,\Psi^{n-1},\dots,\Psi^{l+1},p,\deA)\in
  \type(s^l)$
  such that $\Psi^i\subseteq\type(s^i)$ for $l+1\leq i\leq n$.
\end{lemma}

We postpone the definition of types of stacks to the end of this
section. The proof of Lemma \ref{lem:run-equiv-type} can be found in
Appendix \ref{app:types}. 
Assuming that we already knew what the type of a stack is, it is easy
to define $\ctype_\Xx$, the function mapping configurations to their
types and to prove the first part of Theorem \ref{thm:types}. 

\begin{definition}
  Let $c=(q,s)$ be a configuration with $s=s^n:\dots:s^0$. 
  Set
  \begin{align*}
    \ctype_\Xx(c):=(\type_\Xx(s^\sL),\dots,\type_\Xx(s^0),q).  
  \end{align*}
  We define a partial order on the types of configurations as follows:
  $(\Phi^\sL,\dots,\Phi^0,p)\typOrd(\Psi^\sL,\dots,\Psi^0,q)$
  if and only if  $p=q$ and $\Phi^i\subseteq\Psi^i$ for $0\leq
  i\leq\sL$. 
\end{definition}

\begin{theorem*}[Part 1 of Theorem \ref{thm:types}]
  Let $\Xx$ be a family of sets of runs
  described by a  well-formed grammar,
  $R\in X\in\Xx$, and $c$ be a configuration with
  $\ctype_\Xx(R(0))\typOrd \ctype_\Xx(c)$.
  \begin{enumerate}
  \item  There is a run $S\in X$ starting in $c$
    which has the same final state as $R$.
  \end{enumerate}
\end{theorem*}

\begin{proof}[of part \ref{thm:typesA} of Theorem \ref{thm:types}]
  Let $R$ start in state $p$ and end in state $q$. 
  Let the pds of $R(0)$ be $s^n:s^{n-1}:\dots:s^0$,
  let the pds of $R(\lvert R\rvert)$ be $t^n:t^{n-1}:\dots:t^0$,
  and let the pds of $c$ be $u^n:u^{n-1}:\dots:u^0$.
  The assumptions say that $R$ agrees with 
  $$\deA=(X,\type(t^n),\type(t^{n-1}),\dots,
  \type(t^{\lev(X)+1}),q).$$
  Due to Lemma \ref{lem:run-equiv-type}, there are
  $\Psi^i\subseteq\type(s^i)$ for each $1\leq i\leq n$ such that
  run descriptor 
  $\teA=(\Psi^n,\Psi^{n-1},\dots,\Psi^1,p,\deA)\in\type(s^0)$. 
  Since the types of $s^i$ are included in the types of $u^i$, we also have
  $\teA\in \type(u^0)$, and $\Psi^i\subseteq\type(u^i)$ for $1\leq i\leq n$. 
  Using again Lemma \ref{lem:run-equiv-type} 
  there is a run $S$ from $c$ which agrees with $\deA$.
  By definition this implies that $S\in X$ and $S$ ends in state $q$. 

  We conclude with a further observation that will be useful when we
  prove part $2$ of the theorem. 
  Decomposing the final stack of $S$ as
  $v^n:v^{n-1}:\dots:v^0$, we obtain 
  \begin{align} \label{prop:thmTypesA-B}
    \type(t^i)\subseteq\type(v^i)\text{ for each }
    \lev(X)+1\leq i\leq \sL    
  \end{align}
  because by construction the final stack of $S$ realises the types
  described by $\deA$.\qed
\end{proof}

The proof of part \ref{thm:typesB} requires a more detailed knowledge
about the types. Thus, we postpone it to the end of this section.

Next we prepare the definition of types.
We first define composers. The intention is that a composer gives us
the type of a $k$-stack $s^k$ from the types of its decomposition as
$s^k=t^k:t^{k-1}:\dots:t^l$. 

\begin{definition}\label{def:composer}
  Let $0\leq l\leq k\leq n$, and let $\Psi^i\subseteq\Tt^i$ for each 
  $l\leq i\leq k$.  
  Their \emph{composer}
  $\comp(\Psi^k,\Psi^{k-1},\dots,\Psi^l)$ is 
  the subset of $\Tt^k$ containing all tuples 
  $(\Sigma^n,\Sigma^{n-1},\dots,\Sigma^{k+1},q,\deA)$ such that 
  there is a 
  $(\Sigma^n,\Sigma^{n-1},\dots,\Sigma^{l+1},q,\deA)\in\Psi^l$
  such that $\Sigma^i\subseteq\Psi^i$ for $l+1\leq i\leq k$ and $\deA\in\Dd^k$; 
  additionally the composer contains $\mathsf{ne}$ if $l<k$ or
  if $\mathsf{ne}\in\Psi^k$. 
\end{definition}
\begin{remark}
  Note that $\comp(\Psi^k)=\Psi^k$. 
  Furthermore, note that the definitions concerning $\mathsf{ne}$ are
  compatible with our intended meaning: a $k$-stack $s^k$ is nonempty
  if it is composed as
  $s^k=t^k:t^{k-1}:\dots:t^l$ for $l<k$. 
  All $t^i$ may be empty stacks, but the resulting $k$-stack contains a list of
  $(k-1)$-stacks whose topmost element is the possibly empty $(k-1)$-stack 
  $t^{k-1}:\dots:t^l$. Even if all elements of a list are empty
  stacks, it is not empty itself as long as it contains at least one
  element. 
\end{remark}

Note that the following properties of $\comp$ follow directly from the
definition.
\begin{lemma}\label{lem:comp-mon}
  Let $1\leq k\leq n$, and for each $0\leq i\leq k$ let
  $\Psi^i\subseteq\Phi^i \in\Powerset(\Tt^i)$. 
  Then $\comp(\Psi^k,\Psi^{k-1},\dots,\Psi^0)\subseteq
  \comp(\Phi^k,\Phi^{k-1},\dots,\Phi^0)$. 
\end{lemma}

\begin{lemma}\label{lem:comp-decomposition}
  Let $0\leq l< j < k \leq n$, and  let
  $\Psi^i \in\Powerset(\Tt^i)$ for all $l\leq i\leq k$. 
  Then $\comp(\Psi^k,\Psi^{k-1},\dots,\Psi^l)=
  \comp(\Psi^k,\Psi^{k-1},\dots,\Psi^{j+1},
    \comp(\Psi^{j},\Psi^{j-1},\dots,\Psi^l))$. 
\end{lemma}

In the set $\Rr_\Xx$ of rules we distinguish a subset $\Rr_\Xx^{>0}$
of those rules which are not of the form $\Rule{X}{}$ (i.e., they
describe runs of positive lengths). 
In the next step towards the definition of types, for each rule
$r=(\Rule{X}{\delta X_1\dots X_m})\in\Rr_\Xx^{>0}$ where the
operation in $\delta$ is $op$, we define two numbers $L(op)$ and
$M(r)$ as follows.
\begin{align*}
  L(op)&=
  \begin{cases}
    k & \text{if } op=\pop{k} \text{ or } op=\col{k},\\
    0 & \text{otherwise,}
  \end{cases}\\
  M(r)&=
  \begin{cases}
    \lev(X) & \text{if } m=0,\\
    \lev(X_1) &\text{otherwise}.
  \end{cases}
\end{align*}
Note that the inequality $0\leq L(op)\leq M(r)\leq n$ 
follows from the definition of wf-rules.
Then, to each rule $r=(\Rule{X}{\delta X_1\dots X_m})\in\Rr^{>0}_\Xx$
we assign a set 
\begin{align*}
  T(r)\subseteq&
  \big(\Powerset(\Tt^n)\times\Powerset(\Tt^{n-1})\times
  \dots\times\Powerset(\Tt^0)\big)\times\\   
  &\times\big(\Powerset(\Tt^n)\times\Powerset(\Tt^{n-1})\times
  \dots\times\Powerset(\Tt^{M(r)})\big)\times \Dd_X. 
\end{align*}
The intuitive meaning of this set is as follows.
Consider a run $R$ described by a wf-rule $r$.
The first part of a tuple describes the types of the stack of $R(0)$
(decomposed as $s^n:\dots:s^0$), 
the second part describes the types of the stack of $R(1)$ (decomposed
as $t^n:\dots:t^{M(r)}$)
and the last part is an element with which this run agrees.
In fact, in $T(r)$ we do not care whether $\delta$ can connect two
stacks of the described types for $R(0)$ and $R(1)$. 
We define it in such a way that under the assumption that $\delta$ may
connect two stacks of the corresponding types, the third part of
the tuple in $T(r)$ in fact permits a run that agrees with this
description and starts in the stack whose types are described by the
first part of the tuple. 
The question whether $\delta$ can transform a stack of a certain type
into a stack of another type is later dealt with when defining another
function $U$.
                  
\begin{definition}\label{def:set-T}
  Let $r=(\Rule{X}{\delta X_1 \dots X_m})\in\Rr^{>0}_\Xx$ and
  $\delta=(q_0, a, l, q_1, op)$.
  We distinguish the following cases.
  \begin{enumerate}
  \item	\label{def:set-T1}
    Assume that $r=(\Rule{X}{\delta})$.
    The set $T(r)$ contains all tuples 
    \begin{align*}
      ((\Psi^n,\Psi^{n-1},\dots,\Psi^0),
      (\Phi^n,\Phi^{n-1},\dots,\Phi^{\lev(X)}), \deA)\\
      \text{for }
      \deA=(X,\Omega^n,\Omega^{n-1},\dots,\Omega^{\lev(X)+1},q_1)
    \end{align*} 
    such that $\Omega^i\subseteq \Phi^i$ for $\lev(X)+1\leq i\leq n$
    (recall that $q_1$ is the state reached after application of $\delta$).
  \item	
    \label{def:set-T2}
    Assume that $r=(\Rule{X}{\delta X_1})$.
    The set $T(r)$ contains all tuples
    \begin{align*}
      ((\Psi^n,\Psi^{n-1},\dots,\Psi^0),
      (\Phi^n,\Phi^{n-1},\dots,\Phi^{\lev(X_1)}), \deA) \\
      \text{for }
      \deA=(X,\Omega^n,\Omega^{n-1},\dots,\Omega^{\lev(X)+1},q') 
    \end{align*}
    such that in $\Phi^{\lev(X_1)}$ we have a tuple
    \begin{align*}
      (\Phi^n&,\Phi^{n-1},\dots,\Phi^{\lev(X_1)+1},q_1, \deB)\\
      &\text{where }
      \deB=(X_1,\Omega^n,\Omega^{n-1},\dots,\Omega^{\lev(X_1)+1},q').
    \end{align*}
  \item	
    \label{def:set-T3}
    Assume that $r=(\Rule{X}{\delta X_1X_2})$.
    The set $T(r)$ contains all tuples
    \begin{align*}
      ((\Psi^n,\Psi^{n-1},\dots,\Psi^0),
      (\Phi^n,\Phi^{n-1},\dots,\Phi^{\lev(X_1)}), \deA)\\
      \text{for }
      \deA=(X,\Omega^n,\Omega^{n-1},\dots,\Omega^{\lev(X)+1},q')
    \end{align*} 
    such that in $\Phi^{\lev(X_1)}$ we have a tuple
    \begin{align*}
      (\Phi^n&,\Phi^{n-1},\dots,\Phi^{\lev(X_1)+1},q_1,\deB)\\
      &\text{for }
      \deB=(X_1,\Sigma^n,\Sigma^{n-1},\dots,\Sigma^{\lev(X_1)+1},q_2) 
    \end{align*}
    and in $\Psi^0$ we have a tuple
    \begin{align*}
      (\Sigma^n,\Sigma^{n-1},\dots,\Sigma^{\lev(X_1)+1},
      \Psi^{\lev(X_1)},\dots,\Psi^1,q_2,\deC)& \\
      \text{where }
      \deC=(X_2,\Omega^n,\Omega^{n-1},\dots,\Omega^{\lev(X_2)+1},q')&.
    \end{align*}
  \end{enumerate}
\end{definition}
\begin{remark}
  Recall that for each  rule of the form $\Rule{X}{\delta X_1X_2}$
  the definition of well-formed rules
  requires that any run of the form $\delta X_1$ starts and ends with the
  same topmost level $k$ stack for $k=\lev(X_1)$. Thus, for each run
  performing $\delta X_1$ such that the $X_1$-part agrees with $\deB$
  the type of the stack is completely determined: if the final stack
  decomposes as $s^n:s^{n-1}:\dots:s^1:s^0$, the type of $s^k$ for
  $k>\lev(X_1)$ is determined by $\deB$ and the type of $s^k$ for
  $k\leq \lev(X_1)$ is determined by the type of the initial stack,
  i.e., it is $\Psi^k$. 
\end{remark}

Using the function $T$ we define a function $U$ 
and a function $\stype$ which assigns types to $0$-stacks.
In fact, $U$ and $\stype$ are defined as simultaneous
fixpoints of sequences $(U_z)_{z\in\N}$ and $(\stype_z)_{z\in\N}$. 
For each $z\in\mathbb{N}$, each operation $op$, each number $1\leq
K\leq n$ and each $\Sigma^K\subseteq\Tt^K$ we define the set 
\begin{align*}
  U_z(op,K,\Sigma^K)&\subseteq
  \left(\Powerset(\Tt^n)\times\Powerset(\Tt^{n-1})\times\dots
  \times\Powerset(\Tt^0)\right)\times\\  
  &\times\left(\Powerset(\Tt^n)\times\Powerset(\Tt^{n-1})\times\dots
  \times\Powerset(\Tt^{L(op)})\right).  
\end{align*}
The intention is that $U_z(op,K,\Sigma^K)$ 
contains a tuple
\begin{align*}
  ((\Psi^n,\Psi^{n-1},\dots,\Psi^0),(\Pi^n,\Pi^{n-1},\dots,\Pi^{L(op)}))  
\end{align*}
 if
for a stack $s=s^n:s^{n-1}:\dots:s^0$ such that $\Psi^i$ is contained
in the type of $s^i$ and
$s^0$ carries a link of
level $K$ to a stack of type $\Sigma^K$, we can decompose
$op(s)=t^n:t^{n-1}:\dots:t^{L(op)}$  and $\Pi^j\subseteq \type_\Xx(t^j)$ for all
$L(op)\leq j \leq n$. 
When we enter the fixpoint $U(op,K,\Sigma^K)$ we are able to
replace the ``if'' by an ``if and only if''. 
In the definition of $U_z$ we use $\stype_{z-1}$ in order to
approximate the type of the topmost $0$-stack if $op=\push{1}$ (from
below).
In this case, the ``and only if'' part requires to
consider the complete type of the $0$-stack pushed onto the stack. 
The fixpoint $\stype$ of the
functions $\stype_{z}$ yields this complete type information. 
For the definition of $U_0$ and $\stype_0$, we assume that
$\stype_{-1}$ maps any 
input to $\emptyset$. 

\begin{definition}\label{def:set-U}
  Let $op$ be a stack operation, let $1\leq K\leq n$, let
  $\Sigma^K\subseteq\Tt^K$, and let $z\in\mathbb{N}$. 
  Assume that $\stype_{z-1}$ is already defined.
  We have four cases according to the operation used in $\delta$.
  \begin{enumerate}
  \item	\label{def:set-U1}
    Assume that $op=\pop{k}$.
    Then the set $U_z(op,K,\Sigma^K)$ contains all tuples
    $$((\Psi^n,\Psi^{n-1},\dots,\Psi^0),(\Psi^n,\Psi^{n-1},\dots,\Psi^k))$$
    where $\mathsf{ne}\in\Psi^k$.
  \item	\label{def:set-U2}
    Assume that $op=\col{k}$.
    If $k\neq K$ or $\mathsf{ne}\not\in\Sigma^K$, the set
    $U_z(op,K,\Sigma^K)$ is empty. 
    If $k=K$ and $\mathsf{ne}\in\Sigma^K$, the set $U_z(op,K,\Sigma^K)$
    contains all tuples 
    \begin{align*}
      ((\Psi^n,\Psi^{n-1},\dots,\Psi^0),
      (\Psi^n,\Psi^{n-1},\dots,\Psi^{k+1},\Sigma'^k))      
    \end{align*}
    such that $\Sigma'^k\subseteq\Sigma^K$.
  \item	\label{def:set-U3}
    Assume that $op=\push1_{b,k}$.
    The set $U_z(op,K,\Sigma^K)$ contains all tuples
    \begin{align*}
      ((\Psi^n,\Psi^{n-1},\dots,\Psi^0),(\Pi^n,\Pi^{n-1},\dots,\Pi^0))      
    \end{align*}
    which satisfy 
    $\Pi^i=\Psi^i$ for $2\leq i\leq k$, 
    $\Pi^1=\comp(\Psi^1,\Psi^0)$ and
    $\Pi^0\subseteq\stype_{z-1}(b,k,\Psi^k)$. 
  \item	\label{def:set-U4}
    Assume that $op=\push{k}$ with
    $k\geq 2$.
    The set $U_z(op,K,\Sigma^K)$ contains all tuples
    $$((\Psi^n,\Psi^{n-1},\dots,\Psi^0),(\Pi^n,\Pi^{n-1},\dots,\Pi^0))$$
    such that 
    $\Pi^i=\Psi^i$ for $0\leq i\leq n$ with $i\neq k$
    and $\Pi^k=\comp(\Psi^k,\Psi^{k-1},\dots,\Psi^0)$. 
  \end{enumerate}
\end{definition}

\begin{definition}\label{def:types}
  Let $z\in\mathbb{N}$, let $a\in\Gamma$, let $1\leq K\leq n$, and let
  $\Sigma^K\subseteq\Tt^K$.  
  Assume that $\stype_{z-1}$ and
  $U_z$ are already defined. 
  We define $\stype_z(a,K,\Sigma^K)$ as the set containing 
  \begin{enumerate}
  \item \label{def:types1}
    all tuples
    $$(\Psi^n,\dots,\Psi^1,q_0,
    (X,\Omega^n,\Omega^{n-1},\dots,\Omega^{\lev(X)+1},q_0))$$  
    such that $\Rr_\Xx$ contains the rule $\Rule{X}{}$, and
    $\Omega^i\subseteq\Psi^i$ for $\lev(X)+1\leq i\leq n$ (and
    $q_0$ is an arbitrary state), and 
  \item \label{def:types2} all tuples
    $(\Psi^n,\Psi^{n-1},\dots,\Psi^1,q_0,\deA)$
    such that for $0\leq m \leq 2$ and some rule 
    $r=(\Rule{X}{\delta X_1 \dots X_m})\in\Rr^{>0}_\Xx$ with
    $\delta=(q_0,a,\cdot,\cdot,op)$ 
    we have
    \begin{align*}
      &((\Psi^n,\dots,\Psi^0),(\Pi^n,\Pi^{n-1},\dots,\Pi^{L(op)}))\in U_z(op,K,\Sigma^K),\\ 
      &((\Psi^n,\dots,\Psi^0),
      (\Pi^n,\Pi^{n-1},\dots,\Pi^{M(r)+1},\Phi^{M(r)}),\deA)\in T(r),\\ 
      &\Psi^0\subseteq\stype_{z-1}(a,K,\Sigma^K), \quad\mbox{and}\\
      &\Phi^{M(r)}=\comp(\Pi^{M(r)},\Pi^{M(r-1)},\dots,\Pi^{L(op)}).
    \end{align*}
  \end{enumerate}
\end{definition}

Notice that the sequence $\stype_z$ is monotone with respect to both $z$ and $\Sigma^K$: for 
$\Sigma^K\subseteq \Sigma'^K$ and
each
$z\in\mathbb{N}$  we have
$\stype_z(a,K,\Sigma^K)\subseteq \stype_{z+1}(a,K,\Sigma'^K)$.
Independent of $z\in \N$, the domain and range of $\stype_z$ are fixed
finite sets whence there is some $z\in\N$ such that
$\stype_z=\stype_{z+1}$. This fixpoint is denoted as $\stype_\Xx$
(formally, $\stype_z$ also depends on $\Xx$). 

\begin{definition}\label{def:types-k}
  We define $\type_\Xx(s^k)$ for each  $k$-stack $s^k$ (for $0\leq k\leq n$)
  by induction on the structure of $s^k$. 
  If $s^k$ is empty, 
  \begin{align*}
    \type_\Xx(s^k) := \emptyset.
  \end{align*}
  Otherwise, 
  assume that $k=0$ and $s^k=(a,K,t^K)$ 
  where $a\in\Gamma$,
  $1\leq K\leq n$, and $t^K$ is a $K$-stack such that $\type_\Xx(t^K)$ is
  already defined.
  In this case we set 
  \begin{align*}
    \type_\Xx(s^k)=\stype_\Xx\left(a,K,\type_\Xx(t^K)\right).    
  \end{align*}
  Finally, assume that $k\geq 1$ and
  $s^k=t^k:t^{k-1}$ such that $\type_\Xx(t^k)$ and $\type_\Xx(t^{k-1})$ are
  defined. In this case set 
  \begin{align*}
    \type_\Xx(s^k)=\comp\left(\type_\Xx(t^k),\type_\Xx(t^{k-1})\right).    
  \end{align*}
\end{definition}

With the help of Lemma \ref{lem:run-equiv-type} and the properties of
the composer, the proof of part \ref{thm:typesB} of Theorem \ref{thm:types} is
done as follows. 

\begin{proof}[of part \ref{thm:typesB} of Theorem \ref{thm:types}]
  The proof is by induction on the length of the run.
  Let $X\in\Xx$ be a set of level $0$.
  Let $R\in X$ be a run starting in a pds $s^\sL:s^{\sL-1}:\dots:s^0$
  and ending in a pds $t^\sL:t^{\sL-1}:\dots:t^0$. 
  Furthermore, let $u:=u^\sL:u^{\sL-1}:\dots:u^0$ be a pds.
  Assume that $\type_\Xx(s^i)\subseteq \type_\Xx(u^i)$ for $0\leq i\leq \sL$. 
  We prove that there is a run $S\in X$ such that $S$
  starts in $u$ and ends in a stack $v^\sL:v^{\sL-1}:\dots:v^0$ with
  $\type_\Xx(t^i)\subseteq \type_\Xx(v^i)$ for each $0\leq i\leq \sL$
  and such that $R$ and $S$ have the same initial and final states. 
  We continue by case distinction on the wf-rule $r$ describing $R$.
  \begin{itemize}
  \item   If $r=(\Rule{X}{})$, the run $S$ of length $0$ with
    $S(0)=(q,u)$ for $q$ the initial state of $R$ satisfies the claim.
  \item 
    Assume that $r=(\Rule{X}{\delta})$.
    Because $\Xx$ is described by a well-formed set of rules, the run
    is also described by the rule of the set $X_{\delta}$. 
    Using 
    part \ref{thm:typesA} of Theorem \ref{thm:types}, 
    there is a run $S$ starting in $u$ which performs $\delta$, and
    such that 
    $R$ and $S$ have the same initial and final states. 
    Since $\lev(X)=0$, the operation in $\delta$ is a $\push{k}$
    of some level $k$. 
    Notice that for $1\leq i\leq k-1$ and for $k+1\leq i\leq n$ we have
    $t^i=s^i$ and $v^i=u^i$, so 
    $\type_\Xx(t^i)\subseteq\type_\Xx(v^i)$.
    The same holds for $i=0$ if $k\geq 2$.
    If the operation is $\push 1_{a,j}$, we have $t^0=(a,j,s^j)$ and
    $v^0=(a,j,u^j)$. Since $\stype_\Xx$ is monotone, 
    $\type_\Xx(t^0)=\stype_\Xx(a,j,\type_\Xx(s^j))\subseteq
    \stype_\Xx(a,j,\type_\Xx(u^j))=\type_\Xx(v^0)$.
    We also have $t^k=s^k:s^{k-1}:\dots:s^0$ and
    $v^k=u^k:u^{k-1}:\dots:u^0$.
    Due to Lemmas \ref{lem:comp-decomposition} and \ref{lem:comp-mon},
    we obtain 
    \begin{align*}
      \type_\Xx(t^k)
      &=\comp(\type_\Xx(s^k),\type_\Xx(s^{k-1}),\dots,\type_\Xx(s^0))\subseteq\\
      &\subseteq\comp(\type_\Xx(u^k),\type_\Xx(u^{k-1}),\dots,\type_\Xx(u^0))=
      \type_\Xx(v^k).  
    \end{align*}
  \item 
    Assume that $r=(\Rule{X}{\delta Y})$.
    By definition of a wf-rule, $\lev(Y)=0$.
    Decompose $R=R_1\circ R_2$ where $R_1$ has length $1$.
    As in the above case, we obtain a run $S_1$ from $u$ of length $1$,
    performing transition $\delta$, 
    such that the types of $R_1(1)$ and $S_1(1)$ are appropriately
    contained (for all levels), and 
    that $R_1$ and $S_1$ have the same initial and final states.
    Then we apply the induction assumption for $R_2\in Y$ and obtain a run
    $S_2\in Y$ from $S_1(1)$, 
    such that the types at the end of $R_2$ and $S_2$ are contained as
    required (for all levels), and  
    that the final states are the same. Thus,  $S:= S_1\circ S_2$
    satisfies the claim.
  \item 
    Finally, assume that $r=(\Rule{X}{\delta YZ})$.
    By definition of a wf-rule, $\lev(Z)=0$.
    Decompose $R=R_1\circ R_2$ where $R_1$ performs the transition
    $\delta$ followed by a run from $Y$, and $R_2$ is in $Z$. 
    Since $\Xx$ is described by a well-formed set of rules, 
    $R_1$ is in the set described by $X_{\delta Y}$.
    Using part \ref{thm:typesA} of Theorem
    \ref{thm:types} for $R_1$ and $X_{\delta Y}$, 
    there is a run $S_1$ from $u$ which performs $\delta$ followed by
    a run from $Y$  such 
    that $R_1$ and $S_1$ have the same initial and final states. 
    Decompose the pds of $R_1(\lvert R_1\rvert)$ and $S_1(\lvert S_1\rvert)$ as
    $s'^n:s'^{n-1}:\dots:s'^0$ and $u'^n:u'^{n-1}:\dots:u'^0$. 
    Recall from the proof of part \ref{thm:typesA} of Theorem
    \ref{thm:types} (see page \pageref{prop:thmTypesA-B}) that
    $\type_\Xx(s'^i)\subseteq\type_\Xx(u'^i)$ for 
    $\lev(Y)+1\leq i\leq n$ (notice that $\lev(X_{\delta
      Y})=\lev(Y)$). 
    But by definition of a wf-rule we know that the topmost
    $\lev(Y)$-stack of $R_1(0)$ and of $R_1(\lvert R_1\rvert)$ are the same, so
    $s'^i=s^i$ for $0\leq i\leq\lev(Y)$; 
    for the same reason $u'^i=u^i$ for $0\leq i\leq\lev(Y)$.
    Thus, $\type_\Xx(s'^i)\subseteq\type_\Xx(u'^i)$ for $0\leq i\leq n$.
    Then we apply the induction assumption to $R_2\in Z$ and obtain a run
    $S_2\in Z$ from $S_1(\lvert S_1\rvert)$ 
    such that the types at the end of $R_2$ and $S_2$ are contained as
    required (for all levels) and such  
    that the final states are the same. Thus, $S:=S_1\circ S_2$
    satisfies the claim. \qed
  \end{itemize}
\end{proof}

\subsubsection*{Types in previous papers.}
A similar concept of defining types were already present in \cite{parys2011} and \cite{parys-pumping}.
In both these papers types were used only for systems without collapse.
The types in \cite{parys2011} are defined completely semantically:
the definition is similar to our Lemma \ref{lem:run-equiv-type}.
Then it is necessary to prove that the type of $s^l$ does not depend on the choice of $s^n, s^{n-1},\dots,s^{l+1}$ present in the assumptions of the lemma.
We were unable to give a proof of the analogous fact for systems with collapse.

The types in \cite{parys-pumping} are much more similar to our types: they are also defined syntactically, i.e.~basing on possible transitions of the system.
But these types were defined only for one class of runs, namely for $k$-returns for each $k$.
The generalisation to an arbitrary family described by wf-rules required mainly the invention of a proper definition of these rules.
The generalisation to systems with collapse required mainly the invention of a proper definition of stacks (i.e.~that a $0$-stack should keep the copy of the linked stack, instead of just a link).


\section{Types of Stacks---Proofs}
\label{app:types}
This appendix is devoted to the proof of Lemma \ref{lem:run-equiv-type}.
We assume that the family $\Xx$ is fixed, and we write $\type$ for $\type_\Xx$.
We will first prove the left-to-right implication of this lemma, as
restated below (in a slightly stronger version).

\begin{lemma}\label{lem:run2type}
  Let $\deA\in\Dd_X$ for some
  $X\in\Xx$, and let $0\leq l\leq \lev(X)$. 
  Let $R$ be a run which agrees with $\deA$, where
  $R(0)=(q_0,s^n:s^{n-1}:\dots:s^l)$. 
  Then 
  $$(\type(s^n),\type(s^{n-1}),\dots,\type(s^{l+1}),q_0,\deA)\in
  \type(s^l).$$ 
\end{lemma}

The proof is by induction on the length of $R$.
It is divided into several lemmas; the division follows the steps
in the definition of types, i.e., we prove certain properties of the
functions $T$ 
and $U$ which finally allow to prove the lemma. 
We start with two observations. The first follows
immediately from 
the definitions and the second is a corollary of 
Lemma \ref{lem:comp-decomposition}. 

\begin{proposition} \label{prop:typeNE}
  Let $1\leq k\leq n$, and let $s^k$ be a $k$-stack.
  Then $\mathsf{ne}\in \type(s^k)$ if and only if $s^k$ is not empty.
\end{proposition}

\begin{proposition}\label{prop:composer-lr}
  Let $0\leq l\leq k\leq n$, and let $s=s^k:s^{k-1}:\dots:s^l$ be a $k$-stack.
  The type of $s$, $\type(s)$, is
  $\comp(\type(s^k),\type(s^{k-1}),\dots,\type(s^l))$.
\end{proposition}

The next lemma proves our intuition about $T$.
This lemma uses the ``big'' induction assumption, i.e., 
Lemma \ref{lem:run2type} for shorter runs. 

\begin{lemma}\label{lem:set-T-lr}
  Let $R$ be a run which agrees with some $\deA\in\Dd_X$, and
  let $r=(\Rule{X}{\delta X_1\dots X_m})\in\Rr_\Xx^{>0}$ be a
  rule which describes $R$. 
  Assume that the statement of Lemma \ref{lem:run2type} is true
  for all runs strictly shorter than $R$. 
  Decompose the stack of $R(0)$ as $s^n:s^{n-1}:\dots:s^0$, and
  the stack of $R(1)$ as $t^n:t^{n-1}:\dots:t^{M(r)}$. 
  Then $T(r)$ contains the tuple 
  \begin{align*}
    \eta=\left(\big(\type(s^n),\dots,\type(s^0)\big),
    \big(\type(t^n),\type(t^{n-1}),\dots,\type(t^{M(r)})\big),\deA\right).
  \end{align*} 
\end{lemma}

\begin{proof}
  Let $\deA=(X,\Omega^n,\Omega^{n-1},\dots,\Omega^{\lev(X)+1},q')$ and
  $q_1$ be the state of $R(1)$, which is also the state reached after
  application of  $\delta$. 
  We distinguish three cases depending on the form of $r$, i.e., on
  the value of $m\in\{0,1,2\}$.
  \begin{enumerate}
  \item   Assume that $r$ is $\Rule{X}{\delta}$.
    Then $\lvert R \rvert=1$ whence the state of $R(\lvert R\rvert)$ is $q_1$
    and the stack of $R(\lvert R\rvert)$ is
    $t^n:t^{n-1}:\dots:t^{\lev(X)}$ (recall 
    that $\lev(X)=M(r)$). 
    Since $R$ agrees with $\deA$, $q'=q_1$ and 
    $\Omega^i\subseteq \type(t^i)$ for $\lev(X)+1\leq i\leq n$. 
    Due to Point \ref{def:set-T1} of Definition \ref{def:set-T},
    $\eta\in T(r)$.
  \item 
    Assume that $r$ is $\Rule{X}{\delta X_1}$.
    Let
    $\deB=(X_1,\Omega^n,\Omega^{n-1},\dots,\Omega^{\lev(X_1)+1},q')$ 
    where $\Omega^i=\emptyset$ for $\lev(X_1)+1\leq i\leq \lev(X)$ (and
    the other $\Omega^i$ are specified by $\deA$). 
    We know that $\subrun{R}{1}{\lvert R \rvert}$ is in $X_1$ whence it
    agrees with $\deB$.  
    Application of Lemma \ref{lem:run2type} to the shorter run 
    $\subrun{R}{1}{\lvert R \rvert}$, to
    $\deB$ and to $\lev(X_1)$  
    yields
    \begin{align*}
      \big(\type(t^n),\type(t^{n-1}),\dots,\type(
      t^{\lev(X_1)+1}),q_1, 
        \deB\big)\in 
     \type\left(t^{\lev(X_1)}\right).
    \end{align*}  
    Recall that $M(r)=\lev(X_1)$.
    Due to Point \ref{def:set-T2} of Definition \ref{def:set-T}, 
    $\eta\in T(r)$. 
  \item 
    Assume that $r$ is  $\Rule{X}{\delta X_1X_2}$.
    Fix some $1\leq i\leq \lvert R \rvert$ such that
    $\subrun{R}{1}{i}\in X_1$ and 
    $\subrun{R}{i}{\lvert R \rvert}\in X_2$. 
    Let $u^n:u^{n-1}:\dots:u^0$ be the stack of $R(i)$, and $q_2$ its state.
    Let
    \begin{align*}
      &\deB = (X_1,\type(u^n),\type(u^{n-1}),\dots,\type(u^{\lev(X_1)+1}),q_2),\\
      &\deC = (X_2,\Omega^n,\Omega^{n-1},\dots,\Omega^{\lev(X_2)+1},q'),
    \end{align*}
    where $\Omega^i=\emptyset$ for $\lev(X_2)+1\leq i\leq \lev(X)$.
    The subrun $\subrun{R}{1}{i}$ agrees with $\deB$, and the subrun
    $\subrun{R}{i}{\lvert R  \rvert}$ agrees with $\deC$. 
    Application of Lemma \ref{lem:run2type} 
    to the shorter run $\subrun{R}{1}{i}$, to
    $\deB$ and to $\lev(X_1)$  yields
    \begin{align*}
      (\type(t^n),\type(t^{n-1}),\dots,\type(t^{\lev(X_1)+1}),q_1,\deB)
      \in\type(t^{\lev(X_1)}).
    \end{align*}
    Analogously, application of the lemma to
    $\subrun{R}{i}{\lvert R \rvert}$, to $\deC$
    and to $0$ yields 
    \begin{align*}
      (\type(u^n),\type(u^{n-1}),\dots,\type(u^1),q_2,\deC)\in
      \type(u^0).
    \end{align*}
    The definition of a wf-rule implies that the topmost
    $\lev(X_1)$-stacks of $R(0)$ and $R(i)$ coincide whence
    $\type(s^i)=\type(u^i)$ for $\lev(X_1)\geq i \geq 0$.  
    Thus, Point \ref{def:set-T3} of Definition \ref{def:set-T}
    (with 
    $\Psi^i=\type(s^i)$ , $\Phi^i =  \type(t^i)$  and
    $\Sigma^i = \type(u^i)$), 
    implies that $\eta\in T(r)$.\qed
  \end{enumerate}
\end{proof}

Having related the sets $T(r)$ with the subruns starting after the
first transition of some run described by $r$, 
we now relate the function $U_z$ with the first operation of such a
run. Recall that $\stype_\Xx$ is the fixpoint of the sequence 
$(\stype_{z})_{z\in\N}$ which is reached at some $z'\in\N$, i.e., 
$\stype_{z'}=\stype_{z'-1}=\stype_\Xx$. For the next lemma we fix this
value $z'$.

\begin{lemma}\label{lem:set-U-lr}
  Let $s=s^n:s^{n-1}:\dots:s^0$ and $t=t^n:t^{n-1}:\dots:t^{L(op)}$ be
  pds such that $s^0$ contains a link of level $K$ to a stack $u^K$. 
  Assume that $t=op(s)$ for some operation $op$.
  For all $z > z'$,
  $U_z(op,K,\type(u^K))$ contains 
  \begin{align*}
    \eta=\Big(\big(\type(s^n)&,\type(s^{n-1}),\dots,\type(s^0)\big),\\
    &\big(\type(t^n),\type(t^{n-1}),\dots,\type(t^{L(op)})\big)\Big).
  \end{align*} 
\end{lemma}
\begin{proof}
  The proof is by case distinction on $op$. Fix $z > z'$
  \begin{enumerate}
  \item   Assume that $op=\pop{k}$.
    Then we have $t=s^n:s^{n-1}:\dots:s^k$, so $t^i=s^i$ for $k\leq
    i\leq n$ (recall that $L(op)=k$). 
    In particular $s^k$ is not empty, so $\mathsf{ne}\in \type(s^k)$.
    Due to Point \ref{def:set-U1} of Definition \ref{def:set-U}, we
    conclude that $\eta\in U_z(op,K,\type(u^K))$. 
  \item   Assume that $op=\col{k}$.
    In this case, $k=K$ and $u^K$ is not empty (equivalently:
    $\mathsf{ne}\in\type(u^K)$) because otherwise $\col{k}$ would not
    be applicable.  
    We have $t=s^n:s^{n-1}:\dots:s^{k+1}:u^k$, so $t^i=s^i$ for $k+1\leq
    i\leq n$, and $t^k=u^k$ (recall that $L(op)=k$). 
    Due to Point \ref{def:set-U2} of Definition \ref{def:set-U}, we
    conclude
    that $\eta\in
    U_z(op,K,\type(u^K))$.
  \item Assume that $op=\push1_{b,k}$.
    Then  
    \begin{align*}
      t=s^n:s^{n-1}:\dots:s^2:(s^1:s^0):(b,k,s^k)      
    \end{align*}
    whence
    $t^i=s^i$ for $2\leq i\leq n$, $t^1=s^1:s^0$ and
    $t^0=(b,k,s^k)$ (recall that $L(op)=0$). 
    Due to Proposition \ref{prop:composer-lr},
    $\type(t^1)=\comp(\type(s^1),\type(s^0))$. 
    Additionally, 
    \begin{align*}
      \type(t^0) = \stype(b,k,\type(s^k)) = \stype_{z'}(b,k,\type(s^k)).
    \end{align*}
    Thus, using Point \ref{def:set-U3} of Definition \ref{def:set-U}
    we conclude that
    $\eta\in U_z(op,K,\type(u^K))$.
  \item 
    Finally, assume that $op=\push{k}$ for $k\geq 2$.
    Then we have
    \begin{align*}
    &t=s^n:s^{n-1}:\dots:s^{k+1}:t^k:s^{k-1}:\dots:s^0 \text{ where}\\
    &t^k=(s^k:s^{k-1}:\dots:s^0).
    \end{align*}
    Thus, $t^i=s^i$ for $0\leq i\leq n$ with $i\neq k$ (recall that
    $L(op)=0$).  
    Proposition \ref{prop:composer-lr} implies that
    \begin{align*}
      \type(t^k)=\comp(\type(s^k),\type(s^{k-1}),\dots,\type(s^0)).       
    \end{align*}
    Thus, using  Point \ref{def:set-U4} of Definition \ref{def:set-U},
    we conclude that
    $\eta\in U_z(op,K,\type(u^K))$. \qed
  \end{enumerate}
\end{proof}

We are now prepared to prove Lemma \ref{lem:run2type}.

\begin{proof}[of Lemma \ref{lem:run2type}] \label{proof:run2typeA}
  Let  $R$ be a run with $R(0)=s^n:s^{n-1}:\dots:s^l$ that agrees with 
  \begin{align*}
    \teA=(\type(s^n),\type(s^{n-1}),\dots,\type(s^{l+1}),q_0,\deA).    
  \end{align*}
  We make an external induction on the length of $R$ and an
  internal induction on $l$. 
  \begin{itemize}
  \item 
    Assume that $l=0$ and $\lvert R \rvert=0$.
    Let $s^0=(a,K,u^K)$.
    Since $R$ agrees with $\deA$, we have 
    $\deA=(X,\Omega^n,\Omega^{n-1},\dots,\Omega^{\lev(X)+1},q_0)$
    where $\Omega^i\subseteq\type(s^i)$ for $\lev(X)+1\leq i\leq n$
    (and $q_0$ is the state of $R(0)$). 
    Due to Point \ref{def:types1} of Definition \ref{def:types},
    $\teA\in\stype_z(a,K,\type(u^K))$ for every $z\in\N$ whence 
    $\teA\in\type(s^0)$.
  \item 
    Assume that $l=0$ and $\lvert R\rvert > 0$.
    Then there is a rule 
    $r=(\Rule{X}{\delta X_1\dots X_m})\in\Rr_\Xx^{>0}$ describing $R$.  
    Let $s^0=(a,K,u^K)$. We have $\delta=(q_0,a,\cdot,\cdot,op)$.
    Let $t^n:t^{n-1}:\dots:t^{L(op)}$ be the stack of $R(1)$.
    Lemma \ref{lem:set-U-lr} 
    implies that for all $z>z'$
    the set $U_z(op,K,\type(u^K))$ contains  
    \begin{align*}
      \left((\type(s^n),\type(s^{n-1}),\dots,\type(s^0)),
      (\type(t^n),\type(t^{n-1}),\dots,\type(t^{L(op)}))\right).
    \end{align*} 
    Setting $v^{M(r)}=t^{M(r)}:t^{M(r)-1}:\dots:t^{L(op)}$,
    the stack of $R(1)$ is
    $t^n:t^{n-1}:\dots:t^{M(r)+1}:v^{M(r)}$. 
    Our induction assumption on shorter runs and Lemma
    \ref{lem:set-T-lr} implies that $T(r)$ contains
    \begin{align*}
      &\Big(\big(\type(s^n),\type(s^{n-1}),\dots,\type(s^0)\big),
      \big(\type(t^n),\\  
      &\hspace{0.9cm}\type(t^{n-1}),\dots,\type(t^{M(r)+1}),
      \type(v^{M(r)})\big),\deA\Big).  
    \end{align*}
    Additionally, Proposition \ref{prop:composer-lr} implies that
    \begin{align*}
      \type(v^{M(r)})=\comp\left(\type(t^{M(r)}),\type(t^{M(r)-1}),\dots,
      \type(t^{L(op)})\right). 
    \end{align*}
    Due to Point \ref{def:types2} of Definition \ref{def:types},
    $\teA\in\stype_z(a,K,\type(u^K))$ for $z>z'$ whence  
    $\teA\in\type(s^0)$.
  \item 
    Assume that $l\geq 1$.
    Decompose $s^l=t^l:t^{l-1}$.
    The (inner) induction assumption implies that
    $\type(t^{l-1})$ contains the tuple 
    $$(\type(s^n),\type(s^{n-1}),\dots,\type(s^{l+1}),\type(t^l),q_0,\deA).$$
    From Definition \ref{def:composer}, it follows that 
    \mbox{$\teA\in \comp(\type(t^l),\type(t^{l-1}))=\type(s^l)$.} \qed
  \end{itemize}
\end{proof}

The rest of this appendix deals with the right-to-left implication of
Lemma \ref{lem:run-equiv-type}. 
In the proof we use the notion of \emph{having a witness}.
The intuition is that a stack and a run descriptor  have a witness,
if this right-to-left implication holds for them. 
Our goal is to prove that every such pair has a witness, which means that
the implication is always true. 

\begin{definition}\label{def:witness}
  Let $0\leq k\leq n$, let $s^k$ be a $k$-stack, and let
  $\Phi^k\subseteq \Tt^k$. 
  We define when $(s^k,\Phi^k)$ \emph{has a witness} by induction on
  $k$, starting with $k=n$.  
  We say that $(s^k,\Phi^k)$ has a witness
  if $(s^k,\teA)$ has a witness for every $\teA\in\Phi^k$, as
  defined below. 
  \begin{itemize}
  \item $(s^k,\mathsf{ne})$ has a witness if
    $\mathsf{ne}\in\type(s^k)$ (equivalently: if $k\geq 1$ and $s^k$
    is nonempty). 
  \item For
    \begin{align*}
      &\teA=(\Phi^n,\Phi^{n-1},\dots,\Phi^{k+1},p,\deA) \text{ and}\\
      &\deA=(X,\Omega^n,\Omega^{n-1},\dots,\Omega^{\lev(X)+1},q),
    \end{align*}
    $(s^k, \teA)$ has a    witness if 
    \begin{itemize}
    \item $\teA\in\type(s^k)$ and
    \item 
      for each configuration $c =(p, t^n : \dots : t^{k+1}:s^k)$ 
      such that  $(t^i,\Phi^i)$ has a witness for each $k+1\leq i\leq n$
      there is a run $R_c$ from $c$ to some stack
      $u^n:u^{n-1}:\dots: u^0$ such that $R_c$ 
      agrees with $\deA$ and   
      $(u^i,\Omega^i)$ has a witness for each $\lev(X)+1\leq i\leq n$. 
    \end{itemize}
  \end{itemize}
\end{definition}

We first prove that composers preserve witnesses.

\begin{proposition}\label{prop:low-high-uv}
  Let $0\leq l\leq k\leq n$.
  For each $l\leq i\leq k$ let $s^i$ be an $i$-stack, and
  let $\Phi^i\subseteq\Tt^i$ be such that $(s^i,\Phi^i)$ has a
  witness. 
  Then $(s^k:s^{k-1}:\dots:s^l,\comp(\Phi^k,\Phi^{k-1},\dots,\Phi^l))$ has a witness.
\end{proposition}

\begin{proof}	
  We have to show that for each $\teA^k\in
  \comp(\Phi^k,\Phi^{k-1},\dots,\Phi^l)$,  
  $(s^k:s^{k-1}:\dots:s^l,\teA^k)$ has a witness.
  By Proposition \ref{prop:composer-lr} $\teA^k\in\type(s^k:s^{k-1}:\dots:s^l)$.
  If $\teA^k=\mathsf{ne}$ we are already done.
  Otherwise, 
  $\teA^k=(\Sigma^n, \dots, \Sigma^{k+1}, p, \deA)$ for some
  $\deA=(X, \Omega^n, \dots, \Omega^{\lev(X)+1},q)$. 
  By definition of the composer, $\Phi^l$ contains a tuple
    $\teA^l=(\Sigma^n,\Sigma^{n-1} \dots, \Sigma^{l+1}, p, \deA)$ 
    such that $\Sigma^i\subseteq\Phi^i$ for $l+1\leq i\leq k$. 
    Let $c=(p, t^n:\dots:t^{k+1}:s^k:\dots:s^l)$ be a configuration
    for 
    stacks $t^i$ such that $(t^i, \Sigma^i)$ has a witness for each 
    $k+1\leq i \leq n$.  
    By assumption of the lemma, also $(s^i, \Sigma^i)$ has a witness
    for $l+1\leq i \leq k$ (since $\Sigma^i\subseteq\Phi^i$), and
    $(s^l,\teA^l)$ has a witness (since $\teA^l\in\Phi^l$). 
    Application of Definition \ref{def:witness} to $(s^l,\teA^l)$
    shows that there 
    is a run $R$ from $c$ to some configuration $(q, u^n:\dots:u^0)$ 
    which agrees with $\deA$ and such that $(u^i, \Omega^i)$ has a
    witness for $lev(X)+1\leq i\leq n$. 
    Thus, $R$ also shows that $(s^k:\dots:s^l, \teA^k)$ has a
    witness. \qed
\end{proof}

Our next goal is to show that each
$0$-stack has a witness. 

\begin{lemma}\label{lem:type2witness}
  Let $z\in\N$, $1\leq K\leq n$,  $u^K$  a $K$-stack, and 
  $\Sigma^K\subseteq\Tt^K$ such that $(u^K,\Sigma^K)$ has a witness. 
  Let $a\in\Gamma$, and let $\teA\in \stype_z(a,K,\Sigma^K)$.
  Then $((a,K,u^K),\teA)$ has a witness.
\end{lemma}

\begin{corollary}\label{cor:type2witness}
  Let $1\leq K\leq n$, $u^K$ be a $K$-stack,
  $\Sigma^K\subseteq\Tt^K$ be such that $(u^K,\Sigma^K)$ has a
  witness, and $a\in\Gamma$. 
  $\left((a,K,\Sigma^K), \type((a,K,\Sigma^K))\right)$ has a witness.
\end{corollary}

The proof of the lemma is by induction on the fixpoint stage $z$.
As an auxiliary step we show how the set $T(r)$ can be
used to prove that there is an appropriate run described by $r$.

\begin{lemma}\label{lem:AuxiliaryT}
  Let $r=(\Rule{X}{\delta X_1\dots X_m})$ be a rule from
  $\Rr_\Xx^{>0}$, and let
  $\deA=(X,\Omega^n,\Omega^{n-1},\dots,\Omega^{\lev(X)+1},q')\in\Dd_X$. 
  Let $R_1$ be a run of length $1$ from stack $s^n:s^{n-1}:\dots:s^0$
  to stack $t^n:t^{n-1}:\dots:t^{M(r)}$ performing the transition
  $\delta$. 
  For $0\leq i\leq n$, let $\Psi^i$ be such that $(s^i,\Psi^i)$ has a
  witness, and  
  for $M(r)\leq i\leq n$, let $\Phi^i$ be such that $(t^i,\Phi^i)$ has
  a witness. 
  Assume that  
  $$((\Psi^n,\Psi^{n-1},\dots,\Psi^0),
  (\Phi^n,\Phi^{n-1},\dots,\Phi^{M(r)}),\deA)\in T(r).$$
  Then there exists a run $R$ from $R_1(0)$ which agrees with $\deA$, and 
  ends in a stack $v^n:v^{n-1}:\dots:v^0$ such that $(v^i,\Omega^i)$
  has a witness for $\lev(X)+1\leq i\leq n$. 
\end{lemma}

\begin{proof}
  Let $q_1$ be the state of $R_1(1)$, which is also the 
  state reached after application of $\delta$. 
  We distinguish three cases depending on the shape of $r$, i.e., on
  the value of $m\in\{0,1,2\}$.
  \begin{enumerate}
  \item   Assume that $r$ is $\Rule{X}{\delta}$.
  Recall that $M(r)=\lev(X)$.
  Set $R = R_1$.
  By definition of the set $T(r)$, we have $q_1=q'$, and
  $\Omega^i\subseteq\Phi^i$ for $\lev(X)+1\leq i\leq n$. 
  Notice that $(t^i,\Omega^i)$ has a witness (in particular
  $\Omega^i\subseteq\type(t^i)$) for $\lev(X)+1\leq i\leq 
  n$, because $(t^i,\Phi^i)$ has a witness, and
  $\Omega^i\subseteq\Phi^i$. 
  Observe that $R_1$ is in $X$ whence it agrees with $\deA$.
\item 
  Assume that $r$ is $\Rule{X}{\delta X_1}$.
  Recall that $M(r)=\lev(X_1)$.
  By definition of $T(r)$, we have 
  \begin{align*}
  &\teB=(\Phi^n,\Phi^{n-1},\dots,\Phi^{\lev(X_1)+1},q_1,\deB)\in 
  \Phi^{\lev(X_1)} \text{ where}\\
  &\deB=(X_1,\Omega^n,\Omega^{n-1},\dots,\Omega^{\lev(X_1)+1},q').
  \end{align*}
  Since $(t^{\lev(X_1)},\teB)$ has a witness and $(t^i,\Phi^i)$ has a
  witness for $\lev(X_1)+1\leq i\leq n$, 
  there is a run $R_2$ agreeing with $\deB$ from $R_1(1)$ to a stack
  $v^n:v^{n-1}:\dots:v^0$ such  
  that $(v^i,\Omega^i)$ has a witness for $\lev(X_1)+1\leq i\leq n$.
  As $R$ we take $R_1\circ R_2$; this run is in $X$.
  By definition of a wf-rule we know that $\lev(X_1)\leq\lev(X)$, so
  $R$ agrees with $\deA$, and $(v^i,\Omega^i)$ has a witness for
  $\lev(X)+1\leq i\leq n$. 
\item 
  Assume that $r$ is $\Rule{X}{\delta X_1X_2}$.
  Recall that $M(r)=\lev(X_1)$.
  We have 
  \begin{align*}
    \teB=(\Phi^n,\Phi^{n-1},\dots,\Phi^{\lev(X_1)+1},q_1,\deB) \in \Phi^{\lev(X_1)}
  \end{align*}
  for some
  $\deB=(X_1,\Sigma^n,\Sigma^{n-1},\dots,\Sigma^{\lev(X_1)+1},q_2)$
  and 
  \begin{align*}
    \teC=(\Sigma^n,\dots,\Sigma^{\lev(X_1)+1},
    \Psi^{\lev(X_1)},\dots,\Psi^1,q_2,\deC)\in\Psi^0
  \end{align*}
  where
  $\deC=(X_2,\Omega^n,\Omega^{n-1},\dots,\Omega^{\lev(X_2)+1},q')$.
  Since $(t^{\lev(X_1)},\teB)$ has a witness, and $(t^i,\Phi^i)$ has a
  witness for $\lev(X_1)+1\leq i\leq n$, 
  there is a run $R_2$ agreeing with $\deB$ from $R_1(1)$ to a stack
  $u^n:u^{n-1}:\dots:u^0$ such that 
  $(u^i,\Sigma^i)$ has a witness for $\lev(X_1)+1\leq i\leq n$.
  By definition of a wf-rule we know that the topmost $\lev(X_1)$-stack of
  $R_2(\lvert R_2\rvert)$ is the same as of $R_1(0)$ whence
  $u^i=s^i$ and $(u^i,\Psi^i)$ has a witness for $0\leq i\leq \lev(X_1)$. 
  In particular $(u^0,\teC)$ has a witness.
  Hence, there is a run $R_3$ agreeing with $\deC$ from
  $R_2(\lvert R_2\rvert)$ to a stack $v^n:v^{n-1}:\dots:v^0$ such that 
  $(v^i,\Omega^i)$ has a witness for $\lev(X_2)+1\leq i\leq n$.
  As $R$ we take $R_1\circ R_2\circ R_3$; this run is in $X$.
  By definition of a wf-rule we know that $\lev(X_2)\leq\lev(X)$, so
  $R$ agrees with $\deA$, and $(v^i,\Omega^i)$ has a witness for
  $\lev(X)+1\leq i\leq n$. \qed
  \end{enumerate}
\end{proof}

The next lemma shows how the set $U_z(op, K, \Sigma^K)$ 
can be used to prove that an appropriate run performing operation $op$
exists.

\begin{lemma} \label{lem:AuxiliaryU}  
  Fix a number $z\geq 1$ and assume that
  Lemma  \ref{lem:type2witness} holds for $z-1$.
  Let $s=s^n:s^{n-1}:\dots:s^0$ be a pds, where $s^0$ contains a
  link of level $K$ to a stack $u^K$, 
  and let $op$ be an operation. 
  For $0\leq i\leq n$, let $\Psi^i$ be such that $(s^i,\Psi^i)$
  has a witness; 
  let also $\Sigma^K$ be such that $(u^K,\Sigma^K)$ has a witness.
  Assume that
  \begin{align*}
    \big((\Psi^n,\Psi^{n-1},\dots,\Psi^0),
    (\Pi^n,\Pi^{n-1},\dots, \Pi^{L(op)})\big) \in U_z(r,K,\Sigma^K).
  \end{align*}
  Then $op$ can be applied to $s$, and
  $op(s)=t^n:t^{n-1}:\dots:t^{L(op)}$ is such that  
  $(t^i,\Pi^i)$ has a witness for $L(op)\leq i\leq n$.
\end{lemma}
\begin{proof}
  We proceed by case distinction on the operation $op$ performed by
  $\delta$.  
  \begin{itemize}
  \item   Assume that $op=\pop{k}$.
    Then $L(op)=k$, $\Pi^i=\Psi^i$ for $k\leq i\leq n$ and
    $\mathsf{ne}\in\Psi^k$. 
    Thus $s^k$ is not empty, so $\pop k$ can be applied to $s$, which
    results in the stack $s^n:s^{n-1}:\dots:s^k$. 
    We have $(t^i,\Pi^i)=(s^i,\Psi^i)$ for $k\leq i\leq n$ whence
    $(t^i,\Pi^i)$ has a witness.
  \item 
    Assume that $op=\col k$.
    Then $L(op)=k$, $k=K$, $\mathsf{ne}\in\Sigma^K$, 
    $\Pi^i=\Psi^i$ for $k+1\leq i\leq n$ and $\Pi^k\subseteq\Sigma^K$. 
    Thus, $u^k$ is not empty whence $\col k$ can be applied to $s$. This
    results in the stack $s^n:s^{n-1}:\dots:s^{k+1}:u^K$. 
    We have $(t^i,\Pi^i)=(s^i,\Psi^i)$ for $k+1\leq i\leq n$.
    Moreover, 
    $t^k=u^K$ and $\Pi^k\subseteq\Sigma^K$ whence $(t^i,\Pi^i)$ has a
    witness for $k\leq i\leq n$.
  \item 
    Assume that $op=\push1_{b,k}$.
    Then $L(op)=0$, $\Pi^i=\Psi^i$ for $2\leq i\leq n$, 
    $\Pi^i=\comp(\Psi^1,\Psi^0)$, and
    $\Pi^0\subseteq\stype_{z-1}(b,k,\Psi^k)$. 
    Additionally, 
    \begin{align*}
      \push1_{b,k}(s)=s^n:s^{n-1}:\dots:s^2:(s^1:s^0):(b,k,s^k).      
    \end{align*}
    For $2\leq i\leq n$, we have $(t^i,\Pi^i)=(s^i,\Psi^i)$ whence
    $(t^i,\Pi^i)$ has a witness. 
    Due to Proposition \ref{prop:low-high-uv},
    $(t^1,\Pi^1)=(s^1:s^0,\comp(\Psi^1,\Psi^0))$ has a witness. 
    Since we assumed that Lemma \ref{lem:type2witness} holds for
    $z-1$, we conclude that
    $(t^0,\stype_{z-1}(b,k,\Psi^k))$ has a witness whence
    $(t^0,\Pi^0)$ has a witness.
  \item 
    Assume that $op=\push{k}$ (for $k\geq 2$).
    Then $L(op)=0$, $\Pi^i=\Psi^i$ for $0\leq i\leq n$ with $k\neq i$,
    and 
    $\Pi^k=\comp(\Psi^k,\Psi^{k-1},\dots,\Psi^0)$. 
    Additionally, 
    \begin{align*}
    \push
    k(s)=s^n:s^{n-1}:\dots:s^{k+1}:t^k:s^{k-1}:\dots:s^0
    \end{align*}
    where $t^k=s^k:s^{k-1}:\dots:s^0$.
    For $0\leq i\leq n$ with $i\neq k$ we have
    $(t^i,\Pi^i)=(s^i,\Psi^i)$ whence $(t^i,\Pi^i)$ has a witness. 
    Due to Proposition \ref{prop:low-high-uv},
    \begin{align*}
      (t^k,\Pi^k)=(s^k:s^{k-1}:\dots:s^0,
      \comp(\Psi^k,\Psi^{k-1},\dots,\Psi^0))        
    \end{align*} 
    has a witness. \qed
  \end{itemize}
\end{proof}

\begin{proof}[Lemma \ref{lem:type2witness}]
  The proof is by induction on $z$. 
  Recall that we defined $\stype_{-1}(a, K, \Sigma^k)=\emptyset$.
  Let $z\geq 0$, $\teA\in \stype_z(a,K,\Sigma^K)$ and 
  $s^0=(a,K,u^K)$. 
  Assume that we have already proved the lemma for  $z-1$.
  By definition, $\teA\in\type(s^0)$.
  Let 
  \begin{align*}
    &\teA=(\Psi^n,\Psi^{n-1},\dots,\Psi^1,q_0,\deA),\quad\mbox{and}\\ 
    &\deA=(X,\Omega^n,\Omega^{n-1},\dots,\Omega^{\lev(X)+1},q').
  \end{align*}
  Let $c =(q_0, s^n : \dots : s^1:s^0)$ be a configuration 
  such that  $(s^i,\Psi^i)$ has a witness for each $1\leq i\leq n$.
  We have to construct a run $R$ from $c$ to a stack
  $w^n:w^{n-1}:\dots:w^0$ such that $R$  
  agrees with $\deA$ and $(w^i,\Omega^i)$ has a witness for each
  $\lev(X)+1\leq i\leq n$.  
  We distinguish  two cases.
  \begin{itemize}
  \item 
    Assume that $\teA$ is in $\stype_z(a,K,\Sigma^K)$ thanks to
    the first point of Definition \ref{def:types}. 
    Then $\Omega^i\subseteq \Psi^i$ for $\lev(X)+1\leq i\leq n$, and
    $q'=q_0$. 
    It follows that $(s^i,\Omega^i)$ has a witness (whence in particular
    $\Omega^i\subseteq\type(s^i)$) for $\lev(X)+1\leq i\leq n$ and
    the run $R$ of length $0$ from $c$ agrees with $\deA$.
  \item 
    Assume that $\teA$ is in $\stype_z(a,K,\Sigma^K)$ thanks to the
    second point of Definition \ref{def:types}. 
    Then for some rule $r=(\Rule{X}{\delta X_1 \dots
      X_m})\in\Rr^{>0}_\Xx$, where $\delta=(q_0,a,\cdot,\cdot,op)$,  
    we have
    \begin{align*}
      &\big((\Psi^n,\Psi^{n-1},\dots,\Psi^0),\\ 
      &\ (\Pi^n,\Pi^{n-1},\dots,\Pi^{L(op)})\big)\in
      U_z(op,K,\Sigma^K),\\ 
      &\big((\Psi^n,\Psi^{n-1},\dots,\Psi^0),\\
      &\ (\Pi^n,\Pi^{n-1},\dots,\Pi^{M(r)+1},\Phi^{M(r)}),\deA\big)\in
      T(r),\\ 
      &\Psi^0\subseteq\stype_{z-1}(a,K,\Sigma^K)\text{, and}\\
      &\Phi^{M(r)}=\comp(\Pi^{M(r)},\Pi^{M(r-1)},\dots,\Pi^{L(op)}).
    \end{align*}
    By induction assumption, $(s^0,\Psi^0)$ has a witness.
    Notice that the state and the topmost symbol of $c$ are as
    required by $\delta$. 
    Lemma \ref{lem:AuxiliaryU} implies that $\delta$ can be applied to
    $c$. Let $d$ be the resulting configuration
    and $t^n:t^{n-1}:\dots:t^{L(op)}$ its stack.
    Furthermore, this lemma implies that $(t^i,\Pi^i)$ has a witness
    for $L(op)\leq i\leq n$. 
    Setting $v^{M(r)}=t^{M(r)}:t^{M(r)-1}:\dots:t^{L(op)}$
    the stack of $d$ is
    $t^n:t^{n-1}:\dots:t^{M(r)+1}:v^{M(r)}$. 
    Due to Proposition \ref{prop:low-high-uv}, 
    $(v^{M(r)},\Phi^{M(r)})$ has a witness. 
    Thus, Lemma \ref{lem:AuxiliaryT} can be applied (where as $R_1$ we
    take the run from $c$ to $d$). 
    We obtain a run $R$ from $c$ which agrees with $\deA$, and
    ends in a stack $w^n:w^{n-1}:\dots:w^0$ such that $(w^i,\Omega^i)$
    has a witness for $\lev(X)+1\leq i\leq n$ as required. \qed
  \end{itemize}
\end{proof}

\begin{corollary}\label{cor:witness-k}
  Let $0\leq k\leq n$, let $s^k$ be a $k$-stack, and let
  $\Phi^k\subseteq \type(s^k)$. 
  Then $(s^k,\Phi^k)$ has a witness.
\end{corollary}

\begin{proof}
  It is enough to prove this corollary for $\Phi^k=\type(s^k)$.
  We just make an induction on the structure of the stack.
  Assume that $k=0$ and let $s^0=(a,K,t^K)$. From the induction
  assumption we know that $(t^K,\type(t^K))$ has a witness. 
  Using Corollary \ref{cor:type2witness} we obtain that $(s^0,\type(s^0))$ has
  a witness. 
  If $k>0$ and $s^k$ is empty, $\type(s^k)=\emptyset$, whence the
  claim is trivial. 
  Let now $k>0$ and let $s^k$ be nonempty.
  Decompose $s^k=t^k:t^{k-1}$. 
  By definition, $\type(s^k)=\comp(\type(t^k),\type(t^{k-1}))$.
  By induction assumption $(t^k,\type(t^k))$ and
  $(t^{k-1},\type(t^{k-1}))$ have witnesses. 
  Using Proposition \ref{prop:low-high-uv} we conclude that
  $(t^k,\type(t^k))$ has also a witness. \qed
\end{proof}

With this corollary, we can prove the right-to-left implication of
Lemma \ref{lem:run-equiv-type}. 
\begin{proof}[of Lemma \ref{lem:run-equiv-type}]
Assume that there is a 
$\teA=(\Psi^n,\Psi^{n-1},\dots,\Psi^{l+1},p,\deA)\in \type(s^l)$
such that $\Psi^i\subseteq\type(s^i)$ for $l+1\leq i\leq n$.  
Application of the corollary shows that $(s^i,\Psi^i)$ has a witness
for $l+1\leq i\leq n$, 
and $(s^l,\teA)$ has a witness. 
Thus, there is a run from
$(p,s^n:s^{n-1}:\dots:s^l)$ which agrees with $\deA$  as
required.
The other direction has already been proved (see Proof of
Lemma \ref{lem:run2type} on page \pageref{proof:run2typeA}).
\qed  
\end{proof}


\section{Runs, Positions and the History Function}
\label{app:History}
In this section we give a technical analysis of runs and introduce the
history function which is useful to describe certain sets of
runs. Appendix \ref{app:Characterisation} relies on the results
developed here.

\subsection{Positions and Histories of Stacks}
\label{sec:PosAndHist}

In this section we first introduce \emph{positions} of $i$-stacks in a
$k$-stack for $i\leq k$. These positions allow to access each substack
contained in a stack. 
Afterwards we introduce the \emph{history function}. 
Given a run $R$ and a position $x$ in the final stack of the run, this
function determines the origin of this position in the first stack of
$R$, i.e., it returns a position $y$ such that the stack at position
$x$ in the last stack of $R$ was created from the stack at position
$y$ in the first stack of $R$. 
For a $k$-stack $s$ let us denote by $\lvert s \rvert$ its \emph{size}, i.e. the number of $(k-1)$-stacks $s$ consists of.

\begin{definition}
  For each stack $s$ of level $k$ (where $1\leq k\leq \sL$)
  we define the set of positions in $s$ as follows. 

  If $k=1$, a \emph{simple position} in $s$ 
  is a number $x^1\in\N$ such that $x^1\leq \lvert s \rvert$. 
  
  If $k\geq 2$, a \emph{simple position} in $s$ is either a tuple $(0,\dots,0)$ of length $k$, or
  a tuple $(x^k,\dots,x^1)$ where $1\leq x^k\leq \lvert s \rvert$
  and $(x^{k-1},\dots,x^1)$ is a simple position in the $x^k$-th $(k-1)$-stack of $s$ (counted
  bottom up). 

  We say that a simple position $x$ \emph{points to a $k$-stack} if
  $k\in\N$ is  maximal such 
  that $x$ ends in a sequence of $0$'s of length $k$.

  A \emph{position} in $s$  is either a
  simple position in $s$ or a sequence
  $x:=x_0\posNew{k}{y}$ such that $x_0$ is a simple position pointing
  to a $0$-stack $(a,k,t^k)$ in $s$ and $y$ is a position in $t^k$,
  but $y \neq (0,\dots,0)$.\footnote{%
    We forbid nonsimple positions  ending in the simple position
    $(0,0, \dots, 0)$ because of the following interpretation.
    In a $0$-stack $s^0=(a,k,t^k)$ we actually do not consider $t^k$ to be
    a $k$-stack but only the \emph{content} of a $k$-stack. 
    In this interpretation the application of $\col{k}$ when $s^0$ is
    the topmost $0$-stack does not replace the topmost $k$-stack by
    $t^k$ but the \emph{content} of the topmost $k$-stack by the
    \emph{content} of $t^k$. This difference is only
    of syntactical nature but it is useful to exclude such positions
    when defining the history function. 
  }  
  A position $x$ \emph{points to a $k$-stack} if its rightmost simple
  position points to a $k$-stack. 

  For $x,y$ positions in $s$ we say  that $y$ \emph{points into the stack at}
  $x$ (abbreviated $y$ \emph{points into} $x$) if 
  \begin{enumerate}
  \item either $x$ points to a level $0$ stack and
    $y=x \posNew{k}{z}$ or
  \item $x$ points to a level $k\geq 1$ stack, $x$ and $y$ agree on all
    entries where $x$ is nonzero, and $y\neq x$, i.e., $y$ extends the
    position $x$ where $x$ starts to be constantly $0$. 
  \end{enumerate}
  
  Let $s$ be some $n$-stack where $s^i$ denotes the topmost $i$-stack
  of $s$. 
  The 
  \emph{position of the topmost $k$-stack of $s$} is
  $\TOP{k}(s):= (\lvert s^\sL\rvert,\dots,\lvert s^{k+1}\rvert, 0, \dots, 0)$. 

  Finally, we define the nesting rank of a position. This rank counts the number
  of simple positions involved in the position.
  Let $\NestingRank(x):=0$ if $x$ is simple, and
  $\NestingRank(x\posNew{k}{z}):=1+\NestingRank(x)+\NestingRank(z)$. 
\end{definition}
\begin{remark}
  We use the notation
  $z \posNew{k}{z'}$ where $z$ is a non-simple position of a 
  $0$-stack that links to a $k$-stack and $z'$ points to some position
  inside this linked stack.
\end{remark}

We next introduce the \emph{history function}. This function is
useful for giving semantical characterisations of the sets in the
family $\Xx$
defined by a grammar in Section \ref{sec:Runs}.
Our intuition of the history function is the following.
$\hist{R}{x}$ is the (unique) position of a $k$-stack in $s$ from
which $R$ created the $k$-stack at $x$ in $t$ in the sense that the
stack at $x$ in $t$ is a (possibly modified) copy of the stack at
$\hist{R}{x}$ in $s$ not only in terms of content but also in the
way it was produced by $R$. Here, a $\push{1}$ is understood as
copying the topmost $0$-stack and then completely replacing its
content.

\begin{definition}
  Let $R$ be a run from stack $s$ to stack $t$ and
  let $x$ be a position in $t$.  
  If $\lvert R\rvert=0$, then $\hist{R}{x}:=x$.
  If $\lvert R\rvert=1$, we make a case distinction on the operation
  performed by 
  $R$, and on the form of $x$.
  \begin{itemize}
  \item If $R$ performs a $\push{1}_{a,k}$ operation and
    $  x = \TOP{0}(t) = (x^\sL, \dots, x^1)$,      
    then 
    \begin{align*}
      \hist{R}{x}:=(x^\sL,
      \dots, x^2,x^1-1).   
    \end{align*}
  \item If $R$ performs a $\push{1}_{a,k}$ operation and $x$ is of the
    form $\TOP{0}(t)\posNew{k}{(y^{k}_1,\dots,y^1_1)}$,
    then we set 
    \begin{align*}
      \hist{R}{x}:=
      (\lvert u^\sL\rvert, \dots , \lvert u^{k+1} \rvert,
      y^{k}_1,\dots, y^1_1)       
    \end{align*} 
    for $u^i$  the 
    topmost $i$-stack of $t$.
  \item If $R$ performs a $\push{1}_{a,k}$ operation and $x$ is of the
    form $\TOP{0}(t)\posNew{k}{(y^{k}_1,\dots,y^1_1)}\posNew{k'}{z}$,
    then 
    \begin{align*}
      \hist{R}{x}:=
      (\lvert u^\sL\rvert, \dots , \lvert u^{k+1} \rvert,
        y^{k}_1,\dots, y^1_1)\posNew{k'}{z}       
    \end{align*}
    where $u^i$ is the 
    topmost $i$-stack of $t$.
  \item If $R$ performs a $\push{i}$ operation for $2\leq i \leq \sL$ and
    $x$ is of the form $(x^\sL, \dots, x^1)$ 
    such that $(x^\sL, \dots, x^1)$
    is $\TOP{i-1}(t)$ or points into $\TOP{i-1}(t)$, 
    then 
    \begin{align*}
      \hist{R}{x}:=
      (x^\sL, \dots, x^i, x^{i-1}-1, x^{i-2}, \dots,  x^1)        
    \end{align*}
  \item If $R$ performs a $\push{i}$ operation for $2\leq i \leq \sL$ and
    $x$ is of the form 
    $(x^\sL, \dots, x^1)\posNew{k}{z}$
    such that $x$ points into $\TOP{i-1}(t)$, 
    then 
    \begin{align*}
      \hist{R}{x}:=      
       (x^\sL, \dots, x^i, x^{i-1}-1, x^{i-2}, \dots,
        x^1)\posNew{k}{z}.
    \end{align*}
  \item If $R$ performs a $\col{k}$ operation and
    $x=(x_0^\sL,\dots,x_0^1)$
    points into\footnote{
      Recall that $\TOP{k}(t)$ does \emph{not} point into
      $\TOP{k}(t)$.}
    $\TOP{k}(t)$,
    then
    \begin{align*}
      \hist{R}{x} := \TOP{0}(s) \posNew{k}{(x_0^k, \dots, x_0^1)}.      
    \end{align*}
  \item If $R$ performs a $\col{k}$ operation and
    $x=(x_0^\sL, \dots, x_0^1)\posNew{k'}{y}$ 
    points into $\TOP{k}(t)$,
    then
    \begin{align*}
      \hist{R}{x} :=\TOP{0}(s)\posNew{k}{(x_0^k, \dots,
          x_0^1)}\posNew{k'}{y}.
    \end{align*}
    \item In all other cases, we set $\hist{R}{x}:=x$.
  \end{itemize}
  If $\lvert R\rvert\geq 2$, we decompose $R= S\circ T$ where
  $\lvert S\rvert = 1$, and we set
  $\hist{R}{x}:=\hist{S}{\hist{T}{x}}$. 
\end{definition}

Due to the inductive definition of the history function it is
compatible with decomposition of runs in the following sense.
\begin{proposition}\label{prop:histTransitive}
  Let $R,S,T$ be runs such that $R=S\circ T$. If $x,y$ are positions
  such that
  $\hist{T}{x}=y$, then $\hist{R}{x}=z$ if and only if $\hist{S}{y}=z$. 
\end{proposition}


\subsection{Basic Properties of Runs}
\label{app:BasicProp}
In this section we collect useful properties of runs of collapsible pushdown
systems and of the history function. 

A careful look at the definition of the history function shows that a
$k$-stack can be changed only if it is 
the topmost one, and only by an operation of level at most $k$. 
\begin{proposition}\label{prop:NonTopmostStacksDoNotChange}
  Let $R$ be a run of length $1$, and $x$ a position
  of some $k$-stack in $R(1)$. Let $t^k$ be the stack at $x$ in
  $R(1)$ and $s^k$ the stack at $\hist{R}{x}$ in $R(0)$.
  Then exactly one of the following holds.
  \begin{itemize}
  \item	$x=\TOP{k}(R(1))$ and the operation in $R$ is of level below
    $k$. In this case we have $\hist{R}{x}=\TOP{k}(R(0))$ and 
    $\lvert t^k\rvert = \lvert s^k \rvert$.
  \item	$x=\TOP{k}(R(1))$ and the operation in $R$ is of level 
    $k$. In this case we have $\hist{R}{x}=\TOP{k}(R(0))$ and 
    \begin{itemize}
    \item  $\lvert t^k\rvert =\lvert s^k \rvert-1$ if the operation
      is $\pop{k}$, 
    \item  $\lvert t^k\rvert < \lvert s^k \rvert$ if the operation
      is $\col{k}$, and 
    \item  $\lvert t^k\rvert = \lvert s^k \rvert+1$ if the operation
      is $\push{k}$.
    \end{itemize}
  \item	$s^k=t^k$ and if $y$ points into $x$, then $\hist{R}{y}$ points
    to the same position in $\hist{R}{x}$ as $y$ in $x$. 
  \end{itemize}
\end{proposition}

We have an analogous property for longer runs which follows by
straightforward induction on the length of the run.  

\begin{corollary} \label{cor:NonTopmostStacksDoNotChange}
  Let $R$ be a run of length $m$ and $x_0$ a position of some
  $k$-stack in $R(m)$ such that 
  \begin{enumerate}
  \item\label{ntsdnc-vara}	$\hist{\subrun{R}{i}{m}}{x_0} \neq \TOP{k}(R(i))$ for
    all $0\leq i<m$, or  
  \item\label{ntsdnc-varb}	$\hist{\subrun{R}{i}{m}}{x_0} \neq \TOP{k}(R(i))$ for
    all $0<i<m$, and $m\geq 2$.
  \end{enumerate}
  Then
  the $k$-stack at $\hist{R}{x_0}$ is equal to the
  $k$-stack at $x_0$ in $R(m)$. Moreover,
  if $x$ points into $x_0$, then 
  $\hist{R}{x}$ points to the same position in
  $\hist{R}{x_0}$ as $x$ in $x_0$. 
\end{corollary}

Similarly, the relationship of  
$k$-stacks that are next to each other is preserved unless the lower
one becomes the topmost stack.

\begin{proposition}\label{prop:Neihbour-position-lemma}
  Let $R$ be a run, and $x,y$ positions such that $x$ points to a
  $k$-stack that is directly below the $k$-stack to which $y$ points
  (in the same $(k+1)$-stack),  
  i.e., $x$ and $y$ differ only on the last non-zero coordinate by
  $1$ (these positions are not required to be simple). 
  Assume that
  \begin{enumerate}
  \item\label{npl-vara}	$\hist{\subrun{R}{m}{\lvert R
        \rvert}}{x}\not=\TOP k(R(m))$ for all $0\leq m\leq \lvert R\rvert$, or
  \item\label{npl-varb}	$\hist{\subrun{R}{m}{\lvert R
        \rvert}}{y}\not=\TOP k(R(m))$ for all $0\leq m<\lvert R\rvert$.
  \end{enumerate}
  Then $\hist{R}{x}$ points to a $k$-stack that is directly below the $k$-stack to which $\hist{R}{y}$ points.
\end{proposition}
\begin{proof}
  First assume that $\lvert R\rvert =1$.
  For almost every operation in $R$, the history function behaves in
  the same way for two neighbouring $k$-stacks. 
  The only exception is $\push{k+1}$  if $\hist{R}{x}=\hist{R}{y}=\TOP
  k(R(0))$. But this case is forbidden by our assumptions. 
  For $\lvert R \rvert\geq 2$, note that the claim is compatible with
  compositions of runs, so we conclude by induction on the length of
  the run. \qed
\end{proof}

Let $x$ and $y$ be positions such that $y$ points into
$x$. Intuitively, the history function should preserve this
containment because if a stack $s$ is a copy of some other stack $t$
then every stack of lower level in this stack was created from some
stack of lower level inside of $t$. The next lemma provides a formal
statement of this kind. 

\begin{proposition} \label{prop:HistPreservesPointerContainment}
  Let $R$ be some run and $x$ a position of a $k$-stack. 
  Let $y$ point into $x$ such that 
  for $l:=\NestingRank(y) - \NestingRank(x)$ the last $l$ links in $y$
  are of level at most $k$.\footnote{In other words, 
    if $x$ decomposes as $x=\hat x\posNew{k'}{x'}$ for a simple position
    $x'$, then all decompositions of $y$ as 
    $y=\hat x \posNew{k'}{y'} \posNew{k''}{y''}$ satisfy $k''\leq k$.}
  Then $\hist{R}{y}$ points into $\hist{R}{x}$ and 
  for $l':=\NestingRank(\hist{R}{y}) - \NestingRank(\hist{R}{x})$ the
  last $l'$ links in $\hist{R}{y}$ 
  are of level at most $k$.
\end{proposition}

\begin{proof}
  For $\lvert R \rvert = 0$ there is nothing to prove. 
  For $\lvert R \rvert = 1$ the claim follows directly from a 
  tedious but straightforward case distinction on the operation
  performed by $R$. 
  The general case then follows by induction: if $\lvert R \rvert \geq
  2$ we can decompose $R=R_1\circ R_2$ such that by
  induction hypothesis $y':=\hist{R_2}{y}$ points into
  $x':=\hist{R_2}{x}$ and 
  we can apply the lemma again to $R_1, x'$ and $y'$. \qed
\end{proof}

\begin{corollary} \label{cor:positionContainment}
  Let $j>k$. For every run $R$, 
  $\hist{R}{\TOP{k}(R(\lvert R \rvert))}$ points into  
   $\hist{R}{\TOP{j}(R(\lvert R \rvert))}$.
  Additionally, if $\hist{R}{\TOP{k}(R(\lvert R \rvert))}=\TOP k(R(0))$ then $\hist{R}{\TOP{j}(R(\lvert R \rvert))}=\TOP j(R(0))$.
\end{corollary}

According to our intuition that the history function tells us the
original copy from which a stack was created, history can only
decrease a position (with respect to the lexicographic ordering $\lexOrd$). 
On the other hand, when a position is always present in a stack, the
history should point to the same position. 
The next two lemmas prove this intuition.

\begin{lemma} \label{lem:PositionsDecrease}
  Let $R$ be a run and $x$ a position in the final stack of $R$ such
  that $x_0$ is the simple prefix of $x$, i.e., $x_0$ is a simple
  position such that there is a position $x'$ with $x=x_0\posNew{k}{x'}$. 
  If $y:=\hist{R}{x}$ is a simple position,
   $y \lexOrd x_0$.
\end{lemma}

\begin{lemma}
  \label{lem:TopPresentImpliesHistoryIsId}
  Let $R$ be a run and let $x:=\TOP{k}(R(0))$. 
  If $x$ is present in all configurations of $R$ then 
  \begin{enumerate}
  \item $\hist{R}{x}=x$, and
  \item  for all $i\leq \lvert R \rvert$  
    $\hist{\subrun{R}{i}{\lvert R \rvert}}{x}$ is 
    simple if and only if
    $\hist{\subrun{R}{i}{\lvert R \rvert}}{x}=x$.
  \end{enumerate}
\end{lemma}

In the rest of this section we prove these two Lemmas.
For the proofs we use functions
$\mathsf{pack}_i$.\label{def:pack} 
Let $x=x_0 \posNew{k_1}{x_1}\dots\posNew{k_m}{x_m}$ with
$m=\NestingRank(x)$. 
For $0\leq i\leq m$ we define a simple position $\mathsf{pack}_i(x)$
as follows.
\begin{itemize}
\item	$\mathsf{pack}_0(x):=x_0$ and 
\item	for $i\geq 1$, $\mathsf{pack}_i(x)$ is obtained from
  $\mathsf{pack}_{i-1}(x)$ by replacing its last $k_i$ coordinates by
  $x_i$. 
\end{itemize}
Note that $\pack_1$ is closely related to the $\col{k}$ and
the $\push{1}_{a,k}$ operations:
if $R$ is a run of length $1$ performing $\col{k}$ and
$x$ points into the topmost $k$-stack of $R(1)$, then 
$\pack_{i+1}(\hist{R}{x})=\pack_{i}(x)$
for all $0\leq i\leq \NestingRank(x)$. On the other hand,
if $R$ is of length $1$ performing $\push{1}_{a,k}$ and
$x=\TOP{0}(R(1)) \posNew{k}{x_1}$, then
$\pack_{i}(x)=\pack_{i-1}(\hist{R}{x})$ for all $1\leq i \leq \NestingRank(x)$.

In the following, for a simple position $z=(z^n, z^{n-1}, \dots, z^1)$
we call $z^k$ the 
\emph{level $k$ coordinate} of $z$.
Furthermore, we write $x\lexOrdstrict_k y$ for simple positions $x$,
$y$ if  $x\lexOrdstrict y$ and the first coordinate on which they
differ is the level $k$ coordinate. 

\begin{lemma}\label{lem:pack-is-good}
  Let $R$ be a run and $x=x_0 \posNew{k_1}{x_1}\dots\posNew{k_m}{x_m}$ a position in the final stack of $R$. 
  Assume that $y:=\hist{R}{x}$ is a simple position.
  Then $$y\lexOrd \mathsf{pack}_m(x)\lexOrdstrict_{k_m}\mathsf{pack}_{m-1}(x)\lexOrdstrict_{k_{m-1}}\dots\lexOrdstrict_{k_1}\mathsf{pack}_0(x).$$
\end{lemma}

\begin{proof}
  The proof is by induction of the length of $R$.
  If $\lvert R \rvert = 0$ the claim is trivial (as $x=y$ and
  $\NestingRank(x)=0$).
  For $\lvert R \rvert \geq 1$, let $S=\subrun{R}{\lvert R \rvert
    -1}{\lvert R \rvert}$, and let $z=\hist{S}{x}$. 
  The induction assumption, applied for $\subrun{R}{0}{\lvert R \rvert-1}$ and for
  $z=z_0 \posNew{k'_1}{z_1}\dots\posNew{k'_{m'}}{z_{m'}}$, 
  gives us that 
  $$y\lexOrd \pack_{m'}(z) \lexOrdstrict_{k_{m'}} \pack_{m'-1}(z)
  \lexOrdstrict_{k_{m'-1}} \dots \lexOrdstrict_1 \pack_0(z).$$
  We analyse the cases of the definition of the history function, for run $S$.
  \begin{itemize}
  \item	If $S$ performs a $\push 1_{a,k}$ operation and
    $x=\TOP 0(S(1))$, then $m =m' = 0$ and
    $x=\pack_0(x) \lexOrdR z = \pack_0(z) \lexOrdR y$. 
  \item	If $S$ performs a $\push 1_{a,k}$ operation, and $x$ points
    into $\TOP 0(S(1))$ 
    we already remarked that
    $\mathsf{pack}_i(x)=\mathsf{pack}_{i-1}(z)$ and $m=m'+1$. 
    Thus, we immediately conclude that
    $$y\lexOrd \pack_{m}(x) \lexOrdstrict_{k_m} \pack_{m-1}(x) \dots
    \lexOrdstrict_{k_2} \pack_1(x).$$ 
    Note that for $s^k$ the topmost $k$-stack of $S(0)$, $x_1$ points
    to a position in $\pop{k}(s^k)$ while $x_0$ points into
    $\TOP{k}(s)$. Thus, the level $k$ coordinate of $\pack_0(x)$ is
    $\lvert s^k\rvert$  while the corresponding coordinate in
    $\pack_1(x)$ has 
    value at most $\lvert s^k\rvert -1$. Thus,
    $\pack_1(x)\lexOrdstrict_{k} \pack_0(x)$. 
  \item	If $S$ performs a $\push i$ operation and $x$ is
    $\TOP{i-1}(S(1))$ or points into $\TOP{i-1}(S(1))$,
    then every coordinate of $z_0$ is either the same or smaller than
    the same coordinate of $x_0$ (and the rest of $x$ and $z$ is the
    same). 
    Thus, we conclude immediately from the properties of  $z$ that $x$
    also satisfies the claim. 
  \item	If $S$ performs a $\col k$ operation, and $x$ points
    into the topmost $k$-stack of $S(1)$. 
    Then, 
    $\mathsf{pack}_i(x)=\mathsf{pack}_{i+1}(z)\lexOrdR y$ for all
    $i\leq m$, and $m=m'-1$. Thus, the claim follows trivially. 
  \item	If none of the previous cases applies, then $z=x$ and there is
    nothing to show.\qed
  \end{itemize}
\end{proof}

From the previous lemma we can easily deduce Lemma
\ref{lem:PositionsDecrease}.

\begin{proof}[Lemma \ref{lem:PositionsDecrease}] 
  Due to Lemma \ref{lem:pack-is-good}, 
  $y\lexOrd \mathsf{pack}_{\NestingRank(x)}(x)\lexOrd
  \mathsf{pack}_0(x) = x_0$.\qed
\end{proof}

We also obtain the following corollary of Lemma \ref{lem:pack-is-good}.

\begin{corollary} \label{cor:GoodStacksHaveSmallPositions}
  Let $s$ be some pds and $0\leq
  k\leq n$.
  If $x$ is a position in $s$ such that
  $\mathsf{pack}_{\NestingRank(x)}(x)$ points to a $k$-stack and 
  $\mathsf{pack}_{\NestingRank(x)}(x)\lexOrdR\TOP k(s)$,
  then $x=\TOP k(s)$.
\end{corollary}

\begin{proof}
  Decompose $x = x_0 \posNew{k_1}{x_1}\dots\posNew{k_m}{x_m}$ and
  $z:=\mathsf{pack}_m(x)$. 
  Consider any $n$-CPS such that there is a run $R$ from the initial
  configuration $(q_0, \bot_\sL)$ to the pds $s$ (recall that 
  such a run exists by
  definition of a pds).  
  Since $\bot_\sL$ only contains simple positions, 
  $y:=\hist{R}{x}$ is simple.
  Application of Lemma \ref{lem:pack-is-good} gives us $\TOP
  k(s)\lexOrd z\lexOrd x_0$, 
  so $x_0$ points into the topmost $k$-stack of $s$.
  This implies that coordinates of levels greater than $k$ of $\TOP
  k(s)$, $z$, and $x_0$ agree whence $z=\TOP k(s)$. 
  
  If $m=0$, $x=z=\TOP{k}(s)$ and we are done.
  Heading towards a contradiction assume that $m\geq 1$.
  Let $j$ be the maximum of all $k_i$.
  The last $k$ coordinates of $z=\TOP k(s)$ are $0$.
  Since the last $k_m$ coordinates of $z$ are $x_m\neq(0,\dots,0)$ (by definition of a position), 
  we see that $j\geq k_m > k$.
  Thus $z$ points to or into $\TOP{j-1}(s)$, whence its level $j$ coordinate is
  the size of the topmost $j$-stack of $s$. 
  Since each application of $\pack$ preserves the  
  coordinates of level above $j$, $x$ also points into $\TOP{j}(s)$ whence
  its level $j$ coordinate is bounded by the size of the topmost
  $j$-stack of $s$. 
  Now Lemma \ref{lem:pack-is-good} implies that 
  the level $j$
  coordinate in $\mathsf{pack}_{\NestingRank(x)}(x)$ is smaller  than
  that in $x$ which is a contradiction. \qed
\end{proof}

Finally, we prepare the proof of
Lemma \ref{lem:TopPresentImpliesHistoryIsId} with the following lemma.

\begin{lemma}\label{lem:x-present-pack-preserved}
  Let $R$ be a run, and let $x$ be a position of $R(\lvert R\rvert)$, and let $y:=\hist{R}{x}$.
  If  $\mathsf{pack}_{\NestingRank(x)}(x)$ is present in all configurations of $R$,
  then $\mathsf{pack}_{\NestingRank(y)}(y)=\mathsf{pack}_{\NestingRank(x)}(x)$.
\end{lemma}

\begin{proof}[Lemma \ref{lem:x-present-pack-preserved}]
  If we prove the lemma for runs of length $1$, the whole claim
  follows by a simple induction on the length of a run. 
  Let $R$ be a run of length $1$. The proof is by case distinction on
  the definition of the history function.
  \begin{itemize}
  \item	Assume that $R$ performs a $\push 1_{a,k}$ operation, and $x=\TOP 0(R(1))$. 
    Then $x=pack_{\NestingRank(x)}$ is not present in $R(0)$ whence
    there is nothing to show.
  \item	Assume that $R$ performs a $\push j$ operation, and $x$ is
    $\TOP {j-1}(R(1))$ or points into $\TOP{j-1}(R(1))$. 
    Let $x=x_0 \posNew{k_1}{x_1}\dots\posNew{k_m}{x_m}$.
    If $k_i<j$ for all $i$, also $\mathsf{pack}_{\NestingRank(x)}(x)$ is $\TOP {j-1}(R(1))$ or points into $\TOP{j-1}(R(1))$
    (as then $x$ and $\mathsf{pack}_{\NestingRank(x)}(x)$ are equal on all coordinates of level at least $j$).
    But this would mean that $\mathsf{pack}_{\NestingRank(x)}(x)$ was not present in $R(0)$;
    thus $k_i\geq j$ for some $i$ (in particular $m\geq 1$).
    Notice that $y=y_0\posNew{k_1}{x_1}\dots\posNew{k_m}{x_m}$, where
    $y_0$ differs from $x_0$ only on the level $j$ coordinate. 
    This coordinate does not appear  in $\mathsf{pack}_m(x)$ whence
    $\mathsf{pack}_{\NestingRank(y)}(y)=\mathsf{pack}_{\NestingRank(x)}(x)$. 
  \item	In the remaining three cases we easily see (from the
    definition of the history) that
    $\mathsf{pack}_{\NestingRank(y)}(y)=\mathsf{pack}_{\NestingRank(x)}(x)$. \qed
  \end{itemize}
\end{proof}

\begin{proof}[Lemma \ref{lem:TopPresentImpliesHistoryIsId}]
  For a simple position $x=\TOP{k}(R(0))$ we have
  $\mathsf{pack}_{\NestingRank(x)}(x)=x$. 
  For arbitrary $i\leq \lvert R \rvert$, let
  $y_i:=\hist{\subrun{R}{i}{\lvert R \rvert}}{x}$.
  We apply Lemma \ref{lem:x-present-pack-preserved} for
  $\subrun{R}{i}{\lvert R\rvert}$ and obtain
  $\mathsf{pack}_{\NestingRank(y_i)}(y_i)=\TOP{k}(R(0))$. 
  If $y_i$ is simple, this implies $y=x$.
  Corollary \ref{cor:GoodStacksHaveSmallPositions} implies that $y_0=x$.\qed
\end{proof}


\section{A Family of Sets of Runs}
\label{app:Characterisation}
In this appendix we prove that the sets defined in Section
\ref{sec:Runs} satisfy the (informal) claims we formulated. 
In fact, our proof goes from intuition (which is made precise using
the history function) to grammars: 
First, using the history function we give alternative definitions of
pumping runs, 
$\TOP{k}$-non-erasing runs, $k$-returns and $k$-colreturns
and show that they satisfy the intuition given before. We then
show that these runs are actually described by the grammars we
presented in Section \ref{sec:Runs}.

\subsection{Characterisation of Returns and Colreturns}

We start with a definition of returns. 
In Lemmas \ref{lem:return-wfrules} and \ref{lem:Change-level-Return}
we later see that the  grammar from Section \ref{sec:Runs} 
correctly describes the sets of returns.

\begin{definition}\label{def:return}
  A run $R$ of length $m$ is called \emph{$k$-return} (where $1\leq
  k\leq \sL$) if 
  \begin{itemize}
  \item $\hist{R}{\TOP{k-1}(R(m))}$ points to the second topmost
    $(k-1)$-stack\footnote{
    	Whenever we write ``the second topmost $(k-1)$-stack'' we
        assume that it is in the same $k$-stack as the topmost
        $(k-1)$-stack, i.e., we assume that the topmost $k$-stack has
        size at least $2$. 
    } in the topmost $k$-stack of $R(0)$, and
  \item $\hist{\subrun{R}{i}{m}}{\TOP{k-1}(R(m))}\neq\TOP{k-1}(R(i))$
    for all $1\leq i \leq m-1$. 
  \end{itemize}
\end{definition}

The following propositions confirm our intuition about $k$-returns.
\begin{proposition}\label{prop:return-ends-k}
  The last operation of a $k$-return $R$ is $\pop k$ or $\col k$.
\end{proposition}
\begin{proof}
  Let $m:=\lvert R \rvert$.
  Note that in order to satisfy 
  \begin{align*}
    \hist{\subrun{R}{m-1}{m}}{\TOP{k-1}(R(m))}\neq\TOP{k-1}(R(m-1))    
  \end{align*}
  the
  last operation of $R$ is $\pop{j}$ or $\col{j}$ with $j\geq k$. 
  Heading for a contradiction assume that $j>k$.
  It follows that
  \begin{align*}
    \hist{\subrun{R}{m-1}{m}}{\TOP{k}(R(m))}\neq \TOP{k}(R(m-1)).     
  \end{align*}
  Let $i\leq m-2$ be maximal such that 
  \begin{align*}
    x_{i}:=\hist{\subrun{R}{i}{m}}{\TOP k(R(m))}=\TOP k(R(i))    
  \end{align*}
  Such $i$ exists because $i=0$ is of this form (cf.\ Corollary
  \ref{cor:positionContainment}). 
  Due to Corollary \ref{cor:NonTopmostStacksDoNotChange}
  (variant \ref{ntsdnc-varb}),  
  $\hist{\subrun{R}{i}{m}}{\TOP{k-1}(R(m))}$ 
  points to the topmost $(k-1)$-stack of the $k$-stack to which
  $x_{i}$ points. Thus, it points to the topmost $(k-1)$-stack of $R(i)$. 
  But this contradicts the definition of a return.\qed
\end{proof}

\begin{proposition}\label{prop:return-removes}
  For every $k$-return $R$, the topmost $k$-stack of $R(0)$ after
  removing its topmost $(k-1)$-stack is equal to the topmost $k$-stack
  of $R(|R|)$. 
  If $x$ points into $\TOP{k}(R(\lvert R \rvert)$ then 
  $\hist{R}{x}$ points to the same position in the stack at
  $\TOP{k}(R(0))$. 
\end{proposition}
\begin{proof}
  Let $m:=\lvert R \rvert$ and
  let $l$ be the size of the topmost $k$-stack of $R(m)$. For
  $1\leq i\leq l$, 
  let $t^{k-1}_i$ be the $i$-th $(k-1)$-stack (counting from the
  bottom) of the topmost $k$-stack 
  of $R(m)$. Let $x_i$ be the position pointing to $t^{k-1}_i$. 
  By the above proposition, $R$ ends in $\pop{k}$ or $\col{k}$.
  Note that the histories of $x_1, \dots, x_l$ with respect to
  $\subrun{R}{m-1}{m}$ point to $(k-1)$-stacks, where the history of
  $x_1$ points to the bottommost $(k-1)$-stack of some $k$-stack, 
  $x_{i}$ is directly below the history of $x_{i+1}$ 
  for each $1\leq i< l$ and none of these histories point to the
  topmost $(k-1)$-stack  of $R(m-1)$. 
  Note that non-topmost $(k-1)$-stacks in the same $k$-stack are always
  treated the same way by the history function. Thus, 
  a simple induction on the operations performed by
  $\subrun{R}{0}{m-1}$ shows that this property is preserved by the
  history function, i.e.,
  $\hist{R}{x_1}, \hist{R}{x_2}, \dots, \hist{R}{x_l}$ point to the
  first $l$ $(k-1)$-stacks of a $k$-stack. But by definition
  $\hist{R}{x_l}=\hist{R}{\TOP{k-1}(R(m))}$ is the second topmost
  $(k-1)$-stack of $\TOP{k}(R(0))$. 
  Application of Corollary \ref{cor:NonTopmostStacksDoNotChange} (variant \ref{ntsdnc-vara}) shows
  that the $(k-1)$-stack at $x_i$ in $R(m)$ is 
  the same as the $(k-1)$-stack at $\hist{R}{x_i}$ in $R(0)$. 
  This proves the first part of this proposition.
  Similarly, Corollary \ref{cor:NonTopmostStacksDoNotChange} implies the
  preservation of pointers into 
  $\TOP{k}(R(m))$ which completes the proof.\qed
\end{proof}

\begin{corollary}\label{cor:push+return}
  For every run $R$ which starts with a $\push k$ operation (including
  arbitrary $\push 1_{a,l}$ for $k=1$), and continues with a
  $k$-return, the topmost $k$-stacks of $R(0)$ and of $R(|R|)$ coincide. 
  Additionally, if $x$ points into $\TOP{k}(R(\lvert R \rvert)$ then
  $\hist{R}{x}$ points to the same position in the stack at
  $\TOP{k}(R(0))$.  
\end{corollary}

Recall that wf-rules of the form $\Rule{X}{\delta YZ}$ 
must satisfy
 the property that 
whenever $R$ is a composition of a one-step run performing transition
$\delta$ with a run from $Y$, then the topmost $\lev(Y)$-stack of
$R(0)$ and $R(\lvert R \rvert)$ are the same. 
Notice that in the grammars in Section \ref{sec:Runs} such rules
appear only when $\delta$ performs a $\push{}$ operation of some level
$k$, and $Y$ is a set of $k$-returns. 
Since $\lev(X)=k$, the above corollary proves this property.

We now give a definition of $k$-colreturns. 
Lemmas \ref{lem:colreturn-wfrules} and \ref{lem:Change-level-Return} show that, for such definitions, the grammar from Section \ref{sec:Runs}
correctly describes the sets of colreturns.
As already mentioned, the intuition of the definition is the following.
A $k$-colreturn is a run whose
last transition is $\col{k}$ from a stack where the topmost symbol is
a copy of the topmost symbol of the first stack. 

\begin{definition}
  A run $R$ of length $m$ is called \emph{$k$-colreturn} (where $1\leq
  k\leq n$) if 
  \begin{itemize}
  \item $\hist{R}{\TOP{k-1}(R(m))}$ is of the form 
    $\TOP{0}(R(0))\posNew{k}{x}$, where $x$ is simple, and
  \item $\hist{\subrun{R}{i}{m}}{\TOP{k-1}(R(m))}$ is not simple for
    all $0\leq i\leq m-1$. 
  \end{itemize}
\end{definition}

We first prove a decomposition result of $k$-returns and
$k$-colreturns into one transition followed by a sequence of shorter
returns or colreturns.
Later we deal with the change levels.

\begin{lemma}\label{lem:return-wfrules}
    Let $R$ be some run. 
    Then $R$ is a $k$-return if and only if $R$ is of one of the
    following forms. 
    \begin{enumerate}
    \item $\lvert R\rvert=1$ and $R$ performs $\pop{k}$.
    \item $R$ starts with an operation of level at most $k-1$, and
      continues with a $k$-return. 
    \item $R$ starts with $\push1_{a,k}$ and continues with a $k$-colreturn.
    \item $R$ starts with a $\push{j}$ for $j>k$ and continues with a $k$-return.
    \item $R$ starts with a $\push{j}$ for $j\geq k$ (including $\push1_{a,l}$ for $j=k=1$) and decomposes as
      $R=S \circ T \circ U$ where $S$ has length $1$, $T$ is a
      $j$-return and $U$ is a $k$-return.
    \end{enumerate}
\end{lemma}

\begin{lemma}\label{lem:colreturn-wfrules}
    Let $R$ be some run. 
    Then $R$ is a $k$-colreturn if and only if $R$ is of one of the
    following forms
    \begin{enumerate}
    \item $R$ has length $1$ and performs $\col{k}$.
    \item $R$ starts with $\push{j}$ for some $j\geq 2$ and continues
      with a $k$-colreturn.
    \item $R$ starts with a $\push{j}$ (including $\push1_{a,l}$ for
      $j=1$) and decomposes as 
      $R=S\circ T\circ U$ where $S$ has length $1$, $T$ is a
      $j$-return and $U$ is a $k$-colreturn.
    \end{enumerate}
\end{lemma}

Before we start the proof of these two lemmas, we state some auxiliary claims.
First, we observe that the history function either manipulates the simple
prefix of a position, or adds a simple prefix, or removes it.
This will be useful while analysing $k$-colreturns.

\begin{proposition} \label{prop:histPreservesLinkNesting}
  Let $R$ be some run of length $m$.
  Let $x \posNew{k}{y}$ be a position in $R(m)$
  such that 
  \begin{align*}
    \NestingRank(\hist{\subrun{R}{i}{m}}{x\posNew{k}{y}}) >
    \NestingRank(\hist{R}{y})    
  \end{align*}
  for all $i\leq m$. 
  Let $x'=\hist{R}{x}$.
  Then $\hist{R}{x\posNew{k}{y}}=x'\posNew{k}{y}$ (neither $x$ nor $x'$
  have to be simple). 
  Additionally, the $0$-stack of $R(m)$ at position $x$ is the same as
  the $0$-stack of $R(0)$ at position $x'$.  
\end{proposition}
\begin{proof}
  Induction on $m$.
  For $m=1$ we just analyse all cases.
  For $m\geq 2$ we observe that the claim for any decomposition
  $R= S \circ T$ follows from 
  the claim  for $S$ and for $T$.\qed
\end{proof}

\begin{corollary} \label{cor:PushColStackdifference}
  Let $R$ be a run of length $m \geq 2$, and $x$ a simple position in $R(m)$ such that $\hist{R}{x}$ is simple,
  but $\hist{\subrun{R}{i}{m}}{x}$ is not simple for $1\leq i\leq m-1$.
  Then, for some $k$, $R$ starts with $\push{1}_{a,k}$ and ends with $\col{k}$.
  Additionally, $x=\TOP{k-1}(R(m))$ if and only if $\hist{R}{x}$ is the second topmost $(k-1)$-stack of $R(0)$.
\end{corollary}
\begin{proof}
  The last operation of $R$ has to be a $\col k$ for some $k$, because
  otherwise $\hist{\subrun{R}{m-1}{m}}{x}$ would be simple. 
  Then $\hist{\subrun{R}{m-1}{m}}{x}$ is of the form $\TOP
  0(R(m-1))\posNew{k}{x'}$. 
  Proposition \ref{prop:histPreservesLinkNesting}, applied for
  $\subrun{R}{1}{m-1}$,  
  shows that $\hist{\subrun{R}{1}{m}}{x}$ is of the form $z\posNew{k}{x'}$,
  and that the topmost $0$-stack of $R(m-1)$, and the $0$-stack in
  $R(1)$ to which $z$ points to are the same $k$-stack $u^k$, which is
  in fact the topmost $k$-stack of $R(m)$. 
  Because $\hist{\subrun{R}{0}{1}}{z\posNew{k}{x'}}$ is simple,
  necessarily $z=\TOP0(R(1))$, the first operation of $R(0)$ is
  $\push1_{a,k}$, and  
  the topmost $k$-stack of $R(0)$ after removing its topmost $(k-1)$-stack is equal to $u^k$.
  Additionally, $x$ points to the same position in the topmost
  $k$-stack of $R(m)$, as $\hist{R}{x}$ in the topmost $k$-stack of
  $R(0)$. 
  Thus $x=\TOP{k-1}(R(m))$ if and only if $\hist{R}{x}$ is the second topmost $k$-stack of $R(0)$.\qed
\end{proof}

The following lemma proves the intuition that $k$-colreturns make a
copy of the topmost stack symbol and finally use its collapse link of
level $k$ (the proof is almost the same as that of the previous corollary).
\begin{lemma}\label{lem:ColreturnDesreases}
  Let $R$ be a colreturn.
  Then the topmost $k$-stack of $R(|R|)$ is equal to the $k$-stack contained in the topmost $0$-stack of $R(0)$.
  In particular its size is smaller than the size of the topmost $k$-stack of $R(0)$.
  Additionally, the last operation of $R$ is $\col k$, and
  $\hist{R}{\TOP k(R(|R|))}=\TOP k(R(0))$.
\end{lemma}
\begin{proof}
  Let $m:=|R|$, and let $x:=\TOP{k-1}(R(m))$.
  The last operation of $R$ has to be a $\col j$ for some $j$, because
  otherwise $\hist{\subrun{R}{m-1}{m}}{x}$ would be simple. 
  Then $\hist{\subrun{R}{m-1}{m}}{x}$ is of the form $\TOP
  0(R(m-1))\posNew{j}{x'}$. 
  Proposition \ref{prop:histPreservesLinkNesting}, applied for
  $\subrun{R}{0}{m-1}$, implies that
  $\hist{R}{x}=\hist{\subrun{R}{0}{m-1}}{\TOP0(R(m-1))}\posNew{j}{x'}$. 
  By definition of a $k$-colreturn it follows that $j=k$ and
  $\hist{\subrun{R}{0}{m-1}}{\TOP0(R(m-1))}=\TOP0(R(0))$. 
  From this proposition we also conclude that the topmost $0$-stack of
  $R(0)$ and the topmost $0$-stack of $R(m-1)$ store the same
  $k$-stack $u^k$, which is in fact the topmost $k$-stack of $R(m)$. 
  Of course the size of $u^k$ is smaller than the size of the the
  topmost $k$-stack of $R(0)$ 
  because $\col k$ must decrease the size of the topmost $k$-stack
  (see Remark \ref{rem:ColisPop}). 
  Corollary \ref{cor:positionContainment} implies that
  $\hist{\subrun{R}{0}{m-1}}{\TOP k(R(m-1))}=\TOP k(R(0))$. 
  Since the last operation is $\col k$, 
  $\hist{R}{\TOP k(R(m-1))}=\TOP k(R(0))$. \qed
\end{proof}

The next two propositions describe which operations are allowed as the
first operation of a $k$-return and of a $k$-colreturn. 

\begin{proposition}
  \label{prop:ReturnsStartwithPushOrSmallLevel}
  Let $R$ be a $k$-return.
  The first operation of $R$ is neither $\col{j}$ for $j\geq k$ nor
  $\pop{j}$ for $j>k$. 
  If the first operation of $R$ is $\pop{k}$, then $\lvert R \rvert=1$.
\end{proposition}
\begin{proof}
  Let $m:=|R|$.
  If the first operation of a run $R$ is $\col{j}$, $j\geq k$ then by
  definition of the history function $\hist{\subrun{R}{0}{1}}{x}$ does
  not point to any simple position inside $\TOP{j}(R(0))$ for all positions
  $x$ in $R(1)$.  
  Thus, also $\hist{R}{x}$ does not point to any simple position inside
  $\TOP{j}(R(0))$ for all positions $x$ in $R(m)$. 
  But if $R$ is a $k$-return, $\hist{R}{\TOP{k-1}(R(m))}$ is a simple position and points
  into $\TOP{k}(R(0))$ whence it also points into $\TOP{j}(R(0))$. 
  Analogously, one shows that $R$ does not start with $\pop{j}$ for $j>k$. 

  If a $k$-return $R$ starts with $\pop{k}$, it follows that
  $\hist{\subrun{R}{1}{m}}{\TOP{k-1}(R(m))}=\TOP{k-1}(R(1))$. 
  But this is not allowed if $1\leq m-1$. \qed
\end{proof}

\begin{proposition}
  \label{prop:ColreturnsStartwithPush}
  Let $R$ be a $k$-colreturn.
  The first operation of $R$ is a $\push{}$ or $\col k$.
  If the first operation of $R$ is $\col{k}$, then $\lvert R \rvert=1$.
\end{proposition}
\begin{proof}
  Let $m:=|R|$.
  If the first operation of a run $R$ is a $\pop{}$ then by
  definition of the history function $\hist{\subrun{R}{0}{1}}{x}$ does
  not point into $\TOP{0}(R(0))$ for all positions $x$ in $R(1)$.  
  Thus also $\hist{R}{x}$ does not point into $\TOP{0}(R(0))$ for all
  positions $x$ in $R(m)$, in particular for
  $x=\hist{R}{\TOP{k-1}(R(m))}$. This contradicts the definition of a
  $k$-colreturn.

  In $R(0)$ we have a position $\TOP{0}(R(0))\posNew{k}{x}$. Thus,  the
  only collapse operation which can be performed at $R(0)$ is a level
  $k$ collapse, i.e., $\col k$. 
  If $R$ starts with $\col k$, then
  $\hist{\subrun{R}{0}{1}}{y}=\TOP{0}(R(0))\posNew{k}{x}$ for some simple
  $x$ only if $y$ is simple. 
  We conclude that $\hist{\subrun{R}{1}{m}}{\TOP{k-1}(R(m))}$ is simple
  which implies $m=1$.\qed
\end{proof}

We state a last auxiliary lemma and then prove Lemmas
\ref{lem:return-wfrules} and \ref{lem:colreturn-wfrules}.

\begin{lemma} \label{lem:SecondReturnHistory}
  Let $R$ be a $k$-return of length at least $2$. 
  Let  $x$ be the position $\hist{\subrun{R}{1}{\lvert R
      \rvert}}{\TOP{k-1}(R(\lvert R \rvert))}$. Then  one of the
  following holds. 
  \begin{enumerate}
  \item $x$ points to the second topmost $(k-1)$-stack of $R(1)$ and
    the first stack operation of $R$ is of level strictly below $k$ or
    a $\push{j}$ for $j>k$.
  \item $x$ points to the third topmost $(k-1)$-stack of $R(1)$ and
    the first stack operation of $R$ is a push of level $k$.
  \item There is a $j>k$ such that $x$ points to the second topmost
    $(k-1)$-stack of the second 
    topmost $(j-1)$-stack of $R(1)$ and the first stack operation of
    $R$ is $\push{j}$.
  \item $x=\TOP0(R(1))\posNew{k}{\TOP{k-1}(u^k)}$, the first stack
    operation of $R$ is $\push{1}_{a,k}$ and the topmost $0$-stack of
    $R(1)$ is $(a,k,u^k)$. 
  \end{enumerate}
\end{lemma}
\begin{proof}
  Since $R$ is a $k$-return, Proposition \ref{prop:histTransitive} 
  implies that
  $\hist{\subrun{R}{0}{1}}{x}$ points to the second topmost
  $(k-1)$-stack of $R(0)$. 
  We proceed by case distinction on the stack operation of
  $S:=\subrun{R}{0}{1}$. 
  Due to Proposition \ref{prop:ReturnsStartwithPushOrSmallLevel} we
  only have to consider the following cases. 
  \begin{itemize}
  \item 
    Assume that $S$ performs a $\pop{j}$ operation for $j<k$, or a
    $\col{j}$ operation for $j<k$, or a $\push{j}$ operation for
    $2\leq j<k$, or a $\push1_{a,j}$ operation for $j\neq k>1$.
    Then $x$ necessarily points to the second topmost $(k-1)$-stack of
    $R(1)$ (because $S$ makes changes only inside the topmost
    $(k-1)$-stack of $R(0)$). 
  \item
    Assume that $S$ performs a $\push{k}$ for $k\geq 2$, or a
    $\push1_{a,j}$ for $j\neq k=1$.
    Then $x$ necessarily points to the third topmost $(k-1)$-stack of $R(1)$.
  \item
    Assume that $S$ performs a $\push{j}$ operation with $j > k$.
    Then either $x$ points to the second topmost $(k-1)$-stack of
    $R(1)$, or to the second topmost $(k-1)$-stack of the second
    topmost $(j-1)$-stack of $R(1)$. 
  \item
    Assume that $S$ performs $\push{1}_{a,k}$, and $k\geq 2$.
    Then either $x$ points to the second topmost $(k-1)$-stack of $R(1)$, 
    or $x$ is of the form $\TOP0(R(1))\posNew{k}{\TOP{k-1}(u^k)}$
    where $u^k$ is the $k$-stack stored in the topmost $0$-stack of
    $R(1)$. 
  \item
    Assume that $S$ performs $\push{1}_{a,k}$, and $k=1$.
    Then either $x$ points to the third topmost $(k-1)$-stack of $R(1)$,
    or $x$ is of the form $\TOP0(R(1))\posNew{k}{\TOP{k-1}(u^k)}$
    where $u^k$ is the $k$-stack stored in the topmost $0$-stack of
    $R(1)$. \qed
  \end{itemize}
\end{proof}

Now we are prepared to prove Lemma \ref{lem:return-wfrules}, i.e., the
decomposition of returns into one transition followed by shorter
returns or colreturns.

\begin{proof}[of Lemma \ref{lem:return-wfrules}]
  We first show that every return decomposes as required.
  Let $R$ be a $k$-return of length $m$ (by definition $m\geq 1$). Set
  $S:=\subrun{R}{0}{1}$. 
  When $S$ performs a $\pop{k}$ operation, and $m=1$, we immediately
  get case one. 
  Otherwise, 
  we proceed by distinction of the cases of Lemma
  \ref{lem:SecondReturnHistory} for the position
  $x:=\hist{\subrun{R}{1}{|R|}}{\TOP{k-1}(R(m))}$. 
  \begin{itemize}
  \item
    Assume that $x$ points to the second topmost $(k-1)$-stack of $R(1)$.
    Then $\subrun{R}{1}{m}$ is easily seen to be a $k$-return; we get
    case 2 or case 4. 
  \item
    Assume that $x$ points to the third topmost $(k-1)$-stack of $R(1)$.
    Then the operation in $S$ was $\push{k}$ (or $\push{1}_{a,j}$ for $k=1$). 
    \begin{claim}
      There is  some $1<i<m$ such  that
      $\hist{\subrun{R}{i}{m}}{\TOP{k-1}(R(m))}$ points 
      to the second topmost $(k-1)$-stack of $R(i)$.       
    \end{claim}
    Under the assumption that this claim holds, 
    choose the minimal such $i$.
    From the choice of $i$ it follows immediately that
    $U:=\subrun{R}{i}{m}$ is a $k$-return.
    We show that $T:=\subrun{R}{1}{i}$ is also a $k$-return whence
    $R$ decomposes as in case 5 of the lemma.  
    Indeed, $\hist{T}{\hist{U}{\TOP{k-1}(R(m))}}$ is the third
    topmost $(k-1)$-stack of $R(1)$.  
    Since $\TOP{k-1}(R(i))$ is the $(k-1)$-stack directly on top of
    $\hist{U}{\TOP{k-1}(R(i))}$, 
    and because $\hist{\subrun{R}{j}{m}}{\TOP{k-1}(R(m))}$ is not the
    topmost $(k-1)$-stack of $R(j)$ for $0\leq j<m$ (definition of
    $k$-return), 
    we can apply Proposition \ref{prop:Neihbour-position-lemma}
    (variant \ref{npl-vara}) and 
    conclude that 
    $\hist{T}{\TOP{k-1}(R(i))}$ is the second topmost $(k-1)$-stack
    of $R(1)$. 
    By the same proposition, if $\hist{T'}{\TOP{k-1}(R(i))}$ is the topmost
    $(k-1)$-stack of $T'(0)$ for some proper suffix $T'$ of $T$, then  
    $\hist{T'\circ U}{\TOP{k-1}(R(m))}$ is the second topmost
    $(k-1)$-stack of $T'(0)$ which contradicts the minimality of
    $i$. 
    Thus, $T$ is a $k$-return and we showed that $R$ decomposes as
    described in case $5$ of the lemma. We conclude this case by
    proving the claim.
    \begin{proof}[of Claim]
      If the operation leading to $R(m)$ is $\pop k$, $i=m-1$ is a good
      candidate. 
      Otherwise (Proposition \ref{prop:return-ends-k}), this operation
      is $\col k$. 
      Then $\hist{\subrun{R}{m-1}{m}}{\TOP{k-1}(R(m))}$ is not simple.
      Let $i<m-1$ be the last index for which
      $\hist{\subrun{R}{i}{m}}{\TOP{k-1}(R(m))}$ is simple again (such
      $i$ exists because $i=0$ is a good candidate). 
      From Corollary \ref{cor:PushColStackdifference}, applied to
      $\subrun{R}{i}{m}$, we immediately obtain that
      $\hist{\subrun{R}{i}{m}}{\TOP{k-1}(R(m))}$ is the second topmost
      $(k-1)$-stack of $R(i)$, 
      so $i$ is a good candidate. 
    \end{proof}
  \item 
    Assume that $S$ performs a $\push{j}$ operation with $j > k$, and
    $x$ points to the second topmost $(k-1)$-stack of the second
    topmost $(j-1)$-stack of $R(1)$. 
    Let $1< i \leq m$ be minimal such that  
    for  $T:=\subrun{R}{1}{i}$ and $U:=\subrun{R}{i}{m}$ 
    we have $\hist{U}{\TOP{j-1}(R(m))}=\TOP{j-1}(R(i))$.
    We show that $T$ is a $j$-return and $U$ is a $k$-return whence we
    are in case 5. 
    
    Due to the  minimality of $i$ all proper suffixes $T'$ of $T$
    satisfy the inequality
    $\hist{T'}{\TOP{j-1}(T(\lvert T \rvert))}\neq \TOP{j-1}(T'(0))$. 
    Due to Corollary \ref{cor:positionContainment}, the 
    position $\hist{T \circ U}{\TOP{k-1}(R(m))}$ points
    into  
    \begin{equation*}
      \hist{T}{\TOP{j-1}(R(i))}=\hist{T\circ U}{\TOP{j-1}(R(m))}.       
    \end{equation*}
    Thus, $\hist{T}{\TOP{j-1}(R(i))}$ points to the second topmost
    $(j-1)$-stack of 
    $T(0)=R(1)$ and we conclude that $T$ is a $j$-return.
    
    Due to Corollary \ref{cor:positionContainment}, we know that
    $\hist{U}{\TOP{k-1}(R(m))}$ points into
    \begin{equation*}
      \hist{U}{\TOP{j-1}(R(m))}=\TOP{j-1}(R(i)).      
    \end{equation*}
    On the other hand, by Corollary \ref{cor:push+return}, the only
    position $x'$ in the topmost $(j-1)$-stack of $R(i)$ for which
    $\hist{\subrun{R}{0}{i}}{x'}$ points to the second topmost
    $(k-1)$-stack of $R(0)$ 
    is $x$ pointing to the second topmost $(k-1)$-stack of $R(i)$.

    By Proposition \ref{prop:histTransitive} we conclude that
    $\hist{U}{\TOP{k-1}(R(m))}$ points to the second
    topmost $(k-1)$-stack,  
    hence $U$ is a $k$-return.
  \item
    Finally, assume that $S$ performs $\push{1}_{a,k}$, and 
    $x$ is of the form $\TOP0(R(1))\posNew{k}{\TOP{k-1}(u^k)}$ where
    $u^k$ is the $k$-stack stored in the topmost $0$-stack of
    $R(1)$. 
    Let $i>1$ be minimal such that
    $\hist{\subrun{R}{i}{m}}{\TOP{k-1}(R(m))}$ is simple. 
    Recall that $\hist{R}{\TOP{k-1}(R(m))}$ points to the second
    topmost $(k-1)$-stack of $R(0)$. 
    From Corollary \ref{cor:PushColStackdifference}, applied to
    $\subrun{R}{0}{i}$, we see that
    $\hist{\subrun{R}{i}{m}}{\TOP{k-1}(R(m))}=\TOP{k-1}(R(i))$. 
    Since $R$ is a return, this implies $i=m$ and
    due to the minimality of $i$, we conclude directly that
    $\subrun{R}{1}{i}$ is a $k$-colreturn. 
  \end{itemize}
  
  This concludes the proof that every $k$-return decomposes as required by
  the lemma. 
   
  It is left to show that every run that decomposes as described by the
  lemma is a $k$-return. Let $R$ be some run of length $m$. There are the
  following cases. 
  \begin{enumerate}
    \item If $\lvert R \rvert =1$ and it performs $\pop{k}$, the
      definition of $\mathsf{hist}$ implies that $R$ is a $k$-return.
    \item Assume that $R$ starts with an operation of level at most $k-1$, and
      continues with a 
      $k$-return. Since history preserves positions of $(k-1)$-stacks
      under operations of level at most $k-1$, and such operations also preserve the
      existence of all $(k-1)$-stacks, the conditions for $R$ being a
      return are trivially deduced from the fact that
      $\subrun{R}{1}{m}$ is a return. 
    \item Assume that $R$ starts with $\push{1}_{a,k}$ and continues
      with a $k$-colreturn. By definition of a colreturn, the position
      $\hist{\subrun{R}{i}{m}}{\TOP{k-1}(R(m))}$ is not simple for all
      $1\leq i< m$, whence this position is not $\TOP{k-1}(R(i))$. 
      Furthermore, $\hist{\subrun{R}{1}{m}}{\TOP{k-1}(R(m))}$ points into $\TOP0(R(1))$ and has nesting rank $1$, so $\hist{R}{\TOP{k-1}(R(m))}$ is simple.
      By Corollary \ref{cor:PushColStackdifference} it follows that 
      $\hist{R}{\TOP{k-1}(R(m))}$ is the second topmost $(k-1)$-stack
      of $R(0)$. 
    \item Assume that $R$ starts with a $\push{j}$ for $j>k$
      and continues with a $k$-return. Then we conclude similar to the second case. 
    \item Assume that $R$ starts with a $\push{j}$ for $j\geq k$
      (including 
      $\push1_{a,l}$ for $j=1$) and decomposes as 
      $R=S \circ T \circ U$ where $S$ has length $1$, $T$ is a
      $j$-return and $U$ is a $k$-return.
      We know that $x:=\hist{U}{\TOP{k-1}(R(m))}$ is the
      second topmost $(k-1)$-stack of $U(0)$. 
      Corollary \ref{cor:push+return} applied for run $S\circ T$ implies
      that $\hist{R}{\TOP{k-1}(R(m)),R)}$ ($=\hist{S\circ T}{x}$) is
      the second topmost $(k-1)$-stack of $R(0)$. 
      
      For $j=k$, we know that
      $\hist{T'}{\TOP{k-1}(U(0))}\neq \TOP{k-1}(T'(0))$ for every
      suffix $T'$ of $T$ of positive length. 
      We apply Proposition \ref{prop:Neihbour-position-lemma} (variant \ref{npl-varb}) to
      $\TOP{k-1}(U(0))$ and $x$ (the second topmost $(k-1)$-stack of
      $U(0)$ and obtain 
      that $\hist{T'\circ
        U}{\TOP{k-1}(R(m))}\neq\TOP{k-1}(T'(0))$ for every suffix
      $T'$ of $T$. 
      From this we conclude directly that $R$ is a $k$-return. 
      
      For $j>k$, we also see that 
      $\hist{T'\circ U}{\TOP{k-1}(R(m))}\neq \TOP{k-1}(T'(0))$ for
      every suffix  
      $T'$ of $T$ of positive length. 
      Indeed, if $\hist{T'\circ U}{\TOP{k-1}(R(m))}=\TOP{k-1}(T'(0))$,
      then also $\hist{T'\circ U}{\TOP{j-1}(R(m))}=\TOP{j-1}(T'(0))$
      (Corollary \ref{cor:positionContainment})
      which is impossible because $T$ is a $j$-return.
      Again, it is easy to conclude that $R$ is a $k$-return. \qed
    \end{enumerate} 
\end{proof}

Similarly, we now prove the decomposition of colreturns into one
transition followed by shorter returns or colreturns.

\begin{proof}[of Lemma \ref{lem:colreturn-wfrules}]
  We first show that every $k$-colreturn $R$ decomposes as
  described by the lemma.  
  Set $S:=\subrun{R}{0}{1}$ (the definition of a $k$-colreturn
  requires  $\lvert R\rvert \geq 1$).
  If $R$ performs $\col k$ and $|R|=1$ we are in case 1 of Lemma
  \ref{lem:colreturn-wfrules}.
  Otherwise, due to Proposition \ref{prop:ColreturnsStartwithPush},
  the operation in $S$ is $\push{}$. 
  As in the return case, we look at
  $x:=\hist{\subrun{R}{1}{|R|}}{\TOP{k-1}(R(m))}$. 
  By definition of a $k$-colreturn, 
  $\hist{R}{\TOP{k-1}(R(m))}$ is of the form
  $\TOP{0}(R(0))\posNew{k}{x'}$
  for some simple position $x'$.
  By Proposition \ref{prop:histTransitive}, we know that
  $\hist{S}{x}=\hist{R}{\TOP{k-1}(R(m))} =
  \TOP{0}(R(0))\posNew{k}{x'}$. 
  By case distinction on the possible $x$ satisfying this equation (in
  dependence of the operation performed by $S$, 
  there are the following possibilities.
  \begin{enumerate}
  \item 
    $S$ performs a $\push{j}$ for $j\geq 2$  and 
    $x=\TOP{0}(R(1))\posNew{k}{x_1}$ for some simple $x_1$.
    In this case, it is straightforward to see that $\subrun{R}{1}{m}$ is a
    $k$-colreturn. 
  \item Otherwise, 
    $S$ performs a $\push{}$ operation of level $j$ and 
    $x=\TOP{0}(R(0))\posNew{k}{x_1}$ for some simple $x_1$.
    Notice that $\TOP{0}(R(0))$ is the topmost $0$-stack of the second topmost $(j-1)$-stack of $R(1)$.
    Recall that the last operation of $R$ is $\col k$ whence 
    $\hist{\subrun{R}{m-1}{m}}{\TOP{k-1}(R(m))}$ points into the
    topmost $0$-stack of $R(m-1)$ and has nesting rank $1$.
    Let $1<i<m$ be minimal such that 
    $\hist{\subrun{R}{i}{m}}{\TOP{k-1}(R(m))}$ has nesting rank $1$ and points into the topmost $(j-1)$-stack of $R(i)$.
    Let $T:=\subrun{R}{1}{i}$, and let
    $\hist{\subrun{R}{i}{m}}{\TOP{k-1}(R(m))}=y\posNew{k'}{y'}$. 
    Due to Proposition \ref{prop:histPreservesLinkNesting} (applied for run $T$),
    we have $k'=k$, and $\hist{T}{y}=\TOP{0}(R(0))$. 
    Due to Corollary \ref{cor:positionContainment} and since $y$ points
    into $\TOP{j-1}(R(i))$, $\hist{T}{\TOP{j-1}(R(i))}$ contains
    $\TOP{0}(R(0))$
    whence it points to the second topmost $(j-1)$-stack of $R(1)$.
    The same corollary and the minimality of $i$ implies that for each
    suffix $T'$ of $T$ of length at least $1$ we have
    $\hist{T'}{\TOP{j-1}(R(i))}\neq \TOP{j-1}(T'(0))$.  
    Thus, $T$ is a $j$-return. 
    
    Let $U:=\subrun{R}{i}{m}$
    We know that $\hist{S\circ T}{y}=\TOP{0}(R(0))$, and that $y$ is
    in the topmost $(j-1)$-stack of $R(i)$. 
    By Corollary \ref{cor:push+return}, the only $y$ satisfying this
    is $y=\TOP{0}(R(i))$. 
    It follows that $U$ is a $k$-colreturn.
  \end{enumerate}
  
 It is left to show that every run $R$ that decomposes as described by the
 lemma is a $k$-colreturn. 
 In the first two cases we immediately see that $R$ is a $k$-colreturn.
 So assume that $R=S\circ T\circ U$ where $S$ has length $1$ and performs a $\push{j}$ (including $\push1_{a,l}$ for $j=1$),
 $T$ is a $j$-return and $U$ is a $k$-colreturn.
 Let $m:=|R|$.
 As $\hist{U}{\TOP{k-1}(R(m))}$ is of the form
 $\TOP0(U(0))\posNew{k}x$ for simple $x$,  
 by Corollary \ref{cor:push+return} we immediately obtain that $\hist{R}{\TOP{k-1}(R(m))}=\TOP0(R(0))\posNew{k}x$.
 
 In order to prove that $R$ is a $k$-colreturn, we still have to show
 that the position
 $\hist{\subrun{T}{i}{|T|}}{\TOP0(U(0))\posNew{k}x}$ is not simple for
 all $0\leq i\leq |T|$. 
 Heading for a contradiction, assume that there is a greatest index
 $i$ for which the position
 $\hist{\subrun{T}{i}{|T|}}{\TOP0(U(0))\posNew{k}x}$ is simple. 
 Trivially $i<|T|$. 
 If $i=|T|-1$, the last operation of $T$ has to be $\push1_{a,k}$,
 which is impossible in a $j$-return. 
 So $i\leq |T|-2$.
 As $\hist{\subrun{T}{i+1}{|T|}}{\TOP0(U(0))\posNew{k}x}$ is not
 simple (maximality of $i$), it has to be of the form 
 $\TOP0(T(i+1))\posNew{k'}x'$. 
 Proposition \ref{prop:histPreservesLinkNesting} applied to
 $\subrun{T}{i+1}{|T|}$ implies that $k'=k$ and
 $\hist{\subrun{T}{i+1}{|T|}}{\TOP0(U(0))}=\TOP0(T(i+1))$. 
 Due to Corollary \ref{cor:positionContainment}, also
 $\hist{\subrun{T}{i+1}{|T|}}{\TOP{j-1}(U(0))}=\TOP{j-1}(T(i+1))$.
 This is impossible because $T$ is a $j$-return.
\end{proof}

Up to now, we have only dealt with the general shape of $k$-returns
and $k$-colreturns. In Section \ref{sec:Runs} we divided these sets
further according to their \emph{change level}. We next formally
introduce this change level for every $k$-return and $k$-colreturn and
then complete the proof that the rules given in Section \ref{sec:Runs} correctly describe returns and colreturns. 
The 
\emph{change level} keeps track of the maximal level on
which the stack size was changed by the run $R$.
\begin{definition}
  Let $R$ be a $k$-return or a $k$-colreturn. 
  For $1\leq i\leq n$, let $x^i$ be the size of the topmost $i$-stack
  of $R(0)$; similarly $y^i$ for $R(|R|)$. 
  Then $\changelev(R):=\max\{ i: x^i\neq y^i\}$.
\end{definition}
\begin{remark}\label{rem:chl-big}
  If $R$ is a $k$-return ($k$-colreturn), then Proposition
  \ref{prop:return-removes} (Lemma \ref{lem:ColreturnDesreases},
  respectively) implies that  
  the size of topmost $k$-stack of $R(0)$ and of $R(|R|)$ is
  different, so $\changelev(R)\geq k$. 
\end{remark}

We now give a characterisation of the change level of a $k$-return or
$k$-colreturn depending on the change level(s) of the subruns
occurring in its decomposition.

\begin{lemma} \label{lem:Change-level-Return}
  Let $R$ be a $k$-return or a $k$-colreturn. 
  \begin{enumerate}
  \item If $ \lvert R \rvert = 1$, then $\changelev(R)=k$,
  \item If $R$ decomposes as $R=S \circ T$, where $S$ of length $1$
    performs an operation of level $j$, and $T$ is a $k$-return or a
    $k$-colreturn, then $\changelev(R)=\max\{j,\changelev(T)\})$. 
  \item If $R$ decomposes as $R= S\circ T\circ U$ where $S$ of length
    $1$ performs a $\push{j}$ operation (including $\push1_{a,l}$ for
    $j=k=1$), $T$ is a $j$-return, and 
    $U$ is a $k$-return or a $k$-colreturn, then
    $$\changelev(R) =
    \begin{cases}
      \changelev(U) &\text{if } \changelev(T)=j\\
      \max\{\changelev(T), \changelev(U)\} &\text{otherwise.}
    \end{cases}
    $$
  \end{enumerate}
\end{lemma}

We begin the proof with an auxiliary proposition saying that
$k$-returns and $k$-colreturns cannot decrease the size 
of the stacks of level greater than $k$.

\begin{proposition} \label{prop:Change-level-Behaviour}
  Let $R$ be a $k$-return or $k$-colreturn such that $\changelev(R)>k$.
  Then the size of the topmost $\changelev(R)$-stack of $R(0)$ is smaller than the size of the topmost $\changelev(R)$-stack of $R(|R|)$.
\end{proposition}
\begin{proof}
  Let $m:=|R|$.
  For a $k$-return, $\hist{R}{\TOP{k-1}(R(m))}$ points into in the
  topmost $k$-stack of $R(0)$, 
  so by Corollary \ref{cor:positionContainment},
  $\hist{R}{\TOP{k}(R(m))}=\TOP{k}(R(0))$. 
  For a $k$-colreturn we also have
  $\hist{R}{\TOP{k}(R(m))}=\TOP{k}(R(0))$, due to Lemma
  \ref{lem:ColreturnDesreases}. 
  In both cases, by Lemma \ref{lem:PositionsDecrease}, 
  $\TOP k(R(0))\lexOrd \TOP k(R(m))$.
  It follows that the size of the topmost $\changelev(R)$-stack
  of $R(0)$ is smaller than the size of the topmost
  $\changelev(R)$-stack of $R(m)$ 
  (as for $i>\changelev(R)$, the size of the topmost $i$-stack
  of $R(0)$ and of $R(m)$ is the same, and for
  $i=\changelev(R)>k$ they differ). \qed
\end{proof}

Next we prove Lemma \ref{lem:Change-level-Return}.

\begin{proof}[Lemma \ref{lem:Change-level-Return}]
  \begin{enumerate}
  \item   Case 1  is immediate.
  \item 
    Assume we have case 2 of the lemma.
    Notice that neither $S$ nor $T$ can change the size of the $i$-stack
    for $i>\max\{j,\changelev(T)\})$. 
    If $j\neq \changelev(T)$, we see that one of the subruns changes the
    size of the stack of level $\max\{j,\changelev(T)\})$, and the other
    does not change it, 
    so we get $\changelev(R)=\max\{j,\changelev(T)\})$.
    If $j=\changelev(T)$, $\changelev(T)\geq k$ (Remark
    \ref{rem:chl-big}) implies that 
    the operation of the one-step run $S$ is necessarily $\push{}$
    (cf.~Propositions 
    \ref{prop:ReturnsStartwithPushOrSmallLevel} and
    \ref{prop:ColreturnsStartwithPush}). 
    Then the size of the stack of level $j$ is increased by $S$ and by
    $T$ (cf.\ Proposition \ref{prop:Change-level-Behaviour}). Thus, 
    the claim follows immediately.
  \item 
    Next, assume we have case 3 of the lemma.
    None of the parts $S$, $T$, $U$ changes the size of the $i$-stack
    for $i>\max\{\changelev(T),\changelev(U)\})$. 
    If $\changelev(T)=j$, Corollary \ref{cor:push+return}
    implies that the topmost $j$-stack of $R(0)$ and of $U(0)$ is the same,
    thus $\changelev(R)=\changelev(U)$. 
    So assume that $\changelev(T)>j$.
    Then the size of the stack of level
    $\max\{\changelev(T),\changelev(U)\})$ cannot be decreased by $S$ or
    $T$ or $U$ (Proposition \ref{prop:Change-level-Behaviour}), 
    and at least one of $T$ and $U$ increases this value.
    Thus, $\changelev(R)=\max\{\changelev(T),\changelev(U)\})$.\qed
  \end{enumerate}  
\end{proof}

\subsection{Non-Erasing Runs}
\label{sec:nonerasingRuns}

\begin{definition}
  For $0\leq k \leq l$, let $\Nn_{k,\varepsilon}$ be the set of
  $\TOP{k}$-non-erasing runs which is the set of runs
  $R$ such that position $\TOP k(R(0))$ is present in every
  configuration of $R$. 
\end{definition}

Using $k$-returns we can characterise $\TOP{k}$-non-erasing runs in
the following way. 

\begin{lemma}\label{lem:non-topk-erasing}
  Let $R$ be some run and $0\leq k\leq n$.
  $R$ is a $\TOP{k}$-non-erasing run if and only if $R$ has one
  of the following forms.
  \begin{enumerate}
  \item \label{item:non-topk-erasingOne} $\lvert R  \rvert=0$.
  \item \label{item:non-topk-erasingTwo}
    $R$ starts with an operation of level at most $k$, and
    continues with a $\TOP{k}$-non-erasing run. 
  \item \label{item:non-topk-erasingThree}
    $R$ starts with a $\push{j}$ (including arbitrary $\push 1_{a,l}$
    for $j=1$) for $j\geq k+1$, and 
    continues with a $\TOP{j-1}$-non-erasing run. 
  \item \label{item:non-topk-erasingFour}
    $R$ starts with a $\push{j}$ (including arbitrary $\push 1_{a,l}$
    for $j=1$) and decomposes as $R=S\circ 
    T\circ U$, where $S$ has length $1$, $T$ is a $j$-return of change level $j$,
    and $U$ is a $\TOP{k}$-non-erasing run. 
  \end{enumerate}
\end{lemma}

We start the proof by giving two propositions useful in the
right-to-left implication. 

\begin{proposition}\label{prop:topk-non-eras-composition}
  Let $R=S\circ T$ be a run such that $S$ and $T$ are $\TOP
  k$-non-erasing runs for some $k$. 
  Then $R$ is a $\TOP k$-non-erasing run.
\end{proposition}
\begin{proof}
  We claim the following. 
  Take some run such that $x$ and $y$ are simple positions in its
  initial stack such that $x\lexOrd y$. If $y$ is present in all
  configurations of the run, then $x$ is also present in all
  configurations of the run. 
  
  Since $\TOP{k}(R(0))$ is present in all configurations of $S$, the
  claim implies that all
  $k$-stacks present in $R(0)=S(0)$ are also present in $S(\lvert S
  \rvert)$. Thus, the topmost $k$-stack of $T(0)=S(\lvert S \rvert)$
  is lexicographically greater or equal than $\TOP{k}(R(0))$. 
  Again using the claim, $\TOP{k}(R(0))$ is present in all
  configurations of $T$ because $\TOP{k}(T(0))$ is present in all
  configurations of $T$. 

  For the proof of the claim note that the statement of the claim is
  preserved under composition of runs. Thus, we may consider a run $R$
  of length $1$ such that $x\lexOrd y$ are positions in $R(0)$.
  Since $\push{}$ operations do not delete positions in a stack, we
  may assume that $R$ performs $\pop{j}$ or $\col{j}$. 
  Since an application of $\col{j}$ has the same effect as several
  $\pop{j}$, it is sufficient to consider the $\pop{j}$ case (the
  $\col{j}$-case then follows again by the composition closure
  argument).  
  Assume that $R$ performs a $\pop{j}$ and $x$ is present in $R(0)$
  but not in $R(1)$. Then $x$ points into or to the topmost
  $(j-1)$-stack of $R(0)$. Since $x\lexOrd y$, $y$ must also point
  into or to the topmost $(j-1)$-stack of $R(0)$. But then $y$ is not
  present in $R(1)$. \qed
\end{proof}

\begin{proposition}\label{prop:history2topk-non-eras}
  Let $0\leq k\leq n$, and let $R$ be a run such that
  $\hist{R}{y}=\TOP k(R(0))$ for some position $y$ of $R(|R|)$. 
  Then $R$ is a $\TOP k$-non-erasing run.
\end{proposition}

\begin{proof}
  Heading for a contradiction, assume that there is a minimal $i\leq
  \lvert R \rvert$  such that  
  $x_0:=\TOP k(R(0))$ is not present in $R(i)$. 
  All simple positions in $R(i)$ are lexicographically smaller than $x_0$,
  because $x_0$ was removed either by a $\mathsf{pop}$ operation, or
  by a $\mathsf{col}$ operation (cf.\ Remark \ref{rem:ColisPop}). 
  Let $x_1$ be the simple prefix of $\hist{\subrun{R}{i}{\lvert R\rvert}}{y}$.  
  Due to Lemma \ref{lem:PositionsDecrease} applied to $\subrun{R}{0}{i}$,  $x_0\lexOrd x_1$.
  But this is a contradiction. \qed
\end{proof}


\begin{proof}[Lemma \ref{lem:non-topk-erasing}]
  The proof of the right-to-left part is 
  by case distinction on the decomposition of $R$ according to the
  four cases. 
  Case \ref{item:non-topk-erasingOne} is trivial and Cases
  \ref{item:non-topk-erasingTwo} and 
  \ref{item:non-topk-erasingThree} follow directly from 
  Proposition \ref{prop:topk-non-eras-composition}.
  We now investigate Case \ref{item:non-topk-erasingFour}.
  Notice that $\hist{S\circ T}{\TOP{k}(T(\lvert T
    \rvert))}=\TOP{k}(S(0))$: for $k<j$ it 
  follows from Proposition \ref{prop:return-removes}; for $k\geq j$ it
  follows from Corollary \ref{cor:positionContainment}.
  Thus, Proposition \ref{prop:history2topk-non-eras} applied to
  $S\circ T$ and $y:=\TOP k(T(|T|))$ tells us that $S\circ T$ is a
  $\TOP k$-non-erasing run. 
  Due to Proposition \ref{prop:topk-non-eras-composition}, also $R$ is
  a $\TOP k$-non-erasing run. 
 
  Now concentrate on the left-to-right part.
  Let $R$ be a run of length $m$ such that
  $x:=\TOP{k}(R(0))$ is present in all configurations of $R$.
  If $m = 0$, we are in case \ref{item:non-topk-erasingOne}.
  Thus, assume that $m\geq 1$. 
  Note that the first operation cannot be $\col{j}$ or
  $\pop{j}$ for 
  $j \geq k+1$  because this would delete position $x$ from
  $R(1)$ (cf.\ Remark \ref{rem:ColisPop}).  
  Hence, one of the following cases applies.
  \begin{itemize}
  \item Assume that the first operation in $R$ is of level at most $k$. Then 
    $x=\TOP{k}(R(1))$ and $x$ is not removed during
    $\subrun{R}{1}{m}$. Thus,
    $R$ decomposes as in case \ref{item:non-topk-erasingTwo}.
  \item Assume that the first operation in $R$ is $\push{j}$ for some
    $j \geq k+1$ 
    (in the rest of the proof, $\push{1}$ stands for arbitrary
    $\push{1}_{a,k'}$).  
    Furthermore, assume that $y:=\TOP{j-1}(R(1))$ is present in all
    configurations of 
    $T:=\subrun{R}{1}{m}$. Then $R$ decomposes as
    as in case \ref{item:non-topk-erasingThree}.
  \item Otherwise, the first operation is $\push{j}$ for some $j\geq k+1$
    and there is a minimal $l\geq 1$ such that 
    $y:=\TOP{j-1}(R(1))$ is not present in $R(l)$. 
    We claim that $\subrun{R}{1}{l}$ is a $j$-return of change level
    $j$ and that
    the positions $\TOP{k}(R(l))$ and $\TOP{k}(R(0))$ agree. 
    Hence, $\subrun{R}{l}{m}$ is  $\TOP{k}$-non-erasing and
    $R$ decomposes as in case \ref{item:non-topk-erasingFour}.
    
    Let us proof the claim.
    Recall that 
    $x=\TOP{k}(R(0))$ is present in all configurations of
    $\subrun{R}{0}{l}$.  
    Since $x$ points
    into $x':=\TOP{j-1}(R(0))$ or $x'=x$,
    \begin{equation}
      \label{eq:test1}
      x' \text{ is present in all configurations of }\subrun{R}{0}{l}.
    \end{equation}
    Hence we can apply
    Lemma \ref{lem:TopPresentImpliesHistoryIsId} and conclude that
    $\hist{\subrun{R}{0}{l}}{x'}=x'$ and 
    \begin{equation}
      \label{eq:test2}
      \text{if } \hist{\subrun{R}{l'}{l}}{x'}\text{ is
    simple, it is equal to }x'.       
    \end{equation}
    Thus, $\hist{\subrun{R}{1}{l}}{x'}=x'$ because no
    non-simple
    position $z$ in $R(1)$ satisfies 
    $x'=  \hist{\subrun{R}{0}{1}}{z}$. 

    By definition of $\push{j}$,
    $x'$ is the second topmost $(j-1)$-stack in
    $R(1)$. 
    Since $y$ is directly above $x'$ and present in $R(l')$ for all
    $1\leq l'<l$, 
    (\ref{eq:test1}) and (\ref{eq:test2}) imply that
    $\hist{\subrun{R}{l'}{l}}{x'} \neq \TOP{j-1}(R(l'))$.
    Finally, note that 
    from $R(l-1)$ to $R(l)$ the
    $(j-1)$-stack above $x'$ (which is at $y$) is removed but $x'$ is
    present in  $R(l)$.
    Using Remark \ref{rem:ColisPop}, the operation
    is $\pop{j}$ or $\col{j}$ and $y$ points into the topmost
    $j$-stack of $R(l-1)$ whence $x'$ points to the topmost
    $(j-1)$-stack of $R(l)$.  
    
    In summary, $x'=\TOP{j-1}(R(l))$, 
    $\hist{\subrun{R}{1}{l}}{x'} = x'$ is the second topmost
    $(j-1)$-stack  of $R(1)$ and
    $\hist{\subrun{R}{l'}{l}}{x'}$ is not the topmost $(j-1)$-stack of
    $R(l')$ for 
    all $1\leq l' < l$. 
    Thus,  $\subrun{R}{1}{l}$ is a  $j$-return of change level $j$ and
    the claim is proved.\qed
  \end{itemize}
\end{proof}

\subsection{Pumping runs}

In this subsection we give a definition of pumping runs
and prove that the rules from Section \ref{sec:Runs} describe pumping
runs correctly. 

\begin{definition}
  For $x\in\{=,<\}$ and $y\in\{\varepsilon,\notepsBig\}$, let
  $\Pp_{x,y}$ be the set of runs $R$ such that 
  \begin{itemize}
  \item	$\hist{R}{\TOP{0}(R(\lvert R \rvert))}=\TOP{0}(R(0))$, and
  \item	$\TOP{0}(R(\lvert R \rvert))=\TOP{0}(R(0))$ if and only if $x$
    is $=$, and 
  \item	$R$ uses only $\varepsilon$-transitions if and only if
    $y=\varepsilon$. 
  \end{itemize}
  A run $R$ is a \emph{pumping run} if it belongs to some $\Pp_{x,y}$.
\end{definition}
\begin{remark}
  Lemma \ref{lem:PositionsDecrease} implies that for a pumping run
  $R \in \Pp_{<, y}$, we have $\TOP{0}(R(0))\lexOrdstrict
  \TOP{0}(R(\lvert R \rvert))$. In this sense the final stack of a
  pumping run is is greater than its initial one. 
\end{remark}

For the next proofs it is useful to distinguish all $k$-returns of
minimal change level (i.e., of change level $k$) from those of higher
change level. 
\begin{definition}\label{def:ret-agree-not}
  We set $\Rr_{k,=,x}:=\Rr_{k,k,x}$  and
  $\Rr_{k,<,x} :=\bigcup_{i>k} \Rr_{k,i,x}$. 
\end{definition}
\begin{remark}
  A $k$-return $R$  in $\Rr_{k,=,x}$ satisfies
  $\TOP{k}(R(0))=\TOP{k}(R(\lvert R \rvert))$. 
  Due to Proposition \ref{prop:Change-level-Behaviour}, 
  a $k$-return $R$ in $\Rr_{k,<,x}$ satisfies
  $\TOP{k}(R(0))\lexOrdstrict \TOP{k}(R(\lvert R \rvert))$.
\end{remark}

In the rest of this subsection we characterise pumping runs using wf-rules.

\begin{lemma}\label{lem:pumping-run}
  Let $R$ be some run.
  $R$ is a pumping run if and only if $R$ has one of the following forms.
  \begin{enumerate}
  \item $\lvert R \rvert = 0$.
  \item \label{item:pumping-run-two}$R$ starts with a $\push k$
    of any level (including arbitrary $\push 1_{a,l}$ for
    $k=1$), and continues  
    with a pumping run. 
  \item \label{item:pumping-run-three}
    $R$ starts with a $\push k$ of any level (including arbitrary
    $\push 1_{a,l}$ for $k=1$), and decomposes 
    as $R=S\circ T\circ U$, where $S$ has length $1$, $T$ is a
    $k$-return, and $U$ is a pumping run. 
  \end{enumerate}
  Additionally, assuming that $R$ is a pumping run, 
  $\TOP{0}(R(0))=\TOP{0}(R(\lvert R \rvert))$
  if and only if 
  \begin{itemize}
  \item $R$ is of the first form, or
  \item $R$ is of the last form, and $\TOP{k}(T(0))=\TOP{k}(T(\lvert T \rvert))$, 
    and $\TOP{0}(U(0))=\TOP{0}(U(\lvert U \rvert))$ 
  \end{itemize}
\end{lemma}

The characterisation of pumping runs in
terms of the well-formed rules presented in Section \ref{sec:Runs}
follows immediately from the previous lemma.

\begin{remark}
  Observe that $\hist{R}{x}=\TOP{0}(R(0))$ implies that the first
  operation of $R$ is not $\pop{}$ or $\col{}$. 
  Indeed, after such an operation in $R(1)$ we have no position $y$ such 
  that $\hist{\subrun{R}{0}{1}}{y}=\TOP{0}(R(0))$ (which contradicts
  with Proposition \ref{prop:histTransitive}). 
\end{remark}

\begin{proof}[Lemma \ref{lem:pumping-run}]
  The right-to-left direction of the first part is almost immediate.
  In the third case we have to observe that $\hist{S\circ
    T}{\TOP0(U(0))}=\TOP0(R(0))$; it follows from Corollary
  \ref{cor:push+return}. 
  
  Now concentrate on the left-to-right direction of the first part of the lemma.
  Let $R$ be a pumping run of length $m$.
  If $m=0$, we are in case 1.
  Thus, assume that $m\geq 1$.
  Due to the above remark, $R$ starts with a $\push{}$ operation of
  some level $k$.  
  Recall that $\hist{\subrun{R}{0}{1}}{x}=\TOP{0}(R(0))$ only if
  $x=\TOP{0}(R(1))$ or if $x$ points to the topmost $0$-stack of
  the second topmost $(k-1)$-stack. 
  By Proposition \ref{prop:histTransitive},
  $\hist{\subrun{R}{1}{m}}{\TOP{0}(R(m))}$ is one of these positions
  $x$. 
  Now there are two cases.
  \begin{itemize}
  \item If 
    $\hist{\subrun{R}{1}{m}}{\TOP{0}(R(m))}=\TOP{0}(R(1))$, then
    $\subrun{R}{1}{m}$  
    is a pumping run and $R$ decomposes as in case \ref{item:pumping-run-two}.
  \item Otherwise, 
    $\hist{\subrun{R}{1}{m}}{\TOP{0}(R(m))}$ points to the topmost
    $0$-stack of the second topmost  $(k-1)$-stack of 
    $R(1)$. 
    Due to Corollary \ref{cor:positionContainment}, we conclude
    that
    $\hist{\subrun{R}{1}{m}}{\TOP{k-1}(R(m)}$ points to the second
    topmost $(k-1)$-stack of $R(1)$. 
    Let $2\leq i \leq m$ be minimal such that
    $\hist{\subrun{R}{i}{m}}{\TOP{k-1}(R(m))}=\TOP{k-1}(R(i))$. 
    Notice that $T:=\subrun{R}{1}{i}$ satisfies all requirements of a $k$-return.
  
    By Corollary \ref{cor:positionContainment} we know that
    $\hist{\subrun{R}{i}{m}}{\TOP{0}(R(m))}$ points into
    \begin{align*}
    \hist{\subrun{R}{i}{m}}{\TOP{k-1}(R(m))}=\TOP{k-1}(R(i)).      
    \end{align*}
    On the other hand, by Corollary \ref{cor:push+return}, the only
    position $x$ in the topmost $(k-1)$-stack of $R(i)$ for which
    $\hist{\subrun{R}{0}{i}}{x}=\TOP0(R(0))$ is $x=\TOP0(R(i))$. 
    By Proposition \ref{prop:histTransitive} we conclude that
    $\hist{\subrun{R}{i}{m}}{\TOP{0}(R(m))}=\TOP0(R(i))$, 
    hence $U:=\subrun{R}{i}{m}$ is a pumping run.
  \end{itemize}

  Next we prove the last part of the lemma. 
  If $R$ is of length $0$ we immediately get $\TOP 0(R(0))=\TOP 0(R(|R|))$.
  Let $R$ be a run satisfying item \ref{item:pumping-run-three} such that $\TOP{k}(T(0))=\TOP{k}(T(\lvert T \rvert))$ and $\TOP{0}(U(0))=\TOP{0}(U(\lvert U \rvert))$.
  Because the operation in $S$ is $\push k$ we also have $\TOP{k}(R(0))=\TOP{k}(T(\lvert T \rvert))$.
  By Corollary \ref{cor:push+return} we know that the topmost $k$-stack of $R(0)$ and of $T(|T|)$ are the same, so $\TOP{0}(R(0))=\TOP{0}(T(\lvert T \rvert))=\TOP{0}(R(|R|))$.

  Finally assume that $R$ is a pumping run of length $m\geq 1$ such
  that $\TOP{0}(R(m))=\TOP{0}(R(0))$. 
  Then $\hist{R}{\TOP{0}(R(m))} = \TOP{0}(R(m))$. 
  We already have observed that
  $\hist{\subrun{R}{1}{m}}{\TOP{0}(R(m))}$ is simple.  
  Due to Lemma \ref{lem:PositionsDecrease} it is lexicographically bounded
  from above by $\TOP{0}(R(m))$ and from below by
  $\hist{R}{\TOP{0}(R(m))}=\TOP{0}(R(m))$. We conclude that 
  $\hist{\subrun{R}{1}{m}}{\TOP{0}(R(m))}=\TOP{0}(R(0))$.
  From the analysis in the first part, we know that $R$ then satisfies
  case \ref{item:pumping-run-three}, i.e., it decomposes as
  $R=S\circ T \circ U$ where $S$ performs only one $\push{k}$, $T$ is
  a $k$-return and $U$ is a pumping run. Using the same argument
  again, we conclude that
  $\TOP{0}(U(0))=\hist{U}{\TOP{0}(R(m))}=\TOP{0}(R(m))$.
  Using Corollary \ref{cor:positionContainment} we also get that
  $\TOP{k}(T(0))=\TOP{k}(R(0))=\TOP{k}(U(0))$. \qed
\end{proof}

In conclusion, Lemmas \ref{lem:pumping-run},
\ref{lem:non-topk-erasing}, \ref{lem:Change-level-Return},
\ref{lem:colreturn-wfrules} and \ref{lem:return-wfrules} show that the
Rules from Section \ref{sec:Runs} describe sets of runs that satisfy the
intended meaning described in that Section.


\section{Sketch of proof of Theorem \ref{thm:pumping}}
\label{app:sketchPumping}
In this section we describe briefly the proof of the pumping lemma.
The single steps of this proof follow closely the analogous proof for
the non-collapsible pushdown systems in \cite{parys-pumping}.
For the details of these steps we refer the reader to Appendix
\ref{app:PumpingLemma} (which requires Appendix \ref{app:Combinatoric}
as combinatorial background). 

First we list three propositions, which are consequences of Theorem
\ref{thm:types} applied to the family $\Xx$ from Section
\ref{sec:Runs}.

\begin{proposition}\label{prop:types1} 
  Let $R$ be a pumping run of the system $\Ss$ such that $R$ satisfies
  $\ctype_\Xx(R(0)) \typOrd \ctype_\Xx(R(\lvert R \rvert)).$
  Then there is
  a sequence of runs $(R_i)_{i\in\N}$
  such that 
  $R_i(0)=R(0)$, 
  $\ctype_\Xx(R( \lvert R\rvert ))\typOrd \ctype_\Xx(R_i(\lvert R_i\rvert))$,
  and
  \begin{enumerate}
  \item if $R\in\Pp_\noteps$ then $R_i$ contains at least $i$
    non-$\varepsilon$-transitions, 
  \item if $R\in\Pp_{<,\varepsilon}$ then the final stack of $R_{i+1}$
    is greater than the final stack of $R_i$, and $R_i$ uses only
    $\varepsilon$-transitions. 
  \end{enumerate}
\end{proposition}
\begin{proof}
  Set $R_0:=R$.
  Application of Theorem \ref{thm:types} to $R$ and configuration 
  $R_i(\lvert R_i\rvert)$ yields a pumping run $R'_{i+1}$. Set
  $R_{i+1}:=R_i\circ R'_{i+1}$. \qed
\end{proof}

\begin{proposition}\label{prop:types2}
  Let $R$ be a $\TOP{0}$-non-erasing run and $c$ some
  configuration such that 
  $\ctype_\Xx(R(0))\typOrd \ctype_\Xx(c)$. Then there is a
  $\TOP{0}$-non-erasing run $S$ that 
  starts in $c$ and ends in the same state as $R$. 
\end{proposition}

\begin{proposition}\label{prop:types3}
  Let $R$ be a run and $c$ a
  configuration such that $\ctype_\Xx(R(0))\typOrd \ctype_\Xx(c)$. 
  Then there is a run $S$ which starts in $c$ and ends in the same
  state as $R$. 
\end{proposition}

Let us also comment on the crucial properties of pumping runs and $\TOP0$-non-erasing runs. 
A run $R\in\Pp\cup\Nn_0$ ends in a stack which is not smaller than the stack in
which $R$ starts; 
in particular $R\in\Pp_{<,\varepsilon}$ ends in a strictly greater
stack than it starts. 
The classes $\Pp$ and $\Nn_0$ are quite similar. 
The main differences between them 
are the following two.
The definition of
pumping runs is more restrictive, i.e.,  $\Pp\subseteq\Nn_0$ allowing to set 
$\lev(\Pp_\noteps)=\lev(\Pp_{<,\varepsilon})=0$ while $\lev(\Nn_0)>0$.
Thus, Theorem \ref{thm:types} gives a stronger transfer property for
$\Pp$ than for $\Nn_0$. 
On the other hand, $\Nn_0$ is closed under prefixes in the sense that 
for $R\in\Nn_0$ we also have $\subrun{R}{0}{i}\in\Nn_0$ for any
$i\leq \lvert R \rvert$. 
For $R\in \Pp$ we have $\subrun{R}{0}{i}\in \Nn_0$ for all $i<\lvert R
\rvert$ but not always $\subrun{R}{0}{i}\in \Pp$. 

Next we show how these propositions can be used to prove the pumping lemma. 
This part consists of the following steps (where we write $\Gg$ for
the $\varepsilon$-contraction of the graph of $\Ss$ of level $\sL$).
\begin{enumerate}
\item \label{Part3-step1} A simple construction shows that we may
  assume that the state of 
  a reachable configuration $c$ determines whether $c\in \Gg$
  (the state set $Q$ is
  partitioned as $Q=Q_{\varepsilon} \cup Q_{\noteps}$ such
  that $q\in Q_{\varepsilon}$ implies that all edges leading to $q$
  are labelled $\varepsilon$ and $q\in Q_{\noteps}$ implies
  that all edges leading to $q$ are not labelled $\varepsilon$).
  In the rest of the proof we assume that this condition holds.
\item 
  We say a run $R$ \emph{induces a path of length $l$ in $\Gg$} if
  $R(0)\in \Gg$ and $R(\lvert R \rvert)\in \Gg$ and there are $l+1$
  many $i\leq \lvert R \rvert$ such that $R(i)\in \Gg$. 
  Recall that $C_{\Ss}$ was defined in Theorem \ref{thm:pumping}.
  We show that 
  every run $R$ starting at a configuration $c$ of distance $m$ from
  the initial one in $\Gg$ that induces a path of length $C_{\Ss}$
  in $\Gg$ contains a pumping subrun
  $T\in\Pp_{<,\varepsilon}\cup\Pp_{\noteps}$ such that
  $R=S\circ T \circ U$.\footnote{This step relies on the assumption
    that the system is finitely branching. This assumption allows to
    derive a bound on the size of the configuration $c$ in terms of
    the distance $m$.}
  Moreover, if $T\in \Pp_{<,\varepsilon}$ we can show that
  $\subrun{U}{0}{m}$ is a   
  $\TOP{0}$-non-erasing run for some $m\leq \lvert U\rvert$ such that
  $U(m)\in\Gg$. 
\item 
  We conclude with the following case distinction.
  \begin{itemize}
  \item If $T\in \Pp_{<,\varepsilon}$, let $b\leq\lvert S\rvert$ be
    maximal such that $S(b)$ is a node of $\Gg$.
    We apply 
    Proposition \ref{prop:types1} to $T$ and obtain infinitely many
    runs  $(R_j)_{j\in\N}$ ending in configurations $(q,s_i)$ 
    such that $\subrun{S}{b}{\lvert S \rvert} \circ
    R_j$ is an $\varepsilon$-labelled path from $S(b)$ to $(q,s_i)$
    where $s_{i+1}$ is greater than $s_i$ for all $i\in\N$. 
    We can apply Proposition \ref{prop:types2} to $\subrun{U}{0}{m}$ and $(q,s_i)$
    and obtain a $\TOP{0}$-non-erasing run $U'_i$ from $(q,s_i)$. 
    By construction $U_i'$ ends in the same state as $\subrun{U}{0}{m}$. 
    Due to step \ref{Part3-step1}, we conclude that $U'_i$
    ends in some node of $\Gg$ (because $\subrun{U}{0}{m}$ does so and their final
    states coincide). Let $U_i$ be the minimal prefix of
    $U'_i$ such that 
    $U_i(\lvert U_i\rvert)\in \Gg$, i.e., $U_i=\subrun{U'_i}{0}{j}$
    such that the transition between $U_i'(j-1)$ and $U_i'(j)$ is the first transition of $U_i'$ not labelled by $\varepsilon$.
    We know that $U_i$ is also a $\TOP0$-non-erasing run, so it ends
    in a greater or the same stack than it starts. 
    Let $g_i$ be the final configuration of $U_i$.
    Note that
    $\subrun{S}{b}{\lvert S \rvert} \circ R_i \circ U_i$ is a run that
    induces a path of length $1$ from $S(b)$ to $g_i$, i.e., $g_i$ is
    a successor of $S(b)$ in $\Gg$. 
    Since the sizes of the stacks $s_i$ increase strictly for each $i$
    there is a $j$ such that $s_j$ is bigger than $g_i$. Since $g_j$
    is even bigger than $s_j$, $g_i$ and $g_j$ cannot coincide. 
    Inductive use of this argument yields a sequence $(i_k)_{k\in\N}$
    such that the $g_{i_k}$ are pairwise different successors of
    $S(b)$ in $\Gg$. 
    Especially, we conclude that this case cannot occur in a finitely
    branching $\varepsilon$-contraction.
  \item Otherwise, $T\in \Pp_\noteps$. 
    Application of Proposition \ref{prop:types1} to $T$ yields a
    sequence $(R_i)_{i\in\N}$ of runs in $\Pp_\noteps$
    starting in $T(0)$ such that $R_i$ contains at least $i$
    transitions not labelled $\varepsilon$ and such that
    $\ctype_\Xx(T(\lvert T \rvert)) \typOrd 
    \ctype_\Xx(R_i(\lvert R_i \rvert))$.
    Application of Proposition \ref{prop:types3} to $U$ and
    $R_i(\lvert R_i\rvert)$ yields a run $U_i$ from $R_i(\lvert
    R_i\rvert)$ to some configuration with the same final state as
    that of $U$ (which is also the same final state as that of $R$). 
    Using part \ref{Part3-step1}), we conclude  that $S\circ R_i \circ
    U_i$  induces a 
    path of length at least $i$ starting in $c$ for each $i\in\N$.
    Thus, $\Gg$ contains paths of arbitrary length starting in $c$
    and ending in the same state as $R$; this completes the proof of the
    pumping lemma.  
  \end{itemize}
\end{enumerate}


\section{Combinatorics for 
  Theorem \ref{thm:pumping}}
\label{app:Combinatoric}
In this part we collect some combinatorial facts that turn out to be
useful in the proof of Theorem \ref{thm:pumping}.
The first lemma says that if we have a sequence of natural numbers
which increase at most by one from one number to the next and if we
choose a set $G$ of $2^k-1$ of these numbers then we find an increasing
subsequence of $k$ numbers such that each element of the subsequence
is strictly smaller than all following elements of the sequence up to
the next occurrence of a number from $G$. 
This sequence is intended to contain sizes of a stack during a run;
such size can increase by at most $1$ (when a $\push{}$ is performed) and can decrease arbitrarily (when a $\col{}$ is performed).
For $G\subseteq \N$ with $l-1\in G$ and $i < l$ we set
$n_G(i):=\min\{g\in G: g\geq i\}$. 

\begin{lemma}\label{lem:IncreasingNumberSequence} 
  Let $k\in\N\setminus\{0\}$, let 
  $(a_i)_{0 \leq i \leq l}$ be a sequence of natural numbers such that 
  $a_i-a_{i-1} \leq  1$ for all $1\leq i \leq l$ and such that
  $a_0 = \min\{ a_i: 0\leq i \leq l\}$.
  Let $G\subseteq\{0,1, \dots, l-1\}$ be such that 
  $\lvert G\rvert \geq 2^k-1$. 
  
  There is an $e \leq l$ such that $e-1\in G$ and  for
  \begin{align*}
    H_e:= \{ i\leq e-1:\ &a_i \leq a_j\text{ for all } i\leq j \leq e
    \text{ and }\\
    &a_i < a_j\text{ for all } i<j \leq n_G(i)\}
  \end{align*}
  we have $\lvert H_e \rvert \geq k$. 
\end{lemma} 
 
\begin{proof} 
  The proof is by induction on $l$. 
  For $0 \leq b \leq e \leq l$ we write
  $H_{b,e} = H_e \cap \{i\in\N: i\geq b\}$.
  Note that it suffices to find $0\leq b \leq e \leq l$ such that
  $\lvert H_{b,e} \rvert \geq k$. 
  We distinguish the following cases. 
  \begin{enumerate}
  \item 
    Assume that $k=1$. 
    Since $\lvert G \rvert \geq 2^1-1=1$, we can choose some $e'\in
    G$. Let $e:=e'+1 \leq l$. 
    Choose $b\leq e'$ maximal such that 
    $a_b= a_0$. 
    By choice, $a_b < a_j$ for all $b < j \leq e'$, and $a_b=a_0\leq a_e$.
    Thus, $b\in H_{b,e}$ which settles the claim. 
  \item \label{case:lem:IncreasingNumberSequence2}
    Assume that there is some $0<b\leq l$ such that
    $G\subseteq\{b,b+1,\dots,l-1\}$ and $a_b= \min\{ a_i: b\leq i \leq
    e\}$. By induction hypothesis, there is some $e\leq l$ with
    $e-1\in G$  such that $\lvert H_{b,e}\rvert \geq k$. 
  \item  \label{case:lem:IncreasingNumberSequence3}
    Assume that 
    there is some $1\leq l' \leq l$ such that
    $a_i > a_0$ for all $1\leq i \leq l'$ and
    $\lvert G \cap\{1,2,\dots,l'-1\}\rvert \geq 2^{k-1}-1$.
    Since $a_1-a_0\leq 1$,
    it follows that $a_1=a_0+1=\min\{a_i: 1\leq i \leq l'\}$. 
    Thus, we can apply the induction hypothesis to the sequence
    $a_1, a_2, \dots, a_{l'}$ and $k-1$ and obtain $e\leq l'$ such that $e-1\in G$ and 
    $\lvert H_{1,e} \rvert \geq k-1$. Since $a_0<a_i$ for all $1\leq
    i \leq l'$, we conclude that $H_{0,e} = \{a_0\} \cup H_{1,e}$
    contains at least $k$ elements. 
  \item 
    Assume that none of the above cases holds. 
    Then in particular $k\geq 2$. 
    Let $b\geq 1$ be the smallest index such that $a_{b}=a_0$. 
    If such $b$ would not exist, case
    \ref{case:lem:IncreasingNumberSequence3} would hold with $l'=l$. 
    
    Let $G'=G\cap\{b,b+1,\dots,l-1\}$. 
    We have $a_i> a_0$ for $1\leq i\leq b-1$. 
    Because $b-1$ cannot be taken as $l'$ in 
    case \ref{case:lem:IncreasingNumberSequence3}, we either have
    $b=1$, or $\lvert G\cap\{1,2,\dots,b-2\}\rvert \leq 2^{k-1}-2$.  
    In the former case, $\lvert G' \rvert \geq 2^{k}-1-1\geq
    2^{k-1}-1$ and in the latter case
    $\lvert G'\rvert \geq (2^k-1)-2-(2^{k-1}-2) \geq 2^{k-1}-1$.
    Since $a_b=a_0=\min\{a_i: 0\leq i \leq l\}$, the induction
    hypothesis applies to the shorter sequence
    $a_b, a_{b+1}, \dots, a_l$ and $G'$. Thus, there is some 
    $e\leq l$ such that $e-1\in G' \subseteq G$ and 
    $\lvert H_{b,e} \rvert \geq k-1$.

    Since we are not in case \ref{case:lem:IncreasingNumberSequence2}, 
    $G'\neq G$ whence there is some
    $g\in G$ with $0\leq g \leq b-1$. 
    Since $ a_0<a_i$ for all $0<i\leq b-1$ we also have
    $a_0 < a_i$ for all $0 < i \leq g$. 
    Since $a_0$ is the minimal element of the sequence
    and since $0 <b$ we have $\lvert H_{0,e}\rvert \geq \lvert \{a_0\}
    \cup H_{b,e} \rvert \geq k$.\qed
  \end{enumerate}
\end{proof} 
 
\begin{corollary}\label{cor:IncreasingNumberSequence}
  Let $k\in \N\setminus\{0\}$ and let
  $a_0,a_1,\dots,a_l$ be 
  a sequence of positive natural numbers such that
  $a_i-a_{i-1}\leq 1$ for $1\leq i\leq l$.  
  Let $G\subseteq\{0,1,\dots,l-1\}$ be such that 
  $\lvert G \rvert \geq  a_0\cdot 2^k$.  
  Then there exist two indices  $0\leq b<e\leq l$ such that
  \begin{enumerate}
  \item $e-1\in G$,
  \item	$a_b = \min\{ a_i: b\leq i\leq e\}$,
  \item	$a_i > a_b$  for each  $0\leq i < b$ and  
  \item $\lvert H_{b,e} \rvert \geq k$.
  \end{enumerate} 
\end{corollary} 

\begin{proof} 
  For each 
  $0\leq j\leq l$ set  
  \begin{align*}
    m_j :=\min\{a_i: 0\leq i\leq j\}.    
  \end{align*}
  Notice that $m_0,m_1,\dots,m_l$ is a decreasing sequence
  of numbers between $1$ and $a_0$.  Thus, 
  \begin{align*}
    M_i:=\{j: 0 \leq j \leq l:  m_j=i\}    
  \end{align*}
  is a (possibly empty) interval for
  each $1 \leq i \leq a_0$ and the $M_i$ form a partition of 
  $\{j: 0\leq j \leq l\}$ into $a_0$ many sets. Since $\lvert G \rvert
  \geq a_0\cdot 2^k$, 
  there is at least one $1\leq i \leq a_0$
  such that $G_i:= G\cap M_i$ has at least $2^k$ many elements. 
  Set $b$ and $c$ to be the minimal  and maximal element,
  respectively, of $M_i$. By definition $a_b=i=\min\{a_j:j\in  M_i\}$,
  and
  $a_b<a_j$ for all $0\leq j < b$. 
  We can now apply Lemma \ref{lem:IncreasingNumberSequence} to
  $(a_j)_{j\in M_i}$ and $G_i\setminus \{c\}$. This shows the
  existence of some $e\leq l$ such that $e-1\in G_i\subseteq G$ and
  $\lvert H_{b,e} \rvert \geq k$. \qed
\end{proof}

We fix constants $c\geq 2$ and $m$ and define several sequences
which are parameterised by $c$ and $m$. In the next section, we will
always use $c=\lvert \Tt_\Ss\rvert+1$ and $m$ will be the length of a
fixed path in the graph of $\Ss$.
In the final part of this section,  we prove certain properties of
these sequences that we will use in the next section. 

\begin{definition} \label{def:SequencesMandN}
  \begin{enumerate}
  \item Set $M_1 :=(m+1)\cdot c$ and
    $M_j := 2^{M_{j-1}}$ for $j\geq 2$,
  \item set $M_1'= m \cdot c$ and
    $M_j' := 2^{M_{j-1}'}$ for $j\geq 2$, 
  \item set  $N'_0 :=  c$ and 
    $N'_j := M_j'\cdot 2^{N'_{j-1}}$ for $j \geq 1$, 
  \item set  $N_0 :=  c$ and 
    $N_j := M_j\cdot 2^{N_{j-1}}$ for $j \geq 1$, and
  \item set 
    $S_1:= (m+1) \cdot 3 \cdot c \cdot 2^c$ and
    $S_j:= 2^{S_{j-1}}$ for $j\geq 2$. 
  \end{enumerate}
\end{definition}

\begin{lemma} \label{lem:MminusMPrimeGeqNPrime}
  $M_i -M_i'\geq N'_{i-1}$ for all $i\geq 1$.
\end{lemma}
\begin{proof}
  The proof is by induction on $i$. 
  For $i=1$ we just have $M_1-M_1'= (m+1) \cdot c - m \cdot c = c= N'_0$. 
  For $i\geq 2$ we have 
  \begin{equation*}
    M_i-M_i'=2^{M_{i-1}}-2^{M_{i-1}'}=2^{M_{i-1}'}(2^{M_{i-1}-M_{i-1}'}-1)
    \geq
    2^{M_{i-1}'}(2^{N'_{i-2}}-1)    
  \end{equation*}
  where the inequality holds due to our
  induction hypothesis. 
  Since  $2^k\geq 2k$ for each $k\in \N$, we have
  \begin{equation*}
    2^{M_{i-1}'}(2^{N'_{i-2}}-1) \geq 2\cdot M_{i-1}'(2^{N'_{i-2}}-1).    
  \end{equation*}
  Furthermore,
  $N'_{i-2}\geq 1$ implies that $2^{N'_{i-2}}-1\geq 2^{N'_{i-2}-1}$.
  Hence, 
  \begin{equation*}
   2\cdot M_{i-1}'(2^{N'_{i-2}}-1) \geq 
   M_{i-1}' ( 2\cdot 2^{N'_{i-2}-1}) = 
   N'_{i-1}.      \tag*{\qed}
  \end{equation*}
\end{proof}

\begin{lemma} \label{lem:CombinatoricSandN}
  $S_j\geq 3N_j$ for all $j\geq 1$.
\end{lemma}
\begin{proof}
  The proof is by induction on $j$. 
  For $j=1$ we have 
  \begin{align*}
    S_1=(m+1)\cdot 3\cdot c \cdot 2^c = 3\cdot
    M_1\cdot 2^{N_0}=3 N_1.     
  \end{align*}
  Now assume that $j\geq 2$. 
  Since $N_{j-1}\geq N_0\geq 2$ holds, 
  $2^{N_{j-1}}\geq 3$.  
  We also have $N_{j-1}=M_{j-1}\cdot 2^{N_{j-2}}\geq M_{j-1}$, so
  $2^{N_{j-1}}\geq 2^{M_{j-1}}=M_j$.  
  Due to the induction hypothesis, we have
  \begin{equation}
    S_j = 2^{S_{j-1}}\geq 2^{3N_{j-1}}=2^{N_{j-1}}\cdot 2^{N_{j-1}}\cdot
    2^{N_{j-1}}\\
     \geq 3\cdot M_j\cdot 2^{N_{j-1}}=3N_j.  \tag*{\qed} 
  \end{equation}

\end{proof}


\section{Proof of Theorem \ref{thm:pumping}}
\label{app:PumpingLemma}
In this section we complete the proof of our main theorem. 
We start with several technical lemmas which connect behaviour of runs
described in terms of the stack sizes and in terms of the history
function (Subsection \ref{ssec:technical}). 
Then we give a lemma ensuring that a run $R$ satisfying a certain
technical condition has a subrun $S$ which decomposes into a pumping
run and a $\TOP{0}$-non-erasing run (Subsection \ref{ssec:pumping-tech}).
Next, in Subsection \ref{ssec:finitely-branching} we use this lemma in
order to give a bound on the size of  
stacks in finitely branching $\varepsilon$-contractions of
collapsible pushdown graphs: given such a finitely branching
contraction $\Gg$ and a 
configuration $c\in \Gg$ of distance $m$ from the initial
configuration,  there is a bound on the size of the stack of $c$ in
terms of $m$ and of the size and the level of the system $\Ss$
generating $\Gg$.   
Using this bound we derive a pumping construction 
in Subsection \ref{ssec:pumping-proof}
which proves the main theorem: 
for $\Gg$ and $c$ as before, if there is a path starting in $c$ of
length above some bound depending on $m$ and on the size and level of
$\Ss$, then there start infinitely many paths in $c$.

\subsection{Technical Lemmas}\label{ssec:technical}

\begin{lemma}  \label{lem:TopKandIncreasingSequencegiveTopk-1}
  Let $1\leq k\leq n$, let $R$ be a run,
  let $x$ be a position of a $k$-stack of $R(|R|)$, and $y$ the
  position of its topmost $(k-1)$-stack. 
  Let $a_i$ be the size of the $k$-stack of $R(i)$ 
  at $\hist{\subrun{R}{i}{|R|}}{x}$ for each $0 \leq i \leq |R|$. 
  Assume that $\hist{R}{x}=\TOP k(R(0))$ and that $a_0=\min\{a_i:
  0\leq i\leq |R|\}$. 
  Then $\hist{R}{y}=\TOP{k-1}(R(0))$.
\end{lemma}

Recall that the size of a $k$-stack is just the number of its $(k-1)$-stacks.
Before we prove the lemma, we state a proposition which is an immediate consequence of the definition of the history function 
(recall that for every position $x_1\posNew{k}{x_2}$ we have $x_2\not=(0,\dots,0)$).

\begin{proposition}
	Let $S$ be a run of length $1$, and $y$ a position of a $(k-1)$-stack in $S(1)$ for some $1\leq k\leq n$.
	Assume that the last non-zero coordinate of $y$ and $\hist{S}{y}$ differ.
	Then $\hist{S}{y}=\TOP{k-1}(S(0))$.
\end{proposition}

\begin{proof}[Lemma \ref{lem:TopKandIncreasingSequencegiveTopk-1}]
  Set $m:=|R|$. 
  Let $b_i$ be the value of the last nonzero coordinate of
  $\hist{\subrun{R}{i}{m}}{y}$ for each $0\leq i\leq m$  
  (notice that this is a level $k$ coordinate, as
  $\hist{\subrun{R}{i}{m}}{y}$ always points to a $(k-1)$-stack). 
  We claim that 
  \begin{enumerate}
  \item $b_i\geq a_0$ for each $0\leq i\leq m$ and
  \item $z:=\hist{R}{y}$ is simple.
  \end{enumerate}
  Due to Proposition \ref{prop:HistPreservesPointerContainment} $z$
  points into $\TOP k(R(0))$, and contains only links of level at most
  $k$.    
  If the claims hold, we conclude that $b_0 =  a_0$ whence
  $z=\TOP{k-1}(R(0))$. 

  We prove the first claim by induction on $i$ (from $m$ to $0$).
  For $i=m$ we have $b_m=a_m\geq a_0$.
  Let $i<m$.
  If $b_i=b_{i+1}$ we are done.
  By the above proposition, $b_i\not=b_{i+1}$ implies that
  $\hist{\subrun{R}{i}{m}}{y}=\TOP{k-1}(R(i))$. 
  Due to Corollary \ref{cor:positionContainment},
  $\hist{\subrun{R}{i}{m}}{y}$ points into
  $\hist{\subrun{R}{i}{m}}{x}$ whence
  $a_i$ is the size of the topmost $k$-stack of $R(i)$.
  Thus, $b_i=a_i\geq a_0$.

  For the second claim we assume that
  $z$ is not simple and derive a contradiction as follows.
  Since $z$ points to a $(k-1)$-stack, the last link in $z$ is
  of level (at least) $k$. 
  Recall the notation of the $\mathsf{pack}$ function from page
  \pageref{def:pack}. 
  Coordinates of level greater than $k$ in
  $\mathsf{pack}_{\NestingRank(z)}(z)$ are the same as in $\TOP
  k(R(0))$, 
  the level $k$ coordinate is $b_0\geq a_0$,
  and coordinates of level smaller than $k$ are zeroes.
  So $\mathsf{pack}_{\NestingRank(z)}(z)\lexOrdR\TOP{k-1}(R(0))$, and
  points to a $(k-1)$-stack. 
  By Corollary \ref{cor:GoodStacksHaveSmallPositions}
  $z=\TOP{k-1}(R(0))$ which contradicts our assumption that $z$ is not
  simple.  \qed
\end{proof}

\begin{lemma}\label{lem:ForNotEras}
	Let $1\leq k\leq n$, let $R$ be a run with $m=\lvert R
        \rvert$, $x$  a position of a $k$-stack of $R(|R|)$,  
	and $a_i$  the size of the $k$-stack of $R(i)$ at
        $\hist{\subrun{R}{i}{|R|}}{x}$ for each $0 \leq i \leq |R|$. 
	Assume that $a_0<a_i$ for all $0<i<|R|$, and
        $a_0\leq a_m$. 
	Then $R$ is a $\TOP0$-non-erasing run.
\end{lemma}
\begin{proof}
  If $m = 0$ there is nothing to show. Otherwise, 
  $a_0<a_1$ whence the first operation of $R$ has to be $\push k$.
  \begin{itemize}
  \item First assume that $a_m>a_0$ whence
    $a_1=\min\{a_i: 1\leq i\leq m\}$.
    Let $y$ be the topmost $(k-1)$-stack of the $k$-stack of
    $R(m)$ at $x$. 
    Application of Lemma
    \ref{lem:TopKandIncreasingSequencegiveTopk-1} to
    $\subrun{R}{1}{m}$, $x$ and $y$ implies that
    $\hist{\subrun{R}{1}{m}}{y}=\TOP{k-1}(R(1))$. 
    Due to Proposition \ref{prop:history2topk-non-eras} applied
    to $\subrun{R}{1}{m}$ and $k-1$ we obtain that  
    $\subrun{R}{1}{m}$ is a $\TOP{k-1}$-non-erasing run. 
    Since $\TOP0(R(0))\lexOrd\TOP{k-1}(R(1))$, $\TOP0(R(0))$
    cannot be removed by $R$ if $\TOP{k-1}(R(1))$ is not removed
    whence $R$ is a $\TOP0$-non-erasing run. 
  \item 
    Otherwise, $a_m=a_0$. 
    We apply the same argument as above, but to
    $\subrun{R}{0}{m-1}$. We obtain that $\subrun{R}{0}{m-1}$ is a
    $\TOP0$-non-erasing run. 
    Since $a_{m-1}>a_m=a_0$ and since only the topmost $k$-stack can
    change its size,  $x=\TOP k(R(|R|))$, and the operation between
    $R(m-1)$ 
    and $R(m)$ is $\pop k$ or $\col k$. 
    As the topmost $k$-stack of $R(m)$ has size $a_0$ and $\TOP
    0(R(0))$ is present in $R(m-1)$, it is also present in $R(m)$.\qed
  \end{itemize}
\end{proof}

Below we say that an $l$-stack $s$ \emph{occurs} in a $k$-stack $t$;
this includes occurring inside a link, and includes $s=t$.

\begin{lemma}\label{lem:ForClaim2}
  Let $0\leq j<k\leq n$, let $R$ be some run, $x$ some position of a
  $k$-stack in $R(|R|)$, and 
  $a_i$ the size of the $k$-stack at $\hist{\subrun{R}{i}{\lvert R
      \rvert}}{x}$ in $R(i)$.  
  Assume that $a_i>a_{\lvert R \rvert}$ for all $0\leq i<\lvert R \rvert$.
  Then every $j$-stack occurring in the $k$-stack at $x$ in $R(\lvert
  R \rvert)$ occurs also in the $k$-stack at $\hist{R}{x}$ in $R(0)$.  
\end{lemma}
\begin{proof}
  It is enough to prove, for $1\leq b\leq a_{|R|}$ that the $b$-th
  $(k-1)$-stack of the $k$-stack at $x$ in $R(\lvert R \rvert)$ is
  equal to  
  the $b$-th $(k-1)$-stack of the $k$-stack at $\hist{R}{x}$ in $R(0)$.
  We prove this by induction on the length of $R$.
  For $|R|=0$ this is immediate.
  Let $|R|\geq 1$.
  In the light of the induction assumption,
  it is enough to prove, for $1\leq b\leq a_{|R|}$, that the $b$-th
  $(k-1)$-stack of the $k$-stack at $\hist{\subrun{R}{1}{|R|}}{x}$ in
  $R(1)$ is equal to  
  the $b$-th $(k-1)$-stack of the $k$-stack at $\hist{R}{x}$ in $R(0)$.
  \begin{itemize}
  \item	If the first operation in $R$ is of level below $k$, then
    $\hist{\subrun{R}{1}{|R|}}{x}=\hist{R}{x}$ and only the topmost
    $(k-1)$-stack is modified;  
    this is not one of the considered $(k-1)$-stacks, as
    $a_0>a_{|R|}$. 
  \item	If the first operation in $R$ is of level $k$, then
    $\hist{\subrun{R}{1}{|R|}}{x}=\hist{R}{x}$ and some $(k-1)$-stacks
    are removed or added, but none of the considered $(k-1)$-stacks,
    as $a_0,a_1\geq a_{|R|}$ 
    (this is also true for $\col k$, as performing $\col k$ is
    equivalent to performing several $\pop k$). 
  \item	If the first operation in $R$ is of level greater than $k$,
    then the whole $k$-stacks at $\hist{\subrun{R}{1}{|R|}}{x}$ in
    $R(1)$ and at $\hist{R}{x}$ in $R(0)$ are the same. \qed
  \end{itemize}	
\end{proof}

\begin{lemma}\label{lem:tech4}
  Let $1\leq k\leq n$, $R$  a pumping run and $x$ a
  position of a $k$-stack in $R(|R|)$. 
  Assume that the size of the $k$-stack  at $x$ in  $R(|R|)$ is
  greater than that of the $k$-stack at
  $\hist{\subrun{R}{i}{|R|}}{x}$ in $R(i)$ for some $i$. 
  Then $\TOP0(R(0))\neq\TOP0(R(|R|))$. 
\end{lemma}

\begin{proof}
  Let $m:=|R|$, and let $a_j$ be the size of the $k$-stack of $R(j)$
  at $\hist{\subrun{R}{j}{|R|}}{x}$ for all $0\leq j\leq m$. 
  Take the maximal $b$ such that $a_b < a_m$ (note that $i\leq b$).
  Since stack operations increase the number of stacks by at most one, 
  $a_{b+1}=a_b + 1$ whence maximality of $b$ implies
  $a_{b+1}=\min\{a_j:b+1\leq j\leq m\}$.
  Set $S:=\subrun{R}{b+1}{m}$.
  Notice that $\hist{S}{x}=\TOP k(S(0))$ because only the topmost
  $k$-stack can change its size. 
  Since $\hist{R}{\TOP0(R(m))}=\TOP 0(R(0))$, Proposition
  \ref{prop:history2topk-non-eras} implies that $R$ is a
  $\TOP0$-non-erasing run. 
  It means that $\TOP{k-1}(R(0))$ is present in $R(b)$. 
  Because the operation between $R(b)$ and $R(b+1)=S(0)$ is
  necessarily $\push k$, 
  it implies that $\TOP0(R(0))\lexOrdstrict\TOP{k-1}(S(0))$. 
	
  Application of Lemma \ref{lem:TopKandIncreasingSequencegiveTopk-1} to
  $S$ and $x$ shows that $\hist{S}{y}=\TOP{k-1}(S(0))$ for some
  position $y$ in $R(m)$. Again using Proposition
  \ref{prop:history2topk-non-eras}, we conclude that $S$ is a
  $\TOP{k-1}$-non-erasing run and we obtain
  \begin{equation*}
    \TOP{0}(R(0)) \lexOrdstrict \TOP{k-1}(S(0)) \lexOrd \TOP{k-1}(R(m))
    \lexOrdstrict \TOP{0}(R(m)).     \tag*{\qed}
  \end{equation*}
\end{proof}

\subsection{Main Technical Lemma}\label{ssec:pumping-tech}

Below we present our main technical lemma.
It shows how to find subruns of long runs which consist of a pumping
run followed by a $\TOP{0}$-non-erasing run. 
Recall that the function $\ctype_\Xx$ maps
configurations to a finite set of types. 
For each collapsible pushdown system $\Ss$,
let $\Tt_\Ss$ denote the image of $\ctype_\Xx$ with respect to
configurations of $\Ss$.

\begin{lemma}\label{lem:Pumpability} 
  Let $\Ss$ be an $n$-CPS,
  $0\leq k\leq n$, 
  $R$ be a run of $\Ss$, and 
  \begin{align*}
    G_k\subseteq \{i\ < \lvert R \rvert: \hist{\subrun{R}{i}{\lvert
        R \rvert}}{\TOP{k}(R(\lvert R \rvert))}= \TOP{k}(R(i))\}.    
  \end{align*}
  Furthermore, let $s^k$ be the $k$-stack of $R(0)$ to which
  $\hist{R}{\TOP{k}(R(\lvert R \rvert))}$ points. 
  For $1\leq j\leq k$, let $r_j$ be the maximum of the sizes of $j$-stacks occurring in $s^k$.  
  Let
  \begin{align*}
    \hat N_0 := \lvert\Tt_\Ss \rvert +1 \text{ and } \hat N_j=r_j\cdot
    2^{\hat N_{j-1}} \text{ for } 1\leq j\leq k.  
  \end{align*}
  If  $\lvert G_k\rvert \geq \hat N_k$,
  then there are $0\leq x<y<z\leq \lvert R \rvert$ such that  
  \begin{enumerate} 
  \item	\label{pumpability:1} $\ctype_\Xx(R(x))=\ctype_\Xx(R(y))$, 
  \item	\label{pumpability:2} 
  	$\subrun{R}{x}{y}$ is a pumping run,
  \item \label{pumpability:3} $\hist{\subrun{R}{y}{\lvert R \rvert}}{\TOP{k}(R(\lvert R \rvert))}=\TOP{k}(R(y))$,
  \item	\label{pumpability:4} $\TOP0(R(x))\neq\TOP0(R(y))$ or
    \begin{align*}
      G_k\cap\{x,x+1,\dots,y-1\} \neq \emptyset,      
    \end{align*}
  \item	\label{pumpability:5} $z-1\in G_k$, and
  \item \label{pumpability:6} $\subrun{R}{y}{z}$ is a $\TOP0$-non-erasing run.
  \end{enumerate} 
\end{lemma} 

\begin{proof}
  We prove the lemma by induction on $k$. 
  Consider the case that $k=0$. 
  By assumption $\lvert G_0 \rvert \geq \lvert \hat N_0 \rvert > \lvert
  \Tt_\Ss\rvert$. 
  Thus, there are $x,y\in G_0$ with $x<y$ 
  such that $\ctype_\Xx(R(x))=\ctype_\Xx(R(y))$. 
  Since $x,y\in G_0$, 
  \begin{align*}
    &\hist{\subrun{R}{x}{\lvert
        R \rvert}}{\TOP{0}(R(\lvert R \rvert))}= \TOP{0}(R(x)),\text{ and}\\
    &\hist{\subrun{R}{y}{\lvert
        R \rvert}}{\TOP{0}(R(\lvert R \rvert))}= \TOP{0}(R(y)).    
  \end{align*}
  Due to Proposition \ref{prop:histTransitive}, we conclude
  that 
  \begin{align*}
    \hist{\subrun{R}{x}{y}}{\TOP{0}(R(y))}=\TOP{0}(R(x))
  \end{align*}
  which means that $\subrun{R}{x}{y}$ is a pumping run.
  Since $x\in G_0$, we have $G_0\cap \{x, x+1, \dots,
  y-1\}\neq\emptyset$. 
  By definition of $y$, $\subrun{R}{y}{\lvert R \rvert}$ is a pumping
  run of length at least $1$. Due to the characterisation of pumping
  runs (cf. Lemma \ref{lem:pumping-run}), this run starts with some
  $\push{}$ operation. 
  Thus, for  $z:=y+1$, we have $z-1\in G_0$ and $\subrun{R}{y}{z}$ is a $\TOP0$-non-erasing run. 
  Thus, $x$, $y$, and $z$ satisfy the claim of the lemma. 
  
  Now consider the case $k\geq 1$ and assume that the lemma holds for
  all $k'<k$.  
  Let $a_i$ be the size of the $k$-stack of $R(i)$ at
  position $\hist{\subrun{R}{i}{\lvert R \rvert}}{\TOP{k}(R(\lvert R
    \rvert))}$ for $0\leq i\leq \lvert R \rvert$. 
  Due to 
  Proposition 
 \ref{prop:NonTopmostStacksDoNotChange} we know that
  $a_i-a_{i-1}\leq 1$ for all $1\leq i \leq \lvert R \rvert$. 
  By definition $a_0\leq r_k$ whence
  $\lvert G_k \rvert \geq \hat{N}_k\geq a_0\cdot
  2^{\hat{N}_{k-1}}$. 
  Hence, we can apply Corollary \ref{cor:IncreasingNumberSequence} to
  $(a_i)_{0\leq i \leq \lvert R \rvert}$  and obtain 
  indices $0\leq b < e \leq \lvert R \rvert$ such that
  \begin{enumerate}
  \item $e-1\in G_k$, 
  \item $a_b = \min\{a_i: b\leq i \leq e\}$,
  \item $a_i>a_b$ for all $0\leq i < b$ and
  \item $\lvert H_{b,e} \rvert \geq \hat N_{k-1}$ where
    \begin{align*}
      H_{b,e} = 
      \{ i:  &\ b \leq i \leq e-1,\\
      &a_i \leq a_j\text{ for all
      } i\leq j \leq e
      \text{, and }\\
      &a_i < a_j\text{ for all } i<j \leq n_{G_k}(i)\}
    \end{align*}
    with $n_{G_k}(i):=\min\{ g\in G_k: g\geq i\}$.
  \end{enumerate}
  Set $R':=\subrun{R}{b}{e}$ and $G_{k-1}:=\{h-b:h\in H_{b,e}\}$. 
  Let us first assume that the following claims are true:
  \begin{enumerate}[A)]
  \item for each $h\in H_{b,e}$ we have $\hist{\subrun{R}{h}{e}}{\TOP{k-1}(R(e))}= \TOP{k-1}(R(h))$,
  \item for all $i\leq e-1$, 
    $\hist{\subrun{R}{i}{e}}{\TOP{k}(R(e))}= \hist{\subrun{R}{i}{\lvert R \rvert}}{\TOP{k}(R(\lvert R \rvert))}$,
    whence $a_i$ is the size of the $k$-stack in $R(i)$ at $\hist{\subrun{R}{i}{e}}{\TOP{k}(R(e))}$, and
  \item   if $t^{k-1}$ is the $(k-1)$-stack at 
    $\hist{R'}{\TOP{k-1}(R'(\lvert R' \rvert))}$, then 
    the  size of every
    $j$-stack occurring in $t^{k-1}$ for $j \leq k-1$ is bounded by $r_j$. 
  \end{enumerate}
  We postpone the proof of these claims. 
  Claim A implies (by shifting from $R$ to $R'$) that for each $g\in
  G_{k-1}$ we have $\hist{\subrun{R'}{g}{\lvert R'
      \rvert}}{\TOP{k-1}(R'(\lvert R' \rvert))}= \TOP{k-1}(R'(g))$. 
  Together with Claim C this allows us to apply 
  the induction hypothesis to $k-1$, $R'$ and $G_{k-1}$.  
  We obtain three indices $0\leq x'<y'<z'\leq|R'|$; let 
  $x=x'+b$, $y=y'+b$, and
  let $z$ be the smallest index such that $z\geq z'+b$ and $z-1\in G_k$
  (it exists because $z'+b\leq e$ and $e-1\in G_k$).  
  Note that 
  \begin{enumerate}[1'.]
  \item	\label{pumpability:1Prime} $\ctype_\Xx(R(x))=\ctype_\Xx(R(y))$, 
  \item	\label{pumpability:2Prime} $\subrun{R}{x}{y}$ is a pumping run,
  \item \label{pumpability:3Prime} $\hist{\subrun{R}{y}{e}}{\TOP{k-1}(R(e))}=\TOP{k-1}(R(y))$,
  \item	\label{pumpability:4Prime} $\TOP0(R(x))\neq\TOP0(R(y))$ or
    \begin{align*}
      H_{b,e}\cap\{x,x+1,\dots,y-1\} \neq \emptyset,      
    \end{align*}
  \item	\label{pumpability:5Prime} $z'+b-1\in H_{b,e}$, and
  \item \label{pumpability:6Prime} $\subrun{R}{y}{z'+b}$ is a $\TOP0$-non-erasing run.  
  \end{enumerate} 
  Note that items \ref{pumpability:1Prime}' and
  \ref{pumpability:2Prime}' coincide with items \ref{pumpability:1}
  and \ref{pumpability:2} of the lemma. 
  We now prove items \ref{pumpability:3} -- \ref{pumpability:6}.
  \begin{enumerate}
    \setcounter{enumi}{2}
  \item   Due to Corollary \ref{cor:positionContainment}, item
    \ref{pumpability:3Prime}'
    implies that
    \begin{align*}
      \hist{\subrun{R}{y}{e}}{\TOP{k}(R(e))}=\TOP{k}(R(y)).    
    \end{align*}
    Together with Claim B this yields item \ref{pumpability:3}.
  \item 
    Assume that $\TOP0(R(x))=\TOP0(R(y))$. 
    Note that this directly implies 
    \begin{align}\label{condition-k-1-equals}
      \TOP{k}(R(x))=\TOP{k}(R(y) \text{ for all } 0\leq k \leq \sL.
    \end{align}
    Due to \ref{pumpability:4Prime}', there is some 
    \begin{align*}
      h\in H_{b,e}\cap\{x,x+1,\dots,y-1\} \neq \emptyset.    
    \end{align*}
    Items \ref{pumpability:2Prime}', \ref{pumpability:3Prime}', and
    Claim A, after application of Corollary
    \ref{cor:positionContainment}, imply 
    \begin{align*}
      \hist{\subrun{R}{x}{y}}{\TOP{j}(R(y))}&=\TOP j(R(x)),\\
      \hist{\subrun{R}{y}{e}}{\TOP{j}(R(e))}&=\TOP j(R(y)),\quad\mbox{and}\\
      \hist{\subrun{R}{h}{e}}{\TOP{j}(R(e))}&=\TOP j(R(h))
    \end{align*}
    for all $j\geq k-1$. 
    Due to Proposition \ref{prop:histTransitive}, this implies that 
    \begin{align}\label{eq:hist-ok}
      \hist{\subrun{R}{a}{b}}{\TOP{j}(R(b))}=\TOP j(R(a))
    \end{align} 
    for each pair $a,b\in\{x,h,y,e\}$ with $a\leq b$.
    With two applications of Lemma \ref{lem:PositionsDecrease} (to
    $\subrun{R}{h}{y}$ and $\subrun{R}{x}{h}$) we obtain 
    that $\TOP{k-1}(R(h))$ is lexicographically bounded by
    $\TOP{k-1}(R(x))=\TOP{k-1}(R(y))$ from below and from above
    whence it is this position (the equality of the two positions
    comes from equation \eqref{condition-k-1-equals}).
    Claim B and equation \eqref{eq:hist-ok} (setting $j=k$) imply
    that $a_x$, $a_h$ and $a_y$ are the sizes of the topmost
    $k$-stacks of $x$, $h$ and $y$, respectively. 
    It follows that $a_x=a_h=a_y$. 
    Since $h\in H_{b,e}$, 
    there exists some $g\in G_k$ such that
    $x\leq h\leq g$ and $a_j > a_h$ for all $h< j \leq g$.
    As $a_y = a_h$, we conclude that $g < y$ whence
    $G_k\cap \{x, x+1, \dots, y-1\} \neq \emptyset$.
  \item  $z-1\in G_k$ is satisfied by definition of $z$.
  \item If $z=z'+b$, items \ref{pumpability:6} and
    \ref{pumpability:6Prime}' coincide. 
    Assume that $z>z'+b$.
    Because $z'+b-1\in H_{b,e}$, we know that $a_j > a_{z'+b-1}$ for
    $z'+b\leq j\leq z-1$ because $z$ is minimal such that $z-1\geq z'+b-1$ and $z-1\in G_k$.  
    In particular, $z>z'+b$ implies $a_{z'+b} > a_{z'+b-1}$.
    Recall that $z$ was chosed to satisfy $z\leq e$. This together
    with $z'+b-1\in H_{b,e}$ implies that $a_z\geq a_{z'+b-1}$. 
    Thus, Lemma \ref{lem:ForNotEras} can be applied to $\subrun{R}{z'+b-1}{z}$.
    It follows that $\subrun{R}{z'+b-1}{z}$ is a $\TOP0$-non-erasing run.
    Since $\subrun{R}{y}{z'+b}$ is also a $\TOP0$-non-erasing run,
    $\subrun{R}{y}{z}$ is one as well (cf.~Proposition
    \ref{prop:topk-non-eras-composition}).  
  \end{enumerate}
  Thus, $x,y$ and $z$ satisfy the lemma if Claims A -- C hold. 
  We continue with a simultaneous proof of Claims A and B. 
  We start with showing that
  for each $h\in H_{b,e}$
  \begin{align}\label{eq:k-dla-hR}
     \hist{\subrun{R}{h}{\lvert R  \rvert}}{\TOP{k}(R(\lvert R \rvert))}=\TOP{k}(R(h)).    
  \end{align}
  Consider any $h\in H_{b,e}$. 
  If $h\in G_k$, the condition is satisfied by definition of $G_k$. 
  Otherwise, we conclude that $a_{h+1} > a_h$ by definition of
  $H_{b,e}$ and the fact that $n_{G_k}(h)\geq h+1$.
  But only the topmost $k$-stack can change its size whence equation
  (\ref{eq:k-dla-hR}) holds. 

  Recall that $e-1\in G_k$, which implies that
  \begin{align}\label{e-1isTOPk}
    \hist{\subrun{R}{e-1}{\lvert R
        \rvert}}{\TOP{k}(R(\lvert R \rvert))} = \TOP{k}(R(e-1)).
  \end{align}
  Together with equation (\ref{eq:k-dla-hR}) this implies
  \begin{align*}
    &\hist{\subrun{R}{h}{e-1}}{\TOP{k}(R(e-1))}=\\
    &\hspace{0.5cm}=\hist{\subrun{R}{h}{ \lvert R \rvert}}{\TOP{k}(R(\lvert R \rvert))}=\TOP{k}(R(h))    
  \end{align*}
  for each $h\in H_{b,e}$. 
  By definition of $H_{b,e}$, \mbox{$a_h=\min\{a_i: h\leq i \leq e\}$}.
  Additionally, equation (\ref{e-1isTOPk}) implies that $a_i$ (for
  $b\leq i\leq e-1$) is the size of the $k$-stack of $R(i)$ at
  $\hist{\subrun{R}{i}{e-1}}{\TOP k(R(e-1))}$, 
  whence we may apply Lemma
  \ref{lem:TopKandIncreasingSequencegiveTopk-1} to $x:=\TOP k(R(e-1))$
  and to $\subrun{R}{h}{e-1}$. 
  This yields 
  \begin{align} \label{eqn:TOPK-1e-1}
    \hist{\subrun{R}{h}{e-1}}{\TOP{k-1}(R(e-1))}=\TOP{k-1}(R(h))
  \end{align}
  for each $h\in H_{b,e}$.
  
  We continue by case distinction on the operation between $e-1$
  and $e$ in $R$.
  \begin{enumerate}
  \item Due to equation (\ref{e-1isTOPk}),
    the operation at $e-1$ cannot be $\pop{k'}$ or $\col{k'}$ for
    $k'>k$.
  \item If the operation at $e-1$ is of level below $k$ or is a $\push{}$
    operation, then 
    \begin{align*}
      \hist{\subrun{R}{e-1}{e}}{\TOP{k-1}(R(e))}=\TOP{k-1}(R(e-1)).
    \end{align*}
    Due to equation (\ref{eqn:TOPK-1e-1}), this implies Claim A.
    Together with (\ref{e-1isTOPk}) and Corollary
    \ref{cor:positionContainment}, this implies 
    $$\hist{\subrun{R}{e-1}{e}}{\TOP
      k(R(e))}=\hist{\subrun{R}{e-1}{|R|}}{\TOP{k}(R(|R|))}.$$ 
    Using Proposition \ref{prop:histTransitive}, Claim B follows directly. 
  \item Assume that the operation at
    $e-1$ is $\pop{k}$ or $\col{k}$. 
    We conclude immediately that
    \begin{align*}
      \hist{\subrun{R}{e}{\lvert R
          \rvert}}{\TOP{k}(R(\lvert R \rvert))}=\TOP{k}(R(e)),      
    \end{align*}
    because this is the only position $p$ of $R(e)$ that satisfies
    $\hist{\subrun{R}{e-1}{e}}{p}=\TOP{k}(R(e-1))$ (and because $e-1\in
    G_k$). 
    With Proposition \ref{prop:histTransitive}, Claim B follows directly. 
    
    Furthermore, $a_e$ is the size of the stack of $R(e)$ at $\TOP{k}(R(e))$. 
    By definition of $H_{b,e}$,  we have $a_h=\min\{a_i:h\leq i\leq e\}$.
    Application of Lemma \ref{lem:TopKandIncreasingSequencegiveTopk-1}
    to $x:=\TOP k(R(e))$ and to $\subrun{R}{h}{e}$ for each $h\in H_{b,e}$ yields Claim A.  
  \end{enumerate}
  For the proof of Claim C, 
  let $t^{k-1}$ be the $(k-1)$-stack of $R'(0)$ at the position
  $\hist{R'}{\TOP{k-1}(R'(\lvert R' \rvert))}$
  which is by definition the
  $(k-1)$-stack of $R(b)$ at
  $\hist{\subrun{R}{b}{e}}{\TOP{k-1}(R(e))}$.
  Due to Corollary \ref{cor:positionContainment}, 
  $\hist{\subrun{R}{b}{e}}{\TOP{k-1}(R(e))}$ points into
  $\hist{\subrun{R}{b}{e}}{\TOP{k}(R(e))}$.
  Hence, for $j\leq k-1$  every $j$-stack occurring in $t^{k-1}$
  occurs also in the $k$-stack of $R(b)$ at
  $\hist{\subrun{R}{b}{e}}{\TOP{k}(R(e))}$. 
  Due to Claim B,
  $a_i$ is the number of $(k-1)$-stacks of the stack at 
  $\hist{\subrun{R}{i}{e}}{\TOP{k}(R(e))}$ for all $i\leq b$, and the
  $k$-stack of $R(0)$ at $\hist{\subrun{R}{0}{e}}{\TOP k(R(e))}$ is
  $s^k$. 
  We have $a_i>a_b$ for all $0\leq i<b$, so we can apply  Lemma
  \ref{lem:ForClaim2} to $\subrun{R}{0}{b}$ and 
  position $\hist{\subrun{R}{b}{e}}{\TOP{k}(R(e))}$.  
  We conclude that  for $j\leq k-1$ every $j$-stack occurring in
  $t^{k-1}$ occurs also in $s^k$. Thus, its size is bounded by $r_j$. \qed
\end{proof}

\subsection{Finitely Branching
  Epsilon-Contractions}\label{ssec:finitely-branching} 

The basic proof concept for the pumping lemma is as follows. 
If we find a pumping run which starts and ends in configurations of the
same type, then we can apply Proposition \ref{prop:types1} to this run and
obtain arbitrarily long runs in the graph of the
CPS. But if we consider $\varepsilon$-contractions, all runs that we
construct may consist of $\varepsilon$-edges except for a bounded
number of transitions. In this case, the longer and longer runs would
perhaps always induce the same path in the $\varepsilon$-contraction.
In this section we show how to overcome this problem in the case of
finitely branching $\varepsilon$-contractions. 

We first derive a technical condition that allows to
conclude that the $\varepsilon$-contraction of some collapsible
pushdown graph is infinitely branching. 
This result basically uses the naive pumping approach described before
but we add some assumptions such that we really obtain larger and
larger runs that end in larger and larger stacks that belong to the
nodes of the $\varepsilon$-contraction. 
Afterwards, we use this result
in order to define a bound on the difference of stack sizes between two
nodes of a finitely branching $\varepsilon$-contraction that are
connected by an edge. In the next section we use this fact in the
pumping construction in the following way: instead of talking about a
configuration being in some distance from the initial
configuration in the 
$\varepsilon$-contraction, we talk about a configuration having stack
sizes bounded by some numbers. 

Without loss of generality (by doubling the number of states of the system), we can assume that for each state $q$ 
transitions leading to  state $q$ are all $\varepsilon$-transitions or are all non-$\varepsilon$-transitions.

\begin{proposition} \label{prop:inf-branching}
  Let $\Ss$ be some CPS of level $n$ such that for each state $q$     
  transitions leading to  state $q$ are all $\varepsilon$-transitions or are all non-$\varepsilon$-transitions.
  Let $R$ be a run starting in a configuration of the
  $\varepsilon$-contraction $\Gg$ of the graph of $\Ss$. 
  Then $\Gg$ is infinitely branching
  if there are positions $0\leq x < y \leq \lvert R \rvert$ such that
  \begin{enumerate}
  \item $\ctype_\Xx(R(x))=\ctype_\Xx(R(y))$,
  \item $\subrun{R}{x}{y}$ is a pumping run in
    $\Pp_{>,\varepsilon}$, i.e., a pumping run such that
    $\TOP0(R(x))\lexOrdstrict\TOP0(R(y))$ and all edges of
    $\subrun{R}{x}{y}$ are labelled by $\varepsilon$, and 
  \item $\subrun{R}{y}{\lvert R \rvert}$ is a $\TOP{0}$-non-erasing
    run ending with a non-$\varepsilon$-transition.
  \end{enumerate}

\end{proposition}
\begin{proof}
  Let $q$ be the state of $R(y)$. 
  We apply Proposition \ref{prop:types1} to $\subrun{R}{x}{y}$
  and obtain infinitely many $\varepsilon$-labelled runs
  $(R_i)_{i\in\N}$ from $R(x)$ to $c_i=(q,s_i)$ such that $\TOP0(s_i)\lexOrdstrict\TOP0(s_{i+1})$ for all $i\in\N$. 
  Now we apply Proposition \ref{prop:types2} to 
  $\subrun{R}{y}{\lvert R \rvert}$ and to $c_i$. We obtain a
  $\TOP{0}$-non-erasing run 
  $S'_i$ from $c_i$.
  It ends in the same state as $R$ whence it ends with a
  non-$\varepsilon$-transition.  
  Let $S_i$ be the prefix of $S'_i$ which ends after the first
  occurrence of a non-$\varepsilon$-transition.
  Let $z\leq x$ be maximal such that $R(z)$ corresponds to a node of $\Gg$. 
  Then $U_i:=\subrun{R}{z}{y} \circ R_i \circ S_i$ connects $R(z)$
  to one of its successors in $\Gg$ whose stack $t_i$ contains the
  position $\TOP{0}(s_i)$. 
  Since $t_i$ contains only finitely many positions, and the
  $(\TOP{0}(s_j))_{j\geq i}$ form an infinite sequence
  of pairwise distinct positions, for each $i$ there is a $j\geq i$ such that
  $\TOP{0}(s_j)$ is no position in $t_i$. This immediately implies
  $t_i\neq t_j$. By induction, we conclude that the $U_i$ connect
  $R(z)$ with infinitely many pairwise different successors in $\Gg$
  whence $\Gg$ is infinitely branching at $R(z)$. \qed
\end{proof}

Now we are prepared to prove that in each finitely branching
$\varepsilon$-contraction of a collapsible pushdown system the stack
sizes grow only in a bounded manner  from each node to its successors. 
For the combinatorial part in the proof we use the
sequences from Definition \ref{def:SequencesMandN} without reference. 

\begin{lemma}\label{lem:Bounded-Stack-Growth} 
  Let $\Ss$ be a CPS of level $n$ such that
  the $\varepsilon$-contraction $\Gg$ of its configuration graph 
  is finitely branching and such that 
  transitions leading to some state $q$ are either all
  $\varepsilon$-transitions or all
  non-$\varepsilon$-transitions. 
  Set $c:=\lvert \Tt_\Ss\rvert+1$.
  Let $R$ be a run starting in the initial configuration
  whose last edge is not labelled by $\varepsilon$ 
  and which corresponds to a path of length $m$ in $\Gg$. 
  The size of every $k$-stack of $R(\lvert R \rvert)$
  is at most $M_k$ for all 
  $1\leq k\leq n$.
\end{lemma} 
 
\begin{proof} 
  The proof is by induction on $m$. For $m=0$, the claim is trivial
  (because $c\geq 2$ and the initial stack of any level has size
  $1$). 
  Assume that we have proven the claim for all paths of length below
  $m$ and assume that $R$ describes a path of length $m$ in $\Gg$. 
  Let $R(b)$ correspond to the $(m-1)$-st node of $\Gg$ on this path and 
  set $S:=\subrun{R}{b}{\lvert R\rvert}$. 

  Heading for a contradiction assume that $p$ is the position of a
  $k$-stack in $R(\lvert R \rvert)$ such that the size of this stack
  is greater than $M_k$.

  For $0\leq i\leq \lvert S \rvert$, let $a_i$ be the size of the
  $k$-stack of $S(i)$ at $\hist{\subrun{S}{i}{\lvert S \rvert}}{p}$. 
  By induction hypothesis, the size of any $k$-stack 
  of $S(0)$ is bounded by $M'_k$.
  Thus, we have $a_{\lvert S \rvert}> M_k$ and $a_0\leq M_k'$. 
  Let  $G\subseteq\{0, 1, \dots, \lvert S \rvert-1\}$ contain all
  elements $i$ such that $a_i < a_j$ for all $i<j\leq \lvert S \rvert$.
  Since  $ a_i - a_{i-1} \leq 1$ for $1\leq i\leq \lvert S\rvert$, 
  for each $i$ in $\{M_k', M_k'+1, \dots, M_k\}$ we have an index $j$
  such that $a_j=i$ and $a_j\in G$. 
  Using Lemma \ref{lem:MminusMPrimeGeqNPrime} we conclude that
  $\lvert G \rvert \geq M_k - M_k' \geq N_{k-1}'$. Since $M_k'$ is a
  bound on the sizes of $k$-stacks in $S(0)$, it follows that $G$ is
  big enough in order to apply Lemma \ref{lem:Pumpability}  for $k-1$. 
  We want to apply this lemma to the run 
  $T:=\subrun{S}{0}{\max(G)+1}$.

  In order to satisfy the requirements of this lemma, we have to prove
  that 
  $\hist{\subrun{T}{g}{\lvert T \rvert}}{\TOP{k-1}(T(\lvert T \rvert))}
    = \TOP{k-1}(T(g))$
  for all $g\in G$. 
  Choose  $g\in G$ arbitrarily. 
  Since $a_{g+1}>a_g$, the $k$-stack at 
  $\hist{\subrun{S}{g}{\lvert S \rvert}}{p}$
  is smaller than that at
  $\hist{\subrun{S}{g+1}{\lvert S \rvert}}{p}$.
  Due to Proposition \ref{prop:NonTopmostStacksDoNotChange}, 
  this requires that 
  \begin{align*}
    &\hist{\subrun{S}{g}{\lvert S \rvert}}{p}=\TOP{k}(S(g)),\quad \text{and}\\
    &\hist{\subrun{S}{g+1}{\lvert S \rvert}}{p}=\TOP{k}(S(g+1)).
  \end{align*}
  Especially, $\hist{\subrun{S}{\lvert T \rvert}{\lvert S \rvert}}{p}=\TOP{k}(T(\lvert T \rvert))$ whence
  $a_i$ for $i\leq \lvert T \rvert$ is the size of the stack at
  $\hist{\subrun{T}{i}{\lvert T \rvert}}{\TOP{k}(T(\lvert T
    \rvert))}$. 
  We conclude that for all $g\in G$ we have
  \begin{align*} 
  \hist{\subrun{T}{g}{\lvert T \rvert}}{\TOP{k}(T(\lvert T \rvert))}=\TOP{k}(T(g)).
  \end{align*}
  Furthermore, for each $i> g$ we have $a_g < a_i$
  whence we can apply 
  Lemma \ref{lem:TopKandIncreasingSequencegiveTopk-1} to $x:=\TOP
  k(T(|T|))$ and to the run $\subrun{R}{g}{|T|}$ obtaining
  that 
  \begin{align*}
    \hist{\subrun{T}{g}{\lvert T \rvert}}{\TOP{k-1}(T(\lvert T
      \rvert))}=\TOP{k-1}(T(g)).    
  \end{align*}
  Application of Lemma \ref{lem:Pumpability} to $T$ and $k-1$
  yields indices $0\leq x < y <  z \leq \lvert T \rvert$ such that   
  \begin{enumerate} 
  \item	$\ctype_\Xx(S(x))=\ctype_\Xx(S(y))$, 
  \item	$\subrun{S}{x}{y}$ is a pumping run,
    \setcounter{enumi}{3}
  \item \label{prf:bounded-stack-growth-Enum-4} 
    $\TOP0(S(x))\neq\TOP0(S(y))$ or
    \begin{align*}
    G\cap\{x,x+1,\dots,y-1\} \neq \emptyset,      
    \end{align*}
  \item $z-1\in G$, and
  \item \label{prf:bounded-stack-growth-Enum-6} 
    $\subrun{S}{y}{z}$ is a $\TOP0$-non-erasing run.
  \end{enumerate} 

  By definition of $G$, for $g\in G$ and $g<y$ we have $a_g<a_y$.
  In other words the size of the $k$-stack in $S(y)$ at
  $\hist{\subrun{S}{y}{|S|}}{p}$ is greater than that of the $k$-stack in
  $S(g)$ at $\hist{\subrun{S}{g}{|S|}}{p}$. 
  Thus, if there is a $g\in G\cap\{x,x+1,\dots,y-1\}$, 
  application of Lemma \ref{lem:tech4} to $\subrun{S}{x}{y}$ shows that
  $\TOP0(S(x))\neq\TOP0(S(y))$. 
  In the light of Property \ref{prf:bounded-stack-growth-Enum-4}), we
  always have 
  $\TOP0(S(x))\neq\TOP0(S(y))$. 

  Since $z-1\in G$, we have $a_i>a_{z-1}$ for all $z\leq i\leq|S|$. 
  Application of Lemma \ref{lem:ForNotEras} to $\subrun{S}{z-1}{|S|}$
  shows that $\subrun{S}{z-1}{|S|}$ is a $\TOP0$-non-erasing run. 
  Since Property \ref{prf:bounded-stack-growth-Enum-6}) implies that
  $\subrun{S}{y}{z-1}$ is 
  $\TOP{0}$-non-erasing, we conclude using 
  Proposition
  \ref{prop:topk-non-eras-composition}
  that
  $\subrun{S}{y}{\lvert S \rvert}$ is  $\TOP0$-non-erasing run. 
  
  Recall that the last edge of $S$ is the only edge which is not
  labelled $\varepsilon$. 
  Thus the assumptions of 
  Proposition \ref{prop:inf-branching} are satisfied by
  $S$, $x$ and $y$ whence the lemma yields that
  $\Gg$ is infinitely branching. 
  This contradicts our assumption. 
  Thus, we conclude that every $j$-stack in  $R(\lvert R \rvert)$
  has size bounded by $M_j$. \qed
\end{proof}

\subsection{Proof of the Pumping Lemma}\label{ssec:pumping-proof}

Having bounded the size of stacks in finitely branching
$\varepsilon$-contractions of pushdown graphs, we can prove the main
theorem.  

Note that -- doubling the number of states of the system -- we can enforce
that for each $\varepsilon$-transition $\delta_1$ and each
non-$\varepsilon$-transition $\delta_2$, $\delta_1$ leads to a
different state than $\delta_2$. 

Having obtained this condition the proof of the main theorem
follows from the following theorem. 
\begin{theorem}
  Let $\Ss$ be a CPS of level $n$ such that the $\varepsilon$-contraction $\Gg$ of its graph is finitely branching 
  and such that for each state $q$ transitions leading to state $q$ are all $\varepsilon$-transitions or are
  all non-$\varepsilon$-transitions.
  Let $c_m$ be some configuration of distance $m$ from the initial
  configuration. 
  
  Let $S_1=(m+1)\cdot C_{\Ss}$ and $S_j=2^{S_{j-1}}$ for
  $2\leq j\leq n$, where $C_{\Ss}=3 \cdot c \cdot 2^c$ with 
  $c=\lvert \Tt_\Ss \rvert +1$. 
  Assume also that in $\Gg$ there exists a path $p$ of length at
  least $S_n$ which starts in $c_m$.
  
  Then there are infinitely many paths in $\Gg$ which start in
  $c_m$ and end in configurations having the same state as the last
  configuration of $p$.  
\end{theorem}
\begin{proof}
  From Definition \ref{def:SequencesMandN} we obtain sequences
  $M_i$ and $N_i$. Note that the
  sequence $S_i$ defined in this lemma and the sequence $S_i$ defined
  in that definition agree. 
  Due to the existence of $p$, there is a run
  $R$ starting in $c_m$ such that $S_n$ transitions in $R$ are not
  labelled by $\varepsilon$ and especially the last transition is not
  labelled $\varepsilon$. 
  Let $G$ be the set of those $0\leq i < \lvert R \rvert$ such that
  the transition between $R(i)$ and $R(i+1)$ is not labelled $\varepsilon$. 
  Since $\Ss$ is of level $n$, for any configuration $c'$ of $\Ss$ the only position of an $n$-stack in $c'$ is $\TOP{n}(c')=(0,0, \dots, 0)$. 
  Especially, every $g\in G$ satisfies 
  $\hist{\subrun{R}{g}{\lvert R\rvert}}{\TOP{n}(R(\lvert R\rvert))}=\TOP{n}(R(g))$.
  Furthermore, we saw in 
  Lemma \ref{lem:Bounded-Stack-Growth} that $M_i$ is an upper bound
  for the size of each 
  $i$-stack in $c_m$  for each $1\leq i\leq m$.
  Thus,  Lemma \ref{lem:CombinatoricSandN} implies that
  $\lvert G \rvert = S_n\geq 3 N_n > \hat N_n$ 
  and we can apply 
  Lemma \ref{lem:Pumpability} to $R$. 
  We obtain numbers $0\leq x<y<z\leq \lvert R \rvert$ such that  
  \begin{enumerate} 
  \item	$\ctype_\Xx(R(x))=\ctype_\Xx(R(y))$,   
  \item	$R_1:=\subrun{R}{x}{y}$ is a pumping run,
    \setcounter{enumi}{3}
  \item $\TOP0(R(x))\neq\TOP0(R(y))$ or
    \begin{align*}
      G\cap\{x,x+1,\dots,y-1\} \neq \emptyset,
    \end{align*}
  \item $z-1\in G$, and 
  \item $R_2:=\subrun{R}{y}{z}$ is a $\TOP{0}$-non-erasing run.  
  \end{enumerate} 
   $G\cap \{x, x+1, \dots, y_1\} =\emptyset$ is equivalent to saying
   that all labels in $R_1$ are $\varepsilon$. Moreover, since $\TOP0(R(x))\neq\TOP0(R(y))$ in this case, we conclude that
   $R_1\in \Pp_{>, \varepsilon}$.
   As $z-1\in G$, $R_{y,z}$ ends by a non-$\varepsilon$-transition. 
   Thus, Proposition \ref{prop:inf-branching} implies that $\Gg$ is
   infinitely branching which contradicts our assumptions.
   
   Thus, $R_1$ contains at least one edge with a label different from
   $\varepsilon$. 
   Due to Proposition \ref{prop:types1}, 
   there are runs $(S_i)_{i\in\N}$ 
   such that
   \begin{itemize}
   \item  $S_i$ starts in $R(x)$,
   \item contains at least $i$ transitions whose
     label is not $\varepsilon$ and
   \item $\ctype_\Xx(R(x)) \typOrd
     \ctype_\Xx(S_i(|S_i|))$.
   \end{itemize}
   Let $T_i$ be the copy of $\subrun{R}{y}{\lvert R \rvert}$ obtained
   by application of 
   Proposition \ref{prop:types3} starting in $S_i(|S_i|)$. 
   Then $U_i:=\subrun{R}{0}{x}\circ S_i \circ T_i$ is a run from 
   $c_m$ to $e_i:=T_i(\lvert T_i \rvert)$ that contains at least $i$ 
   non-$\varepsilon$ labelled edges. Furthermore, the state of $e_i$
   is the final state of $R$. Due to our assumption on the pushdown
   system, this state determines whether the edge to $e_i$ is labelled
   $\varepsilon$. Since the last edge of $R$ is not labelled
   $\varepsilon$, the edge to $e_i$ is not labelled $\varepsilon$,
   whence $e_i$ is a node in $\Gg$. 
   Thus, $U_i$ induces a path of length at least $i$ starting in
   $c_m$ and ending in a configuration with the same state as the
   final configuration of $p$. \qed
 \end{proof}


\section{Collapsible Pushdown Systems as 
  Tree Generators} 
\label{app:Trees}
In this section we describe how collapsible pushdown system can be
used to generate trees and we show that part 2 of Corollary
\ref{cor:infinite} follows from Theorem \ref{thm:pumping} (recall that
the trees of level $n$ recursion schemes are exactly the trees
generated by level $n$ collapsible pushdown systems). 
We consider ranked, potentially infinite trees. 
We fix an alphabet $A$ of tree labels and a function $rank\colon
A\to\mathbb{N}$. Some node of 
a tree  labelled by $a\in A$ has always $rank(a)$ many children. 

We say that a system $\Ss$ generates a tree over alphabet $(A,rank)$
if it satisfies the following (syntactical and semantical)
restrictions. 
\begin{enumerate}
\item The input alphabet of $\Ss$ is $A\cup\{0,1,\dots,m-1\}$, where
  $m=\max\{rank(a):a\in A\}$. 
\item The state set of $\Ss$ can be partitioned into $Q_\varepsilon,
  Q_0,Q_1,\dots,Q_m$ such that the following holds for every stack
  symbol $\gamma$. 
  For each state $q\in Q_\varepsilon$ there is at most one
  transition $(q,\gamma,a,p,op)$, and $a\in
  A\cup\{\varepsilon\}$; 
  $p\in Q_{rank(a)}$ if $a\in A$, and $p\in Q_\varepsilon$ if
  $a=\varepsilon$. 
  For each state $q\in Q_i$ ($0\leq i\leq m$) there are exactly
  $i$ transitions $(q,\gamma,a,p,op)$;  
  for each of them $a$ is a different number from
  $\{0,1,\dots,i-1\}$, and for each of them $p\in
  Q_\varepsilon$. 
  Additionally, the initial state is in $Q_\varepsilon$. 
\item From each configuration of $\Ss$ reachable from the initial one
  and having a state in $Q_\varepsilon$ there exists a run to a
  configuration having a state in 
  $Q\setminus Q_\varepsilon$. 
  From each configuration of $\Ss$ reachable from the initial
  one and having the state in some $Q_i$ ($0\leq i\leq m$), all
  of the $i$ transitions are applicable. 
\end{enumerate}

\begin{definition}
  The \emph{tree generated by a system} $\Ss$ has as nodes runs
  from the initial configuration to a configuration having a
  state in $Q\setminus Q_\varepsilon$. 
  A node $R$ is labelled by $a\in A$ if the last transition of $R$ is
  labelled by $a$. 
  A node $S$ is the $i$-th child ($0\leq i\leq rank(a)-1$) of $R$ if
  $S=R\circ T$ where the first edge of $T$ is labelled by $i$ and it
  is the only edge of $T$ labelled by a number from 
  $\{0,1,\dots,m-1\}$.
\end{definition}

The conditions on $\Ss$ guarantee that the above definition really
defines an $A$-labelled ranked tree. 
Condition 2 says that the system behaves in a deterministic way if the
state is in $Q_\varepsilon$.  
It performs several $\epsilon$-transitions and, finally, a transition
reading a letter $a$ from $A$; this generates a tree node having
letter $a$. 
Immediately after that the state is in $Q_{rank(a)}$, so there are
$rank(a)$ possible transitions; they correspond to the children of the
node just generated. 
Condition 3 guarantees that this construction will never block. 

Now we come to the proof of the second part of Corollary \ref{cor:infinite}. 
Let $A=\{a,b,c\}$, where $rank(a)=2$, $rank(b)=1$, and $rank(c)=0)$. 
For level $n$ consider the tree $T_n$ in which 
\begin{itemize}
\item the rightmost path is labelled by $a$, and
\item the left subtree of the $i$-th  $a$-labelled node is a path consisting
  of $\exp_n(i)$ many $b$-labelled nodes, ending with a $c$-labelled node. 
\end{itemize}
It is known that $T_n$ can be generated by a pushdown system (without
collapse) of level $n+1$.
(cf.~Example 9 in \cite{blumensath-pumping},
where Blumensath provides a  very similar pushdown system). 

Assume that there exists a collapsible pushdown system of level $n$
which generates $T_n$. 
Let $\Ss$ be the system obtained from it by replacing every
$A$-labelled transition by an $\varepsilon$-transition (so we leave
only labels $0$ and $1$; 
we remove $a$, $b$, $c$ for simplicity).
Let $\Gg$ be the $\varepsilon$-contraction of the configuration graph of $\Ss$.
Let $m$ be a number such that $\exp_n(m-1)>\exp_{n-1}((m+1)\cdot
C_{\Ss L})$, where $C_{\Ss L}$ is the constant from Theorem
\ref{thm:pumping} for $L=\{0,1\}^*$. 
Let $c_m$ be the node of $\Gg$ such that the path from the initial
configuration to $c_m$ is labelled by $1^{m-1}0$. 
By definition of $\Ss$ such node exists, and in $\Gg$ there exists a
path $p$ from $c_m$ of length $\exp_n(m-1)-1$ (labelled by zeroes). 
Application of Theorem \ref{thm:pumping} yields arbitrarily long
paths from $c_m$ which contradicts with our assumption about the form
of the tree generated by $\Ss$.


\section{Decidability of Finite Branching} 
In this section we show that types can be used to 
decide whether a given $n$-CPS $\Ss$ generates a configuration graph
whose $\varepsilon$-contraction is finitely branching.
As a consequence we obtain also an algorithm checking whether this
$\varepsilon$-contraction is finite, 
and whether its unfolding into a tree is finite.

Let us remark that the same can be shown in a nontrivial way using
decidability of $\mu$-calculus on configuration graphs of $n$-CPS's,  
and using (multiple times) the reflection of $n$-CPS's with respect to
the $\mu$-calculus (i.e., the result from \cite{reflection}). 
This algorithm (at least its variant which we have in mind) works in
$m$-EXPTIME for some $m=O(n^2)$; 
the reason is that each use of the $\mu$-calculus reflection increases
the size of the system (more or less) $n$-times exponentially, 
and we use it (more or less) $n$ times.

The proof using types is very elegant: first we observe that
Proposition \ref{prop:inf-branching} holds in an ``if and only if''
version:  the
$\varepsilon$-contraction of a configuration graph is infinitely
branching if and only if it contains a pumping run from
$\Pp_{<,\varepsilon}$ that starts and ends in a stack of the same
type. Due to the pumpability of pumping runs, this is the same as
saying that there are arbitrarily large sequences of pumping runs from
$\Pp_{<, \varepsilon}$. Thus, checking for infinite branching is the
same as checking for long sequences of pumping runs. 
The second ingredient of our proof is the fact that families defined
by well-formed rules are closed under composition 
(cf.\ Lemma \ref{lem:compos}). 
Thus, for a well-chosen  family $\Yy$, the function $\ctype_{\Yy}$
yields information about long sequences of pumping runs and we only
have to check whether the initial configuration has a type that
witnesses such a sequence in order to decide infinite branching of the
$\varepsilon$-contraction of a configuration graph.

As previously we may assume that for each state $q$
transitions leading to  state $q$ are all $\varepsilon$-transitions or are all non-$\varepsilon$-transitions.
Let $\Ss$ be such system, and $\Gg$ be the $\varepsilon$-contraction of the configuration graph of $\Ss$.
Let $\Xx$ be the family of sets of runs defined in Section
\ref{sec:Runs}. Recall that it contains the set $\Pp_{<,\varepsilon}$
of pumping runs increasing the stack and using only
$\varepsilon$-transitions, the set $\Nn_0$ of
$\TOP{0}$-non-erasing runs,  
and the set $\Qq$ of all runs. 

We begin by giving a ``if and only if'' version of Proposition \ref{prop:inf-branching}.
\begin{claim} 
  System $\Ss$ is infinitely branching if and only if there exists a run $R$ from the initial configuration, and indices $0\leq x < y \leq \lvert R \rvert$ such that
  \begin{enumerate}
  \item $\ctype_\Xx(R(x))=\ctype_\Xx(R(y))$,
  \item $\subrun{R}{x}{y}$ is a pumping run in
    $\Pp_{<,\varepsilon}$, i.e., a pumping run such that
    $\TOP0(R(x))\lexOrdstrict\TOP0(R(y))$ and all edges of
    $\subrun{R}{x}{y}$ are labelled by $\varepsilon$, and 
  \item $\subrun{R}{y}{\lvert R \rvert}$ is a $\TOP{0}$-non-erasing
    run, i.e., a run in $\Nn_0$, ending with a non-$\varepsilon$-transition.
  \end{enumerate}
\end{claim}

\begin{proof}
  The right-to-left implication is just Proposition \ref{prop:inf-branching}.
  For the opposite direction we inspect the proof of Lemma
  \ref{lem:Bounded-Stack-Growth}.
  Assume that $\Gg$ is infinitely branching.
  Then for some $m$ the thesis of Lemma \ref{lem:Bounded-Stack-Growth}
  is not satisfied: 
  there exists a run $R$ from the initial configuration which
  corresponds to a path of length $m$ in $\Gg$  
  such that for some $k$ the size of some $k$-stack of $R(|R|)$
  is greater than $M_k$ (as otherwise we trivially have finite
  branching). 
  Choose the minimal such $m$. 
  Notice that the proof of Lemma \ref{lem:Bounded-Stack-Growth} goes
  by contradiction: it indeed assumes that such run $R$ exists. 
  As a conclusion on the end of the proof we obtain a run $S$ which
  satisfies assumptions of Proposition \ref{prop:inf-branching}. 
  This is almost what we need, but $S$ does not necessarily
  start in the initial configuration. 
  However for sure it starts in a reachable configuration, so we can
  append at the beginning of $S$ the run from the initial
  configuration to $S(0)$; 
  such run still satisfies the conditions on the right side of our
  claim.	\qed 
\end{proof}

Now consider a family $\Yy$ (described by wf-rules) containing the set of runs
$$\Vv=\Qq\circ\Pp_{<,\varepsilon}\circ\Pp_{<,\varepsilon}\circ\dots\circ
\Pp_{<,\varepsilon}\circ\Nn_0,$$ 
where the number of the $\Pp_{<,\varepsilon}$ factors is a fixed
number greater than
the number of possible values of $\ctype_\Xx$. 
Due to the claim  a system is infinitely branching if and only if there is a run from $\Vv$ starting in the initial configuration and ending by a non-$\varepsilon$-transition.
Indeed, if the system is infinitely branching, we have a run $R$ like
in the claim. 
Because $\ctype_\Xx(R(x))=\ctype_\Xx(R(y))$, we can again produce a pumping run $S_1$ from $R(y)$ such that $\ctype_\Xx(S_1(0))\typOrd\ctype_\Xx(S_1(|S_1|))$ (by Theorem \ref{thm:types}),
and then again a pumping run $S_2$ from $S_1(|S_1|)$.
This way we can produce arbitrarily many pumping runs (as many as
required in $\Vv$); let $S_m$ be the last of them. 
Then we have $\ctype_\Xx(R(y))\typOrd\ctype_\Xx(S_m(|S_m|))$ and by
Proposition \ref{prop:types2} there is a $\TOP0$-non-erasing run from
$S_m(|S_m|)$ ending in the same state as $R$, thus ending by a
non-$\varepsilon$-transition. 
The composition of all these runs is in $\Vv$.
Oppositely, assume that we have a run $R$ in $\Vv$.
Because the number of the factors $\Pp_{<,\varepsilon}$ in the
definition of $\Vv$ is greater than the number of possible values of
$\ctype_\Xx$,  
we can find two indices $x<y$ in $R$ between these factors such that
$\ctype_\Xx(R(x))=\ctype_\Xx(R(y))$. 
Because $\Pp_{<,\varepsilon}\circ\Pp_{<,\varepsilon}\subseteq
\Pp_{<,\varepsilon}$, we see that
$\subrun{R}{x}{y}\in\Pp_{<,\varepsilon}$. 
Similarly, because $\Pp_{<,\varepsilon}\circ\Nn_0\subseteq\Nn_0$, we
see that $\subrun{R}{y}{|R|}\in\Nn_0$. 
Thus, we can use the claim and obtain that $\Gg$ is infinitely branching.

Now the algorithm checking whether the branching is infinite is very easy:
it is enough to check whether from the initial configuration there is
a run in $\Vv$ ending by a non-$\varepsilon$-transition. 
To do that, we compute $\type_\Yy$ of the initial stack (note that the
definition of $\type_\Yy$ can be translated into an algorithm
computing $\type_\Yy$), and we check
whether it contains a triple 
$(q_I,\Vv,q)$, where $q$ is a state such that all transitions leading
to state $q$ are non-$\varepsilon$-transitions. 
Notice that the number of possible values of $\ctype_\Xx$ is $n$-times
exponential in the size of the system. 
Thus, also the size of the family $\Yy$ is $n$-times exponential
(beside of the whole composition $\Vv$ it contains also all shorter
compositions). 
Thus, the number of run descriptors for the family $\Yy$ is $2n$-times
exponential in the size of the system. 
It follows that the algorithm is in $2n$-EXPTIME.

As a corollary we obtain an algorithm checking whether the
$\varepsilon$-contraction of the configuration graph of a given CPS
$\Ss$ is finite. 
In order to decide this, we convert $\Ss$ into another system $\Rr$
such that 
the $\varepsilon$-contraction of the graph of $\Ss$ is finite if and
only if the $\varepsilon$-contraction of the graph of $\Rr$ is
finitely branching. 
We again assume that for each state $q$ of $\Ss$ transitions leading
to state $q$ are all $\varepsilon$-transitions or are all
non-$\varepsilon$-transitions. 
In $\Rr$ we have the same transitions as in $\Ss$, but all labelled by
$\varepsilon$. 
Additionally, we add a new initial state and a transition labelled
different from $\varepsilon$ to the old initial state which preserves
the stack. 
Moreover,
from each state $q$ such that all transitions leading to
state $q$ in $\Ss$ are not $\varepsilon$-transitions, in $\Rr$ we make a
transition to a new state $q_{die}$ labelled by some letter 
(there are no transitions from state $q_{die}$). 
After this conversion, the whole graph of $\Ss$ ``lives'' in the
$\varepsilon$-transitions following the initial configuration of $\Ss$
but every node of the $\varepsilon$-contraction of the graph of $\Ss$
induces an edge from this initial configuration in the
$\varepsilon$-contraction of the graph of $\Rr$. 

Moreover, we also obtain an algorithm checking whether the unfolding
into a tree of the $\varepsilon$-contraction of the configuration
graph of a given CPS $\Ss$ is finite. 
Indeed, a tree is finite if it is finitely branching (which we check
as above), and if it does not contain infinite paths. 
By Theorem \ref{thm:pumping} this tree contains infinite paths if and
only if it contains a path (from the initial configuration) of length
at least $\exp_{n-1}(C_\Ss)$. 
A run containing at least $\exp_{n-1}(C_\Ss)$
non-$\varepsilon$-transitions can be easily defined using wf-rules,
thus we can check whether such run exists from the initial
configuration by calculating the type of the initial 
configuration. 
(Whether the tree contains infinite paths can be also easily expressed
in $\mu$-calculus, hence decided using the $\mu$-calculus
decidability).


\end{document}